\documentclass[12pt]{article}
\usepackage[english]{babel}
\usepackage{dsfont}
\usepackage{graphicx}
\usepackage{amsmath}
\usepackage{amssymb}
\usepackage{amsthm}
\usepackage{natbib}

\allowdisplaybreaks[2]

\usepackage{color}
\definecolor{DarkBlue}{rgb}{0,0.18,0.55}
\usepackage{hyperref}
\hypersetup{pdfauthor={Denis Chetverikov and Daniel Wilhelm},colorlinks=true,citecolor=DarkBlue,filecolor=DarkBlue,linkcolor=DarkBlue,urlcolor=DarkBlue,pdftex}

\theoremstyle{plain}
\newtheorem{theorem}{Theorem}
\newtheorem{lemma}{Lemma}

\newtheorem{corollary}{Corollary}
\newtheorem{definition}{Definition}
\newtheorem{assumption}{Assumption}

\theoremstyle{definition}
\newtheorem{remark}{Remark}
\newtheorem{example}{Example}

\newcommand*{\QEDA}{\hfill\ensuremath{\square}}

\newcommand{\Ep}{{\mathrm{E}}}

\newcommand{\argmin}{\text{argmin}}

\renewcommand{\Pr}{{\mathrm{P}}}
\renewcommand{\hat}{\widehat}
\renewcommand{\tilde }{\widetilde}

\date{}
\begin{document}
\title{Nonparametric Instrumental Variable Estimation Under Monotonicity\footnote{First version: January 2014. This version: \today. We thank Alex Belloni, Richard Blundell, St\'{e}phane Bonhomme, Moshe Buchinsky, Matias Cattaneo, Xiaohong Chen, Victor Chernozhukov, Andrew Chesher, Joachim Freyberger, Jinyong Hahn, Dennis Kristensen, Simon Lee, Zhipeng Liao, Rosa Matzkin, Eric Mbakop, Ulrich M\"{u}ller, Markus Rei{\ss}, Susanne Schennach, Azeem Shaikh, and Vladimir Spokoiny for useful comments and discussions.}}


\author{Denis Chetverikov\thanks{Department of Economics, University of California at Los Angeles, 315 Portola Plaza, Bunche Hall, Los Angeles, CA 90024, USA; E-Mail address: \texttt{chetverikov@econ.ucla.edu}.} \and Daniel Wilhelm\thanks{Department of Economics, University College London, Gower Street, London WC1E 6BT, United Kingdom; E-Mail address: \texttt{d.wilhelm@ucl.ac.uk}. The author gratefully acknowledges financial support from the ESRC Centre for Microdata Methods and Practice at IFS (RES-589-28-0001).}}
\maketitle
\begin{abstract}
	The ill-posedness of the inverse problem of recovering a regression function in a nonparametric instrumental variable model leads to estimators that may suffer from a very slow, logarithmic rate of convergence. In this paper, we show that restricting the problem to models with monotone regression functions and monotone instruments significantly weakens the ill-posedness of the problem. In stark contrast to the existing literature, the presence of a monotone instrument implies boundedness of our measure of ill-posedness when restricted to the space of monotone functions. Based on this result we derive a novel non-asymptotic error bound for the constrained estimator that imposes monotonicity of the regression function. For a given sample size, the bound is independent of the degree of ill-posedness as long as the regression function is not too steep. As an implication, the bound allows us to show that the constrained estimator converges at a fast, polynomial rate, independently of the degree of ill-posedness, in a large, but slowly shrinking neighborhood of constant functions. Our simulation study demonstrates significant finite-sample performance gains from imposing monotonicity even when the regression function is rather far from being a constant. We apply the constrained estimator to the problem of estimating gasoline demand functions from U.S. data.
\end{abstract}




\newpage
\section{Introduction}\label{sec: introduction}
Despite the pervasive use of linear instrumental variable methods in empirical research, their nonparametric counterparts are far from enjoying similar popularity. Perhaps two of the main reasons for this originate from the observation that point-identification of the regression function in the nonparametric instrumental variable (NPIV) model requires completeness assumptions, which have been argued to be strong (\cite{Santos2012}) and non-testable (\cite{canay2013re}), and from the fact that the NPIV model is ill-posed, which may cause regression function estimators in this model to suffer from a very slow, logarithmic rate of convergence (e.g. \cite{blundell2007}).

In this paper, we explore the possibility of imposing shape restrictions to improve statistical properties of the NPIV estimators and to achieve (partial) identification of the NPIV model in the absence of completeness assumptions.
We study the NPIV model
\begin{equation}\label{eq:model}
	Y = g(X) + \varepsilon,\qquad \Ep[\varepsilon | W] =0,
\end{equation}
where $Y$ is a dependent variable, $X$ an endogenous regressor, and $W$ an instrumental variable (IV). We are interested in identification and estimation of the nonparametric regression function $g$ based on a random sample of size $n$ from the distribution of $(Y,X,W)$. We impose two monotonicity conditions: (i) monotonicity of the regression function $g$ (we assume that $g$ is increasing\footnote{All results in the paper hold also when $g$ is decreasing. In fact, as we show in Section~\ref{sec: identification} the sign of the slope of $g$ is identified under our monotonicity conditions.}) and (ii) monotonicity of the reduced form relationship between the endogenous regressor $X$ and the instrument $W$ in the sense that the conditional distribution of $X$ given $W$ corresponding to higher values of $W$ first-order stochastically dominates the same conditional distribution corresponding to lower values of $W$ (the monotone IV assumption). 

We show that these two monotonicity conditions together significantly change the structure of the NPIV model, and weaken its ill-posedness. In particular, we demonstrate that under the second condition, a slightly modified version of the sieve measure of ill-posedness defined in \cite{blundell2007} is bounded uniformly over the dimension of the sieve space, when restricted to the set of monotone functions; see Section \ref{sec: MIVE} for details. As a result, under our two monotonicity conditions, the constrained NPIV estimator that imposes monotonicity of the regression function $g$ possesses a fast rate of convergence in a large but slowly shrinking neighborhood of constant functions. 

More specifically, we derive a new non-asymptotic error bound for the constrained estimator. The bound exhibits two regimes. The first regime applies when the function $g$ is not too steep, and the bound in this regime is independent of the sieve measure of ill-posedness, which slows down the convergence rate of the unconstrained estimator. In fact, under some further conditions, the bound in the first regime takes the following form: with high probability,
$$
\|\hat{g}^c - g\|_{2,t} \leq C\Big(\Big(\frac{K\log n}{n}\Big)^{1/2} + K^{-s}\Big)
$$
where $\hat{g}^c$ is the constrained estimator, $\|\cdot\|_{2,t}$ an appropriate $L^2$-norm, $K$ the number of series terms in the estimator $\hat{g}^c$, $s$ the number of derivatives of the function $g$, and $C$ some constant; see Section \ref{sec: estimation} for details. Thus, the constrained estimator $\hat{g}^c$ has fast rate of convergence in the first regime, and the bound in this regime is of the same order, up to a log-factor, as that for series estimators of conditional mean functions. The second regime applies when the function $g$ is sufficiently steep. In this regime, the bound is similar to that for the unconstrained NPIV estimators. The steepness level separating the two regimes depends on the sample size $n$ and decreases as the sample size $n$ grows large. Therefore, for a given increasing function $g$, if the sample size $n$ is not too large, the bound is in its first regime, where the constrained estimator $\hat{g}^c$ does not suffer from ill-posedness of the model. As the sample size $n$ grows large, however, the bound eventually switches to the second regime, where ill-posedness of the model undermines the statistical properties of the constrained estimator $\hat{g}^c$ similarly to the case of the unconstrained estimator.

Intuitively, existence of the second regime of the bound is well expected. Indeed, if the function $g$ is strictly increasing, it lies in the interior of the constraint that $g$ is increasing. Hence, the constraint does not bind asymptotically so that, in sufficiently large samples, the constrained estimator coincides with the unconstrained one and the two estimators share the same convergence rate. In {\em finite samples}, however, the constraint binds with non-negligible probability even if $g$ is strictly increasing. The first regime of our {\em non-asymptotic} bound captures this finite-sample phenomenon, and improvements from imposing the monotonicity constraint on $g$ in this regime can be understood as a boundary effect. Importantly, and perhaps unexpectedly, we show that under the monotone IV assumption, this boundary effect is so strong that ill-posedness of the problem completely disappears in the first regime.\footnote{Even though we have established the result that ill-posedness disappears in the first regime under the monotone IV assumption, currently we do not know whether this assumption is necessary for the result.} In addition, we demonstrate via our analytical results as well as simulations that this boundary effect can be strong even far away from the boundary and/or in large samples.


Our simulation experiments confirm these theoretical findings and demonstrate dramatic finite-sample performance improvements of the constrained relative to the unconstrained NPIV estimator when the monotone IV assumption is satisfied. Imposing the monotonicity constraint on $g$ removes the estimator's non-monotone oscillations due to sampling noise, which in ill-posed inverse problems can be particularly pronounced. Therefore, imposing the monotonicity constraint significantly reduces variance while only slightly increasing bias.


In addition, we show that in the absence of completeness assumptions, that is, when the NPIV model is not point-identified, our monotonicity conditions have non-trivial identification power, and can provide partial identification of the model.

We regard both monotonicity conditions as natural in many economic applications. In fact, both of these conditions often directly follow from economic theory. Consider the following generic example. Suppose an agent chooses input $X$ (e.g. schooling) to produce an outcome $Y$ (e.g. life-time earnings) such that $Y=g(X) + \varepsilon$, where $\varepsilon$ summarizes determinants of outcome other than $X$. The cost of choosing a level $X=x$ is $C(x,W,\eta)$, where $W$ is a cost-shifter (e.g. distance to college) and $\eta$ represents (possibly vector-valued) unobserved heterogeneity in costs (e.g. family background, a family's taste for education, variation in local infrastructure). The agent's optimization problem can then be written as
$$ X = \arg\max_x \left\{ g(x)+\varepsilon - c(x,W,\eta) \right\} $$
so that, from the first-order condition of this optimization problem,
\begin{equation}\label{eq: dXdW}
	\frac{\partial X}{\partial W} = \frac{\frac{\partial^2 c}{\partial X \partial W} }{\frac{\partial^2 g}{\partial X^2} - \frac{\partial^2 c}{\partial X^2}} \geq 0
\end{equation}
if marginal cost are decreasing in $W$ (i.e. $\partial^2 c/\partial X \partial W \leq 0$), marginal cost are increasing in $X$ (i.e. $\partial^2 c/\partial X^2 > 0$), and the production function is concave (i.e. $\partial^2 g/\partial X^2 \leq 0$). As long as $W$ is independent of the pair $(\varepsilon,\eta)$, condition \eqref{eq: dXdW} implies our monotone IV assumption and $g$ increasing corresponds to the assumption of a monotone regression function. Dependence between $\eta$ and $\varepsilon$ generates endogeneity of $X$, and independence of $W$ from $(\varepsilon,\eta)$ implies that $W$ can be used as an instrument for $X$.

Another example is the estimation of Engel curves. In this case, the outcome variable $Y$ is the budget share of a good, the endogenous variable $X$ is total expenditure, and the instrument $W$ is gross income. Our monotonicity conditions are plausible in this example because for normal goods such as food-in, the budget share is decreasing in total expenditure, and total expenditure increases with gross income. Finally, consider the estimation of (Marshallian) demand curves. The outcome variable $Y$ is quantity of a consumed good, the endogenous variable $X$ is the price of the good, and $W$ could be some variable that shifts production cost of the good. For a normal good, the Slutsky inequality predicts $Y$ to be decreasing in price $X$ as long as income effects are not too large. Furthermore, price is increasing in production cost and, thus, increasing in the instrument $W$, and so our monotonicity conditions are plausible in this example as well.

Both of our monotonicity assumptions are testable. For example, a test of the monotone IV condition can be found in \cite{lee2009fd}. In this paper, we extend their results by deriving an {\em adaptive} test of the monotone IV condition, with the value of the involved smoothness parameter chosen in a data-driven fashion. This adaptation procedure allows us to construct a test with desirable power properties when the degree of smoothness of the conditional distribution of $X$ given $W$ is unknown. Regarding our first monotonicity condition, to the best of our knowledge, there are no procedures in the literature that consistently test monotonicity of the function $g$ in the NPIV model (\ref{eq:model}). We consider such procedures in a separate project and, in this paper, propose a simple test of monotonicity of $g$ given that the monotone IV condition holds.

\cite{Matzkin:1994kx} advocates the use of shape restrictions in econometrics and argues that economic theory often provides restrictions on functions of interest, such as monotonicity, concavity, and/or Slutsky symmetry. In the context of the NPIV model \eqref{eq:model}, \cite{Freyberger:2013fk} show that, in the absence of point-identification, shape restrictions may yield informative bounds on functionals of $g$ and develop inference procedures when the regressor $X$ and the instrument $W$ are discrete. \cite{Blundell:2013fk} demonstrate via simulations that imposing Slutsky inequalities in a quantile NPIV model for gasoline demand improves finite-sample properties of the NPIV estimator. \cite{grasmair2013tr} study the problem of demand estimation imposing various constraints implied by economic theory, such as Slutsky inequalities, and derive the convergence rate of a constrained NPIV estimator under an abstract projected source condition. Our results are different from theirs because we focus on non-asymptotic error bounds, with special emphasis on properties of our estimator in {\em the neighborhood} of the boundary, we derive our results under easily interpretable, low level conditions, and we find that our estimator does not suffer from ill-posedness of the problem in a large but slowly shrinking neighborhood of constant functions.

\paragraph{Other related literature.}

The NPIV model has received substantial attention in the recent econometrics literature. \cite{Newey:2003p2167}, \cite{Hall:2005p2135}, \cite{blundell2007}, and \cite{darolles2011} study identification of the NPIV model \eqref{eq:model} and propose estimators of the regression function $g$. See \cite{horowitz2011, H14} for recent surveys and further references.
In the mildly ill-posed case, \cite{Hall:2005p2135} derive the minimax risk lower bound in $L^2$-norm and show that their estimator achieves this lower bound. Under different conditions, \cite{chen2011er} derive a similar bound for the mildly and the severely ill-posed case and show that the estimator by \cite{blundell2007} achieves this bound. \cite{Chen:2013fk} establish minimax risk bounds in the sup-norm, again both for the mildly and the severely ill-posed case. The optimal convergence rates in the severely ill-posed case were shown to be logarithmic, which means that the slow convergence rate of existing estimators is not a deficiency of those estimators but rather an intrinsic feature of the statistical inverse problem.

There is also large statistics literature on nonparametric estimation of monotone functions when the regressor is exogenous, i.e. $W=X$, so that $g$ is a conditional mean function. This literature can be traced back at least to \cite{brunk1955gf}. Surveys of this literature and further references can be found in \cite{yatchew1998gf}, \cite{delecroix2000ee}, and \cite{gijbels2004gf}. For the case in which the regression function is both smooth and monotone, many different ways of imposing monotonicity on the estimator have been studied; see, for example, \cite{mukerjee1988}, \cite{Cheng:1981fd}, \cite{wright1981}, \cite{friedman1984gf}, \cite{ramsay1988fd}, \cite{mammen1991fd}, \cite{ramsay1998fg}, \cite{mammen1999rr}, \cite{hall2001fd}, \cite{mammen2001ff}, and \cite{dette2006gg}. Importantly, under the mild assumption that the estimators consistently estimate the derivative of the regression function, the standard unconstrained nonparametric regression estimators are known to be monotone with probability approaching one when the regression function is strictly increasing. Therefore, such estimators have the same rate of convergence as the corresponding constrained estimators that impose monotonicity (\cite{mammen1991fd}). As a consequence, gains from imposing a monotonicity constraint can only be expected when the regression function is close to the boundary of the constraint and/or in finite samples. \cite{Zhang:2002dq} and \cite{Chatterjee:2013eu} formalize this intuition by deriving risk bounds of the isotonic (monotone) regression estimators and showing that these bounds imply fast convergence rates when the regression function has flat parts. Our results are different from theirs because we focus on the endogenous case with $W\neq X$ and study the impact of monotonicity constraints on the ill-posedness property of the NPIV model which is absent in the standard regression problem.

\paragraph{Notation.}

For a differentiable function $f:\mathbb{R}\to\mathbb{R}$, we use $Df(x)$ to denote its derivative.  When a function $f$ has several arguments, we use $D$ with an index to denote the derivative of $f$ with respect to corresponding argument; for example, $D_w f(w,u)$ denotes the partial derivative of $f$ with respect to $w$.  
For random variables $A$ and $B$, we denote by $f_{A,B}(a,b)$, $f_{A|B}(a,b)$, and $f_{A}(a)$ the joint, conditional and marginal densities of $(A,B)$, $A$ given $B$, and $A$, respectively. Similarly, we let $F_{A,B}(a,b)$, $F_{A|B}(a,b)$, and $F_{A}(a)$ refer to the corresponding cumulative distribution functions. For an operator $T:L^2[0,1]\to L^2[0,1]$, we let $\|T\|_2$ denote the operator norm defined as 
$$
\|T\|_2 = \sup_{h\in L^2[0,1]:\,\|h\|_2=1} \|Th\|_2.
$$ Finally, by increasing and decreasing we mean that a function is non-decreasing and non-increasing, respectively.

\paragraph{Outline.}

The remainder of the paper is organized as follows. In the next section, we analyze ill-posedness of the model (\ref{eq:model}) under our monotonicity conditions and derive a useful bound on a restricted measure of ill-posedness for the model (\ref{eq:model}). Section~\ref{sec: estimation} discusses the implications of our monotonicity assumptions for estimation of the regression function $g$. In particular, we show that the rate of convergence of our estimator is always not worse than that of unconstrained estimators but may be much faster in a large, but slowly shrinking, neighborhood of constant functions. Section~\ref{sec: identification} shows that our monotonicity conditions have non-trivial identification power. Section~\ref{sec: testing} provides new tests of our two monotonicity assumptions. In Section~\ref{sec: simulations}, we present results of a Monte Carlo simulation study that demonstrates large gains in performance of the constrained estimator relative to the unconstrained one. Finally, Section~\ref{sec: gasoline demand} applies the constrained estimator to the problem of estimating gasoline demand functions. All proofs are collected in the appendix.

\section{Boundedness of the Measure of Ill-posedness under Monotonicity}\label{sec: MIVE}

In this section, we discuss the sense in which the ill-posedness of the NPIV model \eqref{eq:model} is weakened by imposing our monotonicity conditions. In particular, we introduce a restricted measure of ill-posedness for this model (see equation \eqref{eq: tau definition}) and show that, in stark contrast to the existing literature, our measure is bounded (Corollary \ref{co:bound_tau}) when the monotone IV condition holds. 

The NPIV model requires solving the equation $\Ep[Y|W] = \Ep[g(X)|W]$ for the function $g$. Letting $T:L^2[0,1]\to L^2[0,1]$ be the linear operator defined by $(Th)(w):=\Ep[h(X) | W=w] f_W(w)$ and denoting $m(w) := \Ep[Y|W=w]f_W(w)$, we can express this equation as
\begin{equation}\label{eq: functional equation}
T g = m.
\end{equation}
In finite-dimensional regressions, the operator $T$ corresponds to a finite-dimensional matrix whose singular values are typically assumed to be nonzero (rank condition). Therefore, the solution $g$ is continuous in $m$, and consistent estimation of $m$ at a fast convergence rate leads to consistent estimation of $g$ at the same fast convergence rate. In infinite-dimensional models, however, $T$ is an operator that, under weak conditions, possesses infinitely many singular values that tend to zero. Therefore, small perturbations in $m$ may lead to large perturbations in $g$. This discontinuity renders equation (\ref{eq: functional equation}) ill-posed and introduces challenges in estimation of the NPIV model (\ref{eq:model}) that are not present in parametric regressions nor in nonparametric regressions with exogenous regressors; see \cite{horowitz2011, H14} for a more detailed discussion.

In this section, we show that, under our monotonicity conditions, there exists a finite constant $\bar{C}$ such that for any monotone function $g'$ and any constant function $g''$, with $m'=T g'$ and $m''=T g''$, we have
$$ \|g'-g''\|_{2,t} \leq \bar{C} \|m'-m''\|_2,$$
where $\|\cdot\|_{2,t}$ is a truncated $L^2$-norm defined below. This result plays a central role in our derivation of the upper bound on the restricted measure of ill-posedness, of identification bounds, and of fast convergence rates of a constrained NPIV estimator that imposes monotonicity of $g$ in a large but slowly shrinking neighborhood of constant functions.

We now introduce our assumptions. Let $0\leq x_1<\tilde{x}_1<\tilde{x}_2<x_2\leq 1$ and $0\leq w_1<w_2\leq 1$ be some constants. We implicitly assume that $x_1$, $\tilde{x}_1$, and $w_1$ are close to $0$ whereas $x_2$, $\tilde{x}_2$, and $w_2$ are close to $1$. Our first assumption is the monotone IV condition that requires a monotone relationship between the endogenous regressor $X$ and the instrument $W$.

\begin{assumption}[Monotone IV]\label{as: monotone iv}
	For all $x,w',w''\in(0,1)$,
	\begin{equation}
		w'\leq w''\qquad \Rightarrow \qquad F_{X|W}(x|w') \geq F_{X|W}(x|w'') \label{eq: fosd}.
	\end{equation}
	Furthermore, there exists a constant $C_F>1$ such that
	\begin{equation}\label{eq: lower cdf bound}
		F_{X|W}(x|w_1) \geq C_F F_{X|W}(x|w_2),\qquad \forall x\in(0,x_2)
	\end{equation}
	and
	\begin{equation}\label{eq: upper cdf bound}
		C_F(1-F_{X|W}(x|w_1)) \leq 1- F_{X|W}(x|w_2),\qquad \forall x\in(x_1,1)
	\end{equation}
\end{assumption}

Assumption~\ref{as: monotone iv} is crucial for our analysis. The first part, condition \eqref{eq: fosd}, requires first-order stochastic dominance of the conditional distribution of the endogenous regressor $X$ given the instrument $W$ as we increase the value of the instrument $W$. This condition \eqref{eq: fosd} is testable; see, for example, \cite{lee2009fd}. In Section \ref{sec: testing} below, we extend the results of \cite{lee2009fd} by providing an {\em adaptive} test of the first-order stochastic dominance condition (\ref{eq: fosd}).

The second and third parts of Assumption~\ref{as: monotone iv}, conditions \eqref{eq: lower cdf bound} and \eqref{eq: upper cdf bound}, strengthen the stochastic dominance condition \eqref{eq: fosd} in the sense that the conditional distribution is required to ``shift to the right'' by a {\em strictly} positive amount at least between two values of the instrument, $w_1$ and $w_2$, so that the instrument is not redundant. Conditions \eqref{eq: lower cdf bound} and \eqref{eq: upper cdf bound} are rather weak as they require such a shift only in some intervals $(0,x_2)$ and $(x_1,1)$, respectively.

Condition \eqref{eq: fosd} can be equivalently stated in terms of monotonicity with respect to the instrument $W$ of the reduced form first stage function. Indeed, by the Skorohod representation, it is always possible to construct a random variable $U$ distributed uniformly on $[0,1]$ such that $U$ is independent of $W$, and equation $X=r(W,U)$ holds for the reduced form first stage function $r(w,u):=F^{-1}_{X|W}(u|w):=\inf\{x: F_{X|W}(x|w)\geq u\}$. Therefore, condition (\ref{eq: fosd}) is equivalent to the assumption that the function $w\mapsto r(w,u)$ is increasing for all $u\in[0,1]$.
Notice, however, that our condition \eqref{eq: fosd} allows for general unobserved heterogeneity of dimension larger than one, for instance as in Example~\ref{ex: uniforms} below. 

Condition \eqref{eq: fosd} is related to a corresponding condition in \cite{Kasy2014tr} who assumes that the (structural) first stage has the form $X=\tilde{r}(W,\tilde{U})$ where $\tilde{U}$, representing (potentially multidimensional) unobserved heterogeneity, is independent of $W$, and the function $w\mapsto \tilde{r}(w,\tilde{u})$ is increasing for all values $\tilde{u}$. Kasy employs his condition for identification of (nonseparable) triangular systems with multidimensional unobserved heterogeneity whereas we use our condition \eqref{eq: fosd} to derive a useful bound on the restricted measure of ill-posedness and to obtain a fast rate of convergence of a monotone NPIV estimator of $g$ in the (separable) model \eqref{eq:model}. 
Condition \eqref{eq: fosd} is not related to the monotone IV assumption in the influential work by \cite{manski2000} which requires the function $w\mapsto \Ep[\varepsilon|W=w]$ to be increasing. Instead, we maintain the mean independence condition $\Ep[\varepsilon|W]=0$.

\begin{assumption}[Density]\label{as: density}
	(i) The joint distribution of the pair $(X,W)$ is absolutely continuous with respect to the Lebesgue measure on $[0,1]^2$ with the density $f_{X,W}(x,w)$ satisfying $\int_0^1\int_0^1 f_{X,W}(x,w)^2d x d w\leq C_T$ for some finite constant $C_T$. (ii) There exists a constant $c_f>0$ such that $f_{X|W}(x|w) \geq c_f$ for all $x\in [x_1,x_2]$ and $w\in\{w_1,w_2\}$. (iii) There exists constants $0<c_W\leq C_W<\infty$ such that $c_W\leq f_{W}(w) \leq C_W$ for all $w\in[0, 1]$.
\end{assumption}

This is a mild regularity assumption. The first part of the assumption implies that the operator $T$ is compact. The second and the third parts of the assumption require the conditional distribution of $X$ given $W=w_1$ or $w_2$ and the marginal distribution of $W$ to be bounded away from zero over some intervals. Recall that we have $0\leq x_1<x_2\leq 1$ and $0\leq w_1<w_2\leq 1$. We could simply set $[x_1,x_2]=[w_1,w_2]=[0,1]$ in the second part of the assumption but having $0<x_1<x_2<1$ and $0<w_1<w_2<1$ is required to allow for densities such as the normal, which, even after a transformation to the interval $[0,1]$, may not yield a conditional density $f_{X|W}(x|w)$ bounded away from zero; see Example~\ref{ex: normal density} below. Therefore, we allow for the general case $0\leq x_1<x_2\leq 1$ and $0\leq w_1<w_2\leq 1$. The restriction $f_W(w)\leq C_W$ for all $w\in[0,1]$ imposed in Assumption \ref{as: density} is not actually required for the results in this section, but rather those of Section~\ref{sec: estimation}.

We now provide two examples of distributions of $(X,W)$ that satisfy Assumptions~\ref{as: monotone iv} and \ref{as: density}, and show two possible ways in which the instrument $W$ can shift the conditional distribution of $X$ given $W$. Figure~\ref{fig:cdfXW} displays the corresponding conditional distributions.

\begin{example}[Normal density]\label{ex: normal density}
 Let $(\tilde{X},\tilde{W})$ be jointly normal with mean zero, variance one, and correlation $0<\rho<1$. Let $\Phi(u)$ denote the distribution function of a $N(0,1)$ random variable. Define $X=\Phi(\tilde{X})$ and $W=\Phi(\tilde{W})$. Since $\tilde{X}=\rho\tilde{W}+(1-\rho^2)^{1/2}U$ for some standard normal random variable $U$ that is independent of $\tilde{W}$, we have $$X=\Phi(\rho\Phi^{-1}(W)+(1-\rho^2)^{1/2}U)$$ where $U$ is independent of $W$. Therefore, the pair $(X,W)$ satisfies condition \eqref{eq: fosd} of our monotone IV Assumption \ref{as: monotone iv}. Lemma~\ref{lem: normal density} in the appendix verifies that the remaining conditions of Assumption~\ref{as: monotone iv} as well as Assumption \ref{as: density} are also satisfied.\QEDA
\end{example} 

\begin{example}[Two-dimensional unobserved heterogeneity]\label{ex: uniforms}
Let $X=U_1 + U_2 W$, where $U_1, U_2, W$ are mutually independent, $U_1,U_2\sim U[0,1/2]$ and $W\sim U[0,1]$. Since $U_2$ is positive, it is straightforward to see that the stochastic dominance condition \eqref{eq: fosd} is satisfied. Lemma~\ref{lem: uniforms} in the appendix shows that the remaining conditions of Assumption~\ref{as: monotone iv} as well as Assumption \ref{as: density} are also satisfied.\QEDA
\end{example}

Figure~\ref{fig:cdfXW} shows that, in Example~\ref{ex: normal density}, the conditional distribution at two different values of the instrument is shifted to the right at every value of $X$, whereas, in Example~\ref{ex: uniforms}, the conditional support of $X$ given $W=w$ changes with $w$, but the positive shift in the cdf of $X | W=w$ occurs only for values of $X$ in a subinterval of $[0,1]$.

Before stating our results in this section, we introduce some additional notation. Define the truncated $L^2$-norm $\|\cdot\|_{2,t}$ by
$$
\|h\|_{2,t} := \left( \int_{\tilde{x}_1}^{\tilde{x}_2} h(x)^2 d x \right)^{1/2},\quad h\in L^2[0,1].
$$
Also, let $\mathcal{M}$ denote the set of all monotone functions in $L^2[0,1]$. Finally, define $\zeta:=(c_f, c_W, C_F, C_T,  w_1,  w_2, x_1, x_2, \tilde{x}_1, \tilde{x}_2)$. Below is our first main result in this section.

\begin{theorem}[Lower Bound on $T$]\label{thm: main result 1}
	Let Assumptions \ref{as: monotone iv} and \ref{as: density} be satisfied. Then there exists a finite constant $\bar{C}$ depending only on $\zeta$ such that
	\begin{equation}\label{eq: bounded inverse}
	\|h\|_{2,t}\leq \bar{C}\|T h\|_2
	\end{equation}
 for any function $h\in \mathcal{M}$.
\end{theorem}

To prove this theorem, we take a function $h\in\mathcal{M}$ with $\|h\|_{2,t}=1$ and show that $\|Th\|_2$ is bounded away from zero. A key observation that allows us to establish this bound is that, under monotone IV Assumption~\ref{as: monotone iv}, the function $w\mapsto \Ep[h(X)|W=w]$ is monotone whenever $h$ is. Together with non-redundancy  of the instrument $W$ implied by conditions (\ref{eq: lower cdf bound}) and (\ref{eq: upper cdf bound}) of Assumption \ref{as: monotone iv}, this allows us to show that $\Ep[h(X)|W=w_1]$ and $\Ep[h(X)|W=w_2]$ cannot both be close to zero so that $\|\Ep[h(X)|W=\cdot]\|_2$ is bounded from below by a strictly positive constant from the values of $\Ep[h(X)|W=w]$ in the neighborhood of either $w_1$ or $w_2$. By Assumption~\ref{as: density}, $\|T h\|_2$ must then also be bounded away from zero.


Theorem~\ref{thm: main result 1} has an important consequence. Indeed, consider the linear equation \eqref{eq: functional equation}. By Assumption~\ref{as: density}(i), the operator $T$ is compact, and so
\begin{equation}\label{eq: ill posed inverse problem}
	\frac{\|h_k\|_2}{\|T h_k\|_2} \to\infty\text{ as }k\to\infty\text{ for some sequence }\{h_k,k\geq 1\}\subset L^2[0,1].
\end{equation}
Property~\eqref{eq: ill posed inverse problem} means that $\|T h\|_2$ being small does not necessarily imply that $\|h\|_{2}$ is small and, therefore, the inverse of the operator $T: L^2[0,1]\to L^2[0,1]$, when it exists, cannot be continuous. Therefore, \eqref{eq: functional equation} is ill-posed in Hadamard's sense\footnote{Well- and ill-posedness in Hadamard's sense are defined as follows. Let $A:D\to R$ be a continuous mapping between metric spaces $(D,\rho_D)$ and $(R,\rho_R)$. Then, for $d\in D$ and $r\in R$, the equation $Ad=r$ is called ``well-posed'' on $D$ in Hadamard's sense (see \cite{Hadamard:1923ty}) if (i) $A$ is bijective and (ii) $A^{-1}:R\to D$ is continuous, so that for each $r\in R$ there exists a unique $d=A^{-1}r\in D$ satisfying $Ad=r$, and, moreover, the solution $d=A^{-1}r$ is continous in ``the data'' $r$. Otherwise, the equation is called ``ill-posed'' in Hadamard's sense.}, if no other conditions are imposed. This is the main reason why standard NPIV estimators have (potentially very) slow rate of convergence.
Theorem~\ref{thm: main result 1}, on the other hand, implies that, under Assumptions~\ref{as: monotone iv} and \ref{as: density}, \eqref{eq: ill posed inverse problem} is not possible if $h_k$ belongs to the set $\mathcal{M}$ of monotone functions in $L^2[0,1]$ for all $k\geq 1$ and we replace the $L^2$-norm $\|\cdot\|_2$ in the numerator of the left-hand side of \eqref{eq: ill posed inverse problem} by the truncated $L^2$-norm $\|\cdot\|_{2,t}$, indicating that shape restrictions may be helpful to improve statistical properties of the NPIV estimators. Also, in Remark~\ref{rem: ill-posedness with truncated norm}, we show that replacing the norm in the numerator is not a significant modification in the sense that for most ill-posed problems, and in particular for all severely ill-posed problems, (\ref{eq: ill posed inverse problem}) holds even if we replace $L^2$-norm $\|\cdot\|_2$ in the numerator of the left-hand side of (\ref{eq: ill posed inverse problem}) by the truncated $L^2$-norm $\|\cdot\|_{2,t}$.

Next, we derive an implication of Theorem~\ref{thm: main result 1} for the (quantitative) measure of ill-posedness of the model (\ref{eq:model}). We first define the restricted measure of ill-posedness. For $a\in \mathbb{R}$, let
$$
\mathcal{H}(a) := \left\{ h\in L^2[0,1] :\, \inf_{0\leq x'<x''\leq 1}\frac{h(x'')-h(x')}{x''-x'} \geq -a \right\}
$$
be the space containing all functions in $L^2[0,1]$ with lower derivative bounded from below by $-a$ uniformly over the interval $[0,1]$.
Note that $\mathcal{H}(a')\subset \mathcal{H}(a'')$ whenever $a'\leq a''$ and that $\mathcal{H}(0)$ is the set of increasing functions in $L^2[0,1]$. For continuously differentiable functions, $h\in L^2[0,1]$ belongs to $\mathcal{H}(a)$ if and only if $\inf_{x\in[0,1]}Dh(x)\geq -a$.
%
%
Further, define the {\em restricted measure of ill-posedness}:
\begin{equation}\label{eq: tau definition}
\tau(a) := \sup_{h\in \mathcal{H}(a)\atop \|h\|_{2,t}=1} \frac{\|h\|_{2,t}}{\|Th\|_2}.
\end{equation}
As we discussed above, under our Assumptions \ref{as: monotone iv} and \ref{as: density}, $\tau(\infty)=\infty$ if we use the $L^2$-norm instead of the truncated $L^2$-norm in the numerator in (\ref{eq: tau definition}). We show in Remark~\ref{rem: ill-posedness with truncated norm} below, that $\tau(\infty)=\infty$ for many ill-posed and, in particular, for all severely ill-posed problems even with the truncated $L^2$-norm as defined in \eqref{eq: tau definition}. However, Theorem~\ref{thm: main result 1} implies that $\tau(0)$ is bounded from above by $\bar{C}$ and, by definition, $\tau(a)$ is increasing in $a$, i.e. $\tau(a')\leq \tau(a'')$ for $a'\leq a''$. It turns out that $\tau(a)$ is bounded from above even for some positive values of $a$:
\begin{corollary}[Bound for the Restricted Measure of Ill-Posedness]\label{co:bound_tau}
	Let Assumptions~\ref{as: monotone iv} and \ref{as: density} be satisfied. Then there exist constants $c_\tau>0$ and $0<C_\tau<\infty$ depending only on $\zeta$ such that
			\begin{equation}\label{eq: main corollary tau}
				\tau(a) \leq C_\tau
			\end{equation}
 for all $a\leq c_\tau$.
\end{corollary}
This is our second main result in this section. It is exactly this corollary of Theorem \ref{thm: main result 1} that allows us to obtain a fast convergence rate of our constrained NPIV estimator $\hat{g}^c$ not only when the regression function $g$ is constant but, more generally, when $g$ belongs to a large but slowly shrinking neighborhood of constant functions.


\begin{remark}[Ill-posedness is preserved by norm truncation]\label{rem: ill-posedness with truncated norm}
	Under Assumptions~\ref{as: monotone iv} and \ref{as: density}, the integral operator $T$ satisfies \eqref{eq: ill posed inverse problem}. Here we demonstrate that, in many cases, and in particular in all severely ill-posed cases, (\ref{eq: ill posed inverse problem}) continues to hold if we replace the $L^2$-norm $\|\cdot\|_2$ by the truncated $L^2$-norm $\|\cdot\|_{2,t}$ in the numerator of the left-hand side of (\ref{eq: ill posed inverse problem}), that is, there exists a sequence $\{l_{k},k\geq 1\}$ in $L^2[0,1]$ such that
	\begin{equation}\label{eq: ill posed inverse problem restricted}
	\frac{\|l_{k}\|_{2,t}}{\|Tl_{k}\|_2}\to\infty\text{ as }k\to\infty.
	\end{equation}
	Indeed, under Assumptions \ref{as: monotone iv} and \ref{as: density}, $T$ is compact, and so the spectral theorem implies that there exists a spectral decomposition of operator $T$, $\{(h_j,\varphi_j),j\geq 1\}$, where $\{h_j,j\geq 1\}$ is an orthonormal basis of $L^2[0,1]$ and $\{\varphi_j,j\geq 1\}$ is a decreasing sequence of positive numbers such that $\varphi_j\to 0$ as $j\to\infty$, and $\|Th_j\|_2=\varphi_j\|h_j\|_2=\varphi_j$. Also, Lemma \ref{lem: degenerate basis} in the appendix shows that if $\{h_j,j\geq 1\}$ is an orthonormal basis in $L^2[0,1]$, then for any $\alpha>0$, $\|h_j\|_{2,t}>j^{-1/2-\alpha}$ for infinitely many $j$, and so there exists a subsequence $\{h_{j_k},k\geq 1\}$ such that $\|h_{j_k}\|_{2,t}> {j_k}^{-1/2-\alpha}$. Therefore, under a weak condition that $j^{1/2+\alpha}\varphi_j\to 0$ as $j\to\infty$, using $\|h_{j_k}\|_2=1$ for all $k\geq 1$, we conclude that for the subsequence $l_k=h_{j_k}$,
	$$
	\frac{\|l_k\|_{2,t}}{\|Tl_k\|_2}\geq \frac{\|h_{j_k}\|_2}{{j_k}^{1/2+\alpha}\|Th_{j_k}\|_2}=\frac{1}{{j_k}^{1/2+\alpha}\varphi_{j_k}}\to\infty\text{ as }k\to\infty
	$$
	leading to (\ref{eq: ill posed inverse problem restricted}). Note also that the condition that $j^{1/2+\alpha}\varphi_j\to 0$ as $j\to\infty$ necessarily holds if there exists a constant $c>0$ such that $\varphi_j\leq e^{-c j}$ for all large $j$, that is, if the problem is severely ill-posed.
	 Thus, under our Assumptions \ref{as: monotone iv} and \ref{as: density}, the restriction in Theorem \ref{thm: main result 1} that $h$ belongs to the space $\mathcal{M}$ of \textit{monotone} functions in $L^2[0,1]$  plays a crucial role for the result (\ref{eq: bounded inverse}) to hold. On the other hand, whether the result (\ref{eq: bounded inverse}) can be obtained for all $h\in\mathcal{M}$ without imposing our monotone IV Assumption~\ref{as: monotone iv} appears to be an open (and interesting) question. \QEDA
\end{remark}

\begin{remark}[Severe ill-posedness is preserved by norm truncation]\label{rem: severe ill-posedness}
One might wonder whether our monotone IV Assumption \ref{as: monotone iv} excludes all severely ill-posed problems, and whether the norm truncation significantly changes these problems. Here we show that there do exist severely ill-posed problems that satisfy our monotone IV Assumption \ref{as: monotone iv}, and also that severely ill-posed problems remain severely ill-posed even if we replace the $L^2$-norm $\|\cdot\|_{2}$ by the truncated $L^2$-norm $\|\cdot\|_{2,t}$. Indeed, consider Example~\ref{ex: normal density} above. Because, in this example, the pair $(X,W)$ is a transformation of the normal distribution, it is well known that the integral operator $T$ in this example has singular values decreasing exponentially fast. More specifically, the spectral decomposition $\{(h_k,\varphi_k),k\geq 1\}$ of the operator $T$ satisfies $\varphi_k=\rho^k$ for all $k$ and some $\rho<1$. Hence,
	$$
	\frac{\|h_k\|_2}{\|Th_k\|_2}=\left(\frac{1}{\rho}\right)^k.
	$$
Since $(1/\rho)^k\to \infty$ as $k\to\infty$ exponentially fast, this example leads to a severely ill-posed problem. Moreover, by Lemma \ref{lem: degenerate basis},  for any $\alpha>0$ and $\rho'\in (\rho,1)$,
	$$
	\frac{\|h_k\|_{2,t}}{\|Th_k\|_2}>\frac{1}{k^{1/2+\alpha}} \left(\frac{1}{\rho}\right)^k\geq \left(\frac{1}{\rho'}\right)^k
	$$
for infinitely many $k$. Thus, replacing the $L^2$ norm $\|\cdot\|_2$ by the truncated $L^2$ norm $\|\cdot\|_{2,t}$ preserves the severe ill-posedness of the problem. However, it follows from Theorem \ref{thm: main result 1} that uniformly over all $h\in\mathcal{M}$,
	$
	\|h\|_{2,t}/\|Th\|_2\leq \bar{C}.
	$
	Therefore, in this example, as well as in all other severely ill-posed problems satisfying Assumptions \ref{as: monotone iv} and \ref{as: density}, imposing monotonicity on the function $h\in L^2[0,1]$ significantly changes the properties of the ratio $\|h\|_{2,t}/\|Th\|_2$.\QEDA
\end{remark}

\begin{remark}[Monotone IV Assumption does not imply control function approach]
Our monotone IV Assumption \ref{as: monotone iv} does not imply the applicability of a control function approach to estimation of the function $g$. Consider Example~\ref{ex: uniforms} above. In this example, the relationship between $X$ and $W$ has a two-dimensional vector $(U_1,U_2)$ of unobserved heterogeneity. Therefore, by Proposition~4 of \cite{kasy2011fd}, there does not exist any control function $C:[0,1]^2\to \mathbb{R}$ such that (i) $C$ is invertible in its second argument, and (ii) $X$ is independent of $\varepsilon$ conditional on $V=C(X,W)$. As a consequence, our monotone IV Assumption \ref{as: monotone iv} does not imply any of the existing control function conditions such as those in \cite{newey1999re} and \cite{Imbens:2002p4861}, for example.\footnote{It is easy to show that the existence of a control function does not imply our monotone IV condition either, so our and the control function approach rely on conditions that are non-nested.} Since multidimensional unobserved heterogeneity is common in economic applications (see \cite{imbens2007re} and \cite{Kasy2014tr}), we view our approach to avoiding ill-posedness as complementary to the control function approach.\QEDA
\end{remark}

\begin{remark}[On the role of norm truncation]
	Let us also briefly comment on the role of the truncated norm $\|\cdot\|_{2,t}$ in (\ref{eq: bounded inverse}). 	There are two reasons why we need the truncated $L^2$-norm $\|\cdot\|_{2,t}$ rather than the usual $L^2$-norm $\|\cdot\|_2$. First, Lemma~\ref{lem: crucial lemma} in the appendix shows that, under Assumptions~\ref{as: monotone iv} and \ref{as: density}, there exists a constant $0<C_2 <\infty$ such that
	$$ \|h\|_1 \leq C_2 \|Th\|_1 $$
	for any increasing and continuously differentiable function $h\in L^1[0,1]$. This result does not require any truncation of the norms and implies boundedness of a measure of illposedness defined in terms of $L^1[0,1]$-norms: $\sup_{h\in L^1[0,1], h\, \text{increasing}} \|h\|_1 / \|Th\|_1$. To extend this result to $L^2[0,1]$-norms we need to introduce a positive, but arbitrarily small, amount of truncation at the boundaries, so that we have a control $\|h\|_{2,t}\leq C\|h\|_1$ for some constant $C$ and all monotone functions $h\in\mathcal{M}$. Second, we want to allow for the normal density as in Example~\ref{ex: normal density}, which violates condition (ii) of Assumption~\ref{as: density} if we set $[x_1,x_2]=[0,1]$. \QEDA
\end{remark}

\begin{remark}[Bounds on the measure of ill-posedness via compactness]
Another approach to obtain a result like (\ref{eq: bounded inverse}) would be to employ compactness arguments. For example, let $b>0$ be some (potentially large) constant and consider the class of functions $\mathcal{M}(b)$ consisting of all functions $h$ in $\mathcal{M}$ such that $\|h\|_\infty = \sup_{x\in[0,1]} |h(x)|\leq b$. It is well known that the set $\mathcal{M}(b)$ is compact under the $L^2$-norm $\|\cdot\|_2$, and so, as long as $T$ is invertible, there exists some $C>0$ such that $\|h\|_2 \leq C\|T h\|_2$ for all $h\in\mathcal{M}(b)$ since (i) $T$ is continuous and (ii) any continuous function achieves its minimum on a compact set. This bound does not require the monotone IV assumption and also does not require replacing the $L^2$-norm by the truncated $L^2$-norm. Further, defining $\widetilde{\tau}(a,b): = \sup_{h\in\mathcal{H}(a): \|h\|_\infty\leq b, \|h\|_2 =1}\|h\|_2/\|T h\|_2$ for all $a>0$ and using the same arguments as those in the proof of Corollary \ref{co:bound_tau}, one can show that there exist some finite constants $c,C>0$ such that $\widetilde{\tau}(a,b)\leq C$ for all $a\leq c$. This (seemingly interesting) result, however, is not useful for bounding the estimation error of an estimator of $g$ because, as the proof of Theorem \ref{thm: risk bounds} in the next section reveals, obtaining meaningful bounds would require a result of the form $\widetilde{\tau}(a,b_n)\leq C$ for all $a\leq c$ for some sequence $\{b_n, n\geq 1\}$ such that $b_n\to\infty$, even if we know that $\sup_{x\in[0,1]}|g(x)|\leq b$ and we impose this constraint on the estimator of $g$. In contrast, our arguments in Theorem \ref{thm: main result 1}, being fundamentally different, do lead to meaningful bounds on the estimation error of the constrained estimator $\hat{g}^c$ of $g$. 
\QEDA
\end{remark}

\section{Non-asymptotic Risk Bounds Under Monotonicity}
\label{sec: estimation}


The rate at which unconstrained NPIV estimators converge to $g$ depends crucially on the so-called sieve measure of ill-posedness, which, unlike $\tau(a)$, does not measure ill-posedness over the space $\mathcal{H}(a)$, but rather over the space $\mathcal{H}_n(\infty)$, a finite-dimensional (sieve) approximation to $\mathcal{H}(\infty)$. In particular, the convergence rate is slower the faster the sieve measure of ill-posedness grows with the dimensionality of the sieve space $\mathcal{H}_n(\infty)$. The convergence rates can be as slow as logarithmic in the severely ill-posed case. Since by Corollary~\ref{co:bound_tau}, our monotonicity assumptions imply boundedness of $\tau(a)$ for some range of finite values $a$, we expect these assumptions to translate into favorable performance of a constrained estimator that imposes monotonicity of $g$. This intuition is confirmed by the novel non-asymptotic error bounds we derive in this section (Theorem \ref{thm: risk bounds}).

Let $(Y_i,X_i,W_i)$, $i=1,\dots,n$, be an i.i.d. sample from the distribution of $(Y,X,W)$.
To define our estimator, we first introduce some notation. Let $\{p_k(x),k\geq 1\}$ and $\{q_k(w),k\geq 1\}$ be two orthonormal bases in $L^2[0,1]$. For $K=K_n\geq 1$ and $J=J_n\geq K_n$, denote
\begin{align*}
	&p(x):=(p_1(x),\dots,p_K(x))'\text{ and }q(w):=(q_1(w),\dots,q_J(w))'.
\end{align*}
Let $\mathbf{P} := (p(X_1),\ldots, p(X_n))'$ and $\mathbf{Q} := (q(W_1),\ldots, q(W_n))'$. Similarly, stack all observations on $Y$ in $\mathbf{Y}:= (Y_1,\ldots,Y_n)'$. 
%
Let $\mathcal{H}_n(a)$ be a sequence of finite-dimensional spaces defined by
$$\mathcal{H}_n(a) := \left\{ h\in \mathcal{H}(a) :\, \exists b_1,\ldots, b_{K_n}\in\mathbb{R} \text{ with } h=\sum_{j=1}^{K_n} b_j p_j \right\}$$
which become dense in $\mathcal{H}(a)$ as $n\to\infty$.
Throughout the paper, we assume that $\|g\|_2\leq C_b$ where $C_b$ is a large but finite constant known by the researcher.
We define two estimators of $g$: the unconstrained estimator $\hat{g}^u(x) := p(x)'\hat{\beta}^u$ with
\begin{equation}\label{eq:ghat}
	\hat{\beta}^u:=\argmin_{b\in\mathbb{R}^{K}:\|b\|\leq C_b} (\mathbf{Y}-\mathbf{P} b)'\mathbf{Q}(\mathbf{Q}'\mathbf{Q})^{-1}\mathbf{Q}'(\mathbf{Y}-\mathbf{P} b)
\end{equation}
which is similar to the estimator defined in \cite{horowitz2012} and a special case of the estimator considered in \cite{blundell2007}, and the constrained estimator $\hat{g}^c(x):=p(x)'\hat{\beta}^c$ with
\begin{equation}\label{eq:ghatm}
	\hat{\beta}^c:=\argmin_{b\in\mathbb{R}^{K}:\, p(\cdot)'b \in \mathcal{H}_{n}(0),\|b\|\leq C_b} (\mathbf{Y}-\mathbf{P} b)'\mathbf{Q}(\mathbf{Q}'\mathbf{Q})^{-1}\mathbf{Q}'(\mathbf{Y}-\mathbf{P} b),
\end{equation}
which imposes the monotonicity of $g$ through the constraint $p(\cdot)'b \in \mathcal{H}_{n}(0)$. 

To study properties of the two estimators we introduce a finite-dimensional, or sieve, counterpart of the restricted measure of ill-posedness $\tau(a)$ defined in \eqref{eq: tau definition} and also recall the definition of the (unrestricted) sieve measure of ill-posedness. Specifically, define the {\em restricted} and {\em unrestricted} sieve measures of ill-posedness $\tau_{n,t}(a)$ and $\tau_n$ as
$$\tau_{n,t}(a) := \sup_{h\in \mathcal{H}_n(a)\atop \|h\|_{2,t}=1} \frac{\|h\|_{2,t}}{\|Th\|_2}\quad\text{ and }\quad \tau_n:=\sup_{h\in\mathcal{H}_n(\infty)}\frac{\|h\|_{2}}{\|T h\|_2}.$$
The sieve measure of ill-posedness defined in \cite{blundell2007} and also used, for example, in \cite{horowitz2012} is $\tau_n$. \cite{blundell2007} show that $\tau_{n}$ is related to the singular values of $T$.\footnote{In fact, \cite{blundell2007} talk about the eigenvalues of $T^*T$, where $T^*$ is the adjoint of $T$ but there is a one-to-one relationship between eigenvalues of $T^*T$ and singular values of $T$.} If the singular values converge to zero at the rate $K^{-r}$ as $K\to\infty$, then, under certain conditions, $\tau_n$ diverges at a polynomial rate, that is $\tau_{n} = O(K_n^r)$. This case is typically referred to as ``mildly ill-posed''. On the other hand, when the singular values decrease at a fast exponential rate, then $\tau_{n} = O(e^{cK_n})$, for some constant $c>0$. This case is typically referred to as ``severely ill-posed''.

Our restricted sieve measure of ill-posedness $\tau_{n,t}(a)$ is smaller than the unrestricted sieve measure of ill-posedness $\tau_n$ because we replace the $L^2$-norm in the numerator by the truncated $L^2$-norm and the space $\mathcal{H}_n(\infty)$ by $\mathcal{H}_n(a)$. As explained in Remark~\ref{rem: ill-posedness with truncated norm}, replacing the $L^2$-norm by the truncated $L^2$-norm does not make a crucial difference but, as follows from Corollary \ref{co:bound_tau}, replacing $\mathcal{H}_n(\infty)$ by $\mathcal{H}_n(a)$ does. In particular, since $\tau(a)\leq C_\tau$ for all $a\leq c_\tau$ by Corollary \ref{co:bound_tau}, we also have $\tau_{n,t}(a)\leq C_\tau$ for all $a\leq c_\tau$ because $\tau_{n,t}(a)\leq \tau(a)$. Thus, for all values of $a$ that are not too large, $\tau_{n,t}(a)$ remains bounded uniformly over all $n$, no matter how fast the singular values of $T$ converge to zero.



We now specify conditions that we need to derive non-asymptotic error bounds for the constrained estimator $\hat{g}^c(x)$.
\begin{assumption}[Monotone regression function]\label{as: monotone regression}
The function $g$ is monotone increasing. 
\end{assumption}
\begin{assumption}[Moments]\label{as: moments}
For some constant $C_B<\infty$, (i) $\Ep[\varepsilon^2|W]\leq C_B$ and (ii) $\Ep[g(X)^2|W]\leq C_B$.
\end{assumption}
\begin{assumption}[Relation between $J$ and $K$]\label{as: J-K}
For some constant $C_J<\infty$, $J\leq C_J K$.
\end{assumption}
Assumption \ref{as: monotone regression}, along with Assumption \ref{as: monotone iv}, is our main monotonicity condition. Assumption \ref{as: moments} is a mild moment condition. Assumption \ref{as: J-K} requires that the dimension of the vector $q(w)$ is not much larger than the dimension of the vector $p(x)$. Let $s>0$ be some constant.
\begin{assumption}[Approximation of $g$]\label{as: g approximation}
	There exist $\beta_n\in\mathbb{R}^K$ and a constant $C_g<\infty$ such that the function $g_n(x):=p(x)'\beta_n$, defined for all $x\in[0,1]$, satisfies (i) $g_n\in \mathcal{H}_n(0)$, (ii) $\|g-g_n\|_2\leq C_g K^{-s}$, and (iii) $\|T(g-g_n)\|_2\leq C_g\tau_n^{-1}K^{-s}$.
\end{assumption}
The first part of this condition requires the approximating function $g_n$ to be increasing. The second part requires a particular bound on the approximation error in the $L^2$-norm. \cite{devore1977gg,devore1977ok} show that the assumption $\|g-g_n\|_2\leq C_g K^{-s}$ holds when the approximating basis $p_1,\dots,p_K$ consists of polynomial or spline functions and $g$ belongs to a H\"{o}lder class with smoothness level $s$. Therefore, approximation by monotone functions is similar to approximation by all functions. The third part of this condition is similar to Assumption~6 in \cite{blundell2007}.
\begin{assumption}[Approximation of $m$]\label{as: m approximation}
	There exist $\gamma_n\in\mathbb{R}^J$ and a constant $C_m<\infty$ such that the function $m_n(w):=q(w)'\gamma_n$, defined for all $w\in[0,1]$, satisfies $\|m-m_n\|_2\leq C_m \tau_n^{-1}J^{-s}$ .
\end{assumption}
This condition is similar to Assumption 3(iii) in \cite{horowitz2012}. Also, define the operator $T_n: L^2[0,1]\to L^2[0,1]$ by
$$
(T_n h)(w):=q(w)'\Ep[q(W)p(X)']\Ep[p(U) h(U)],\qquad w\in[0,1]
$$
where $U\sim U[0,1]$.
\begin{assumption}[Operator $T$]\label{as: T approximation}
(i) The operator $T$ is injective and (ii) for some constant $C_a<\infty$, $\|(T-T_n)h\|_2\leq C_a\tau_n^{-1}K^{-s}\|h\|_2$ for all $h\in\mathcal{H}_n(\infty)$.
\end{assumption}
This condition is similar to Assumption~5 in \cite{horowitz2012}.
Finally, let
\begin{align*}
	&\xi_{K,p}:=\sup_{x\in[0,1]}\|p(x)\|,\quad \xi_{J,q}:=\sup_{w\in[0,1]}\|q(w)\|,\quad \xi_n:=\max(\xi_{K,p},\xi_{J,q}).
\end{align*}

We start our analysis in this section with a simple observation that, if the function $g$ is strictly increasing and the sample size $n$ is sufficiently large, then the constrained estimator $\hat{g}^c$ coincides with the unconstrained estimator $\hat{g}^u$, and the two estimators share the same rate of convergence.
\begin{lemma}[Asymptotic equivalence of constrained and unconstrained estimators]\label{lem: asymptotic problem}
Let Assumptions \ref{as: monotone iv}-\ref{as: T approximation} be satisfied. In addition, assume that $g$ is continuously differentiable and $D g(x)\geq c_g$ for all $x\in[0,1]$ and some constant $c_g>0$. If $\tau_n^2\xi_n^2\log n/n\to 0$, $\sup_{x\in[0,1]}\|D p(x)\|(\tau_n(K/n)^{1/2}+K^{-s})\to 0$, and $\sup_{x\in[0,1]}|D g(x) - D g_n(x)|\to 0$ as $n\to\infty$, then
\begin{equation}\label{eq: asymptotic coincidence}
\Pr\Big(\widehat{g}^c(x)=\widehat{g}^u(x)\text{ for all }x\in[0,1]\Big)\to 1\text{ as }n\to\infty.
\end{equation}
\end{lemma}
The result in Lemma \ref{lem: asymptotic problem} is similar to that in Theorem 1 of \cite{mammen1991fd}, which shows equivalence (in the sense of \eqref{eq: asymptotic coincidence}) of the constrained and unconstrained estimators of conditional mean functions. Lemma~\ref{lem: asymptotic problem} implies that imposing monotonicity of $g$ cannot lead to improvements in the rate of convergence of the estimator if $g$ is strictly increasing. However, the result in Lemma~\ref{lem: asymptotic problem} is asymptotic and only applies to the interior of the monotonicity constraint. It does not rule out faster convergence rates on or near the boundary of the monotonicity constraint nor does it rule out significant performance gains in finite samples. In fact, our Monte Carlo simulation study in Section~\ref{sec: simulations} shows significant finite-sample performance improvements from imposing monotonicity even if $g$ is strictly increasing and relatively far from the boundary of the constraint. Therefore, we next derive a {\em non-asymptotic} estimation error bound for the constrained estimator $\hat{g}^c$ and study the impact of the monotonicity constraint on this bound.

\begin{theorem}[Non-asymptotic error bound for the constrained estimator]\label{thm: risk bounds}
Let Assumptions \ref{as: monotone iv}-\ref{as: T approximation} be satisfied, and let $\delta\geq 0$ be some constant. Assume that $\xi_n^2\log n/n\leq c$ for sufficiently small $c>0$. Then with probability at least $1-\alpha - n^{-1}$, we have
\begin{equation}\label{eq: main rate 2}
\|\hat{g}^c-g\|_{2,t}\leq C \Big\{\delta + \tau_{n,t}\Big(\frac{\|D g_n\|_{\infty}}{\delta}\Big)\Big(\frac{K}{\alpha n}+\frac{\xi_n^2\log n}{n}\Big)^{1/2} + K^{-s}\Big\}
\end{equation}
 and
\begin{equation}\label{eq: main rate 3}
\|\hat{g}^c-g\|_{2,t}\leq C\min \Big\{\|D g\|_{\infty} + \Big(\frac{K}{\alpha n}+\frac{\xi_n^2\log n}{n}\Big)^{1/2}, \tau_n \Big(\frac{K}{\alpha n} + \frac{\xi_n^2\log n}{n}\Big)^{1/2}\Big\} +C K^{-s}.
\end{equation}
Here the constants $c,C<\infty$ can be chosen to depend only on the constants appearing in Assumptions \ref{as: monotone iv}-\ref{as: T approximation}.
\end{theorem}
This is the main result of this section. An important feature of this result is that since the constant $C$ depends only on the constants appearing in Assumptions \ref{as: monotone iv}-\ref{as: T approximation}, the bounds \eqref{eq: main rate 2} and \eqref{eq: main rate 3} hold uniformly over all data-generating processes that satisfy those assumptions with the same constants. In particular, for any two data-generating processes in this set, the same finite-sample bounds \eqref{eq: main rate 2} and \eqref{eq: main rate 3} hold with the same constant $C$, even though the unrestricted sieve measure of ill-posedness $\tau_n$ may be of different order of magnitude for these two data-generating processes.

Another important feature of the bound (\ref{eq: main rate 2}) is that it depends on the restricted sieve measure of ill-posedness that we know to be smaller than the unrestricted sieve measure of ill-posedness, appearing in the analysis of the unconstrained estimator. In particular, we know from Section \ref{sec: MIVE} that $\tau_{n,t}(a)\leq \tau(a)$ and that, by Corollary \ref{co:bound_tau}, $\tau(a)$ is uniformly bounded if $a$ is not too large. Employing this result, we obtain the bound (\ref{eq: main rate 3}) of Theorem \ref{thm: risk bounds}.\footnote{Ideally, it would be of great interest to have a tight bound on the restricted sieve measure of ill-posedness $\tau_{n,t}(a)$ for all $a\geq 0$, so that it would be possible to optimize (\ref{eq: main rate 2}) over $\delta$. Results of this form, however, are not yet available in the literature, and so the optimization is not possible.}

The bound \eqref{eq: main rate 3} has two regimes depending on whether the following inequality
\begin{equation}\label{eq: two regimes}
\|D g\|_{\infty} \leq (\tau_n - 1) \Big(\frac{K}{\alpha n} + \frac{\xi_n^2\log n}{n}\Big)^{1/2}
\end{equation}
holds. The most interesting feature of this bound is that in the first regime, when the inequality \eqref{eq: two regimes} is satisfied, the bound is independent of the (unrestricted) sieve measure of ill-posedness $\tau_n$, and can be small if the function $g$ is not too steep, regardless of whether the original NPIV model (\ref{eq:model}) is mildly or severely ill-posed. This is the regime in which the bound relies upon the monotonicity constraint imposed on the estimator $\hat{g}^c$. For a given sample size $n$, this regime is active if the function $g$ is not too steep. 

As the sample size $n$ grows large, the right-hand side of inequality \eqref{eq: two regimes} decreases (if $K=K_n$ grows slowly enough) and eventually becomes smaller than the left-hand side, and the bound \eqref{eq: main rate 3} switches to its second regime, in which it depends on the (unrestricted) sieve measure of ill-posedness $\tau_n$. This is the regime in which the bound does not employ the monotonicity constraint imposed on $\hat{g}^c$. However, since $\tau_n\to \infty$, potentially at a very fast rate, even for relatively large sample sizes $n$ and/or relatively steep functions $g$, the bound may be in its first regime, where the monotonicity constraint is important. The presence of the first regime and the observation that it is active in a (potentially very) large set of data generated processes provides a theoretical justification for the importance of imposing the monotonicity constraint on the estimators of the function $g$ in the NPIV model (\ref{eq:model}) when the monotone IV Assumption \ref{as: monotone iv} is satisfied.

A corollary of the existence of the first regime in the bound \eqref{eq: main rate 3} is that the constrained estimator $\hat{g}^c$ possesses a very fast rate of convergence in a large but slowly shrinking neighborhood of constant functions, independent of the (unrestricted) sieve measure of ill-posedness $\tau_n$:

\begin{corollary}[Fast convergence rate of the constrained estimator under local-to-constant asymptotics]\label{cor: fast rate under local-to-constant}
Consider the triangular array asymptotics where the data generating process, including the function $g$, is allowed to vary with $n$. Let Assumptions \ref{as: monotone iv}-\ref{as: T approximation} be satisfied with the same constants for all $n$. In addition, assume that $\xi_n^2\leq C_\xi K$ for some $0<C_\xi<\infty$ and $K\log n/n\to 0$. If $\sup_{x\in[0,1]}Dg(x)=O((K\log n/n)^{1/2})$, then
\begin{equation}\label{eq: fast rate main}
\|\hat{g}^c-g\|_{2,t} = O_p((K\log n/n)^{1/2} + K^{-s}).
\end{equation}
In particular, if $\sup_{x\in[0,1]}D g(x)=O(n^{-s/(1+2s)}\sqrt{\log n})$ and $K=K_n=C_Kn^{1/(1+2s)}$ for some $0<C_K<\infty$, then
$$
\|\hat{g}^c-g \|_{2,t} = O_p(n^{-s/(1+2s)}\sqrt{\log n}).
$$
\end{corollary}
\begin{remark}[On the condition $\xi_n^2\leq C_\xi K$]
The condition $\xi_n^2\leq C_\xi K$, for $0<C_\xi<\infty$, is satisfied if the sequences $\{p_k(x),k\geq 1\}$ and $\{q_k(w),k\geq 1\}$ consist of commonly used bases such as Fourier, spline, wavelet, or local polynomial partition series; see \cite{Belloni:2014hl} for details.
\QEDA
\end{remark}
The local-to-constant asymptotics considered in this corollary captures the finite sample situation in which the regression function is not too steep relative to the sample size.
The convergence rate in this corollary is the standard polynomial rate of nonparametric conditional mean regression estimators up to a $(\log n)^{1/2}$ factor, regardless of whether the original NPIV problem without our monotonicity assumptions is mildly or severely ill-posed. One way to interpret this result is that the constrained estimator $\hat{g}^c$ is able to recover regression functions in the shrinking neighborhood of constant functions at a fast polynomial rate. Notice that the neighborhood of functions $g$ that satisfy $\sup_{x\in[0,1]}Dg(x)=O((K\log n/n)^{1/2})$ is shrinking at a slow rate because $K\to\infty$, in particular the rate is much slower than $n^{-1/2}$. Therefore, in finite samples, we expect the estimator to perform well for a wide range of (non-constant) regression functions $g$ as long as the maximum slope of $g$ is not too large relative to the sample size.

\begin{remark}[The convergence rate of $\hat{g}^c$ is not slower than that of $\hat{g}^u$]\label{rem: standard rate}
If we replace the condition $\xi_n^2\log n/n\leq c$ in Theorem \ref{thm: risk bounds} by a more restrictive condition $\tau_n^2\xi_n^2\log n/n\leq c$, then in addition to the bounds (\ref{eq: main rate 2}) and (\ref{eq: main rate 3}), it is possible to show that with probability at least $1-\alpha - n^{-1}$, we have
$$
\|\hat{g}^c-g\|_2\leq C(\tau_n (K/(\alpha n))^{1/2} + K^{-s}).
$$
This implies that the constrained estimator $\hat{g}^c$ satisfies
$
\|\hat{g}^c - g\|_2= O_p(\tau_n(K/n)^{1/2} + K^{-s}),
$
which is the standard minimax optimal rate of convergence established for the unconstrained estimator $\hat{g}^u$ in \cite{blundell2007}.\QEDA
\end{remark}

In conclusion, in general, the convergence rate of the constrained estimator is the same as the standard minimax optimal rate, which depends on the degree of ill-posedness and may, in the worst-case, be logarithmic. This case occurs in the interior of the monotonicity constraint when $g$ is strictly monotone. On the other hand, under the monotone IV assumption, the constrained estimator converges at a very fast rate, independently of the degree of ill-posedness, in a large but slowly shrinking neighborhood of constant functions, a part of the boundary of the monotonicity constraint. In finite samples, we expect to experience cases between the two extremes, and the bounds \eqref{eq: main rate 2} and \eqref{eq: main rate 3} provide information on what the performace of the constrained estimator depends in that general case. Since the first regime of  bound (\ref{eq: main rate 3}) is active in a large set of data generating processes and sample size combinations, and since the fast convergence rate in Corollary~\ref{cor: fast rate under local-to-constant} is obtained in a large but slowly shrinking neighborhood of constant functions, we expect the boundary effect due to the monotonicity constraint to be strong even far away from the boundary and for relatively large sample sizes.

\begin{remark}[Average Partial Effects]
	We expect similar results to Theorem~\ref{thm: risk bounds} and Corollary~\ref{cor: fast rate under local-to-constant} to hold in the estimation of linear functionals of $g$, such as average marginal effects. In the unconstrained problem, estimators of linear functionals do not necessarily converge at polynomial rates and may exhibit similarly slow, logarithmic rates as for estimation of the function $g$ itself (e.g. \cite{ECT:9621058}). Therefore, imposing monotonicity as we do in this paper may also improve statistical properties of estimators of such functionals. While we view this as a very important extension of our work, we develop this direction in a separate paper.\QEDA
\end{remark}

\begin{remark}[On the role of the monotonicity constraint]
	Imposing the monotonicity constraint in the NPIV estimation procedure reduces variance by removing non-monotone oscillations in the estimator that are due to sampling noise. Such oscillations are a common feature of unconstrained estimators in ill-posed inverse problems and lead to large variance of such estimators. The reason for this phemonon can be seen in the convergence rate of unconstrained estimators,\footnote{see, for example, \cite{blundell2007}} $\tau_n (K/n)^{1/2} + K^{-s}$, in which the variance term $(K/n)^{1/2}$ is blown up by the multiplication by the measure of ill-posedness $\tau_n$. Because of this relatively large variance of NPIV estimators we expect the unconstrained estimator to possess non-monotonicities even in large samples and even if $g$ is far away from constant functions. Therefore, imposing monotonicity of $g$ can have significant impact on the estimator's performance even in those cases.\QEDA
\end{remark}

\begin{remark}[On robustness of the constrained estimator, I]
	Implementation of the estimators $\hat{g}^c$ and $\hat{g}^u$ requires selecting the number of series terms $K=K_n$ and $J=J_n$. This is a difficult problem because the measure of ill-posedness $\tau_n = \tau(K_n)$, appearing in the convergence rate of both estimators, depends on $K=K_n$ and can blow up quickly as we increase $K$. Therefore, setting $K$ higher than the optimal value may result in a severe deterioration of the statistical properties of $\hat{g}^u$. The problem is alleviated, however, in the case of the constrained estimator $\hat{g}^c$ because $\hat{g}^c$ satisfies the bound (\ref{eq: main rate 3}) of Theorem \ref{thm: risk bounds}, which is independent of $\tau_n$ for sufficiently large $K$. In this sense, the constrained estimator $\hat{g}^c$ possesses some robustness against setting $K$ too high. \QEDA
\end{remark}

\begin{remark}[On robustness of the constrained estimator, II]
	Notice that the fast convergence rates in the local-to-constant asymptotics derived in this section are obtained under two monotonicity conditions, Assumptions~\ref{as: monotone iv} and \ref{as: monotone regression}, but the estimator imposes only the monotonicity of the regression function, not that of the instrument. Therefore, our proposed constrained estimator consistently estimates the regression function $g$ even when the monotone IV assumption is violated.\QEDA
\end{remark}

\begin{remark}[On alternative estimation procedures]
	In the local-to-constant asymptotic framework where $\sup_{x\in[0,1]}Dg(x)=O((K\log n/n)^{1/2})$, the rate of convergence in \eqref{eq: fast rate main} can also be obtained by simply fitting a constant. However, such an estimator, unlike our constrained estimator, is not consistent when the regression function $g$ does not drift towards a constant. Alternatively, one can consider a sequential approach to estimating $g$, namely one can first test whether the function $g$ is constant, and then either fit the constant or apply the unconstrained estimator $\hat{g}^u$ depending on the result of the test. However, it seems difficult to tune such a test to match the performance of the constrained estimator $\hat{g}^c$ studied in this paper. \QEDA
\end{remark}

\begin{remark}[Estimating partially flat functions]\label{rem: partially flat functions}
Since the inversion of the operator $T$ is a global inversion in the sense that the resulting estimators $\hat{g}^c(x)$ and $\hat{g}^u(x)$ depend not only on the shape of $g(x)$ locally at $x$, but on the shape of $g$ over the whole domain, we do not expect convergence rate improvements from imposing monotonicity when the function $g$ is partially flat. However, we leave the question about potential improvements from imposing monotonicity in this case for future research.
\QEDA
\end{remark}

\begin{remark}[Computational aspects]
	The implementation of the constrained estimator in \eqref{eq:ghatm} is particularly simple when the basis vector $p(x)$ consists of polynomials or B-splines of order $2$. In that case, $Dp(x)$ is linear in $x$ and, therefore, the constraint $Dp(x)'b \geq 0$ for all $x\in[0,1]$ needs to be imposed only at the knots or endpoints of $[0,1]$, respectively. The estimator $\hat{\beta}^c$ thus minimizes a quadratic objective function subject to a (finite-dimensional) linear inequality constraint. When the order of the polynomials or B-splines in $p(x)$ is larger than $2$, imposing the monotonicity constraint is slightly more complicated, but it can still be transformed into a finite-dimensional constraint using a representation of non-negative polynomials as a sum of squared polynomials:\footnote{We thank A. Belloni for pointing out this possibility.} one can represent any non-negative polynomial $f:\mathbb{R}\to\mathbb{R}$ as a sum of squares of polynomials (see the survey by \cite{Reznick:2000ve}, for example), i.e. $f(x) = \tilde{p}(x)' M \tilde{p}(x)$ where $\tilde{p}(x)$ is the vector of monomials up to some order and $M$ a matrix of coefficients. Letting $f(x) = Dp(x)'b$, our monotonicity constraint $f(x)\geq 0$ can then be written as $\tilde{p}(x)' M \tilde{p}(x)\geq 0$ for some matrix $M$ that depends on $b$. This condition is equivalent to requiring the matrix $M$ to be positive semi-definite. $\hat{\beta}^c$ thus minimizes a quadratic objective function subject to a (finite-dimensional) semi-definiteness constraint.

	For polynomials defined not over whole $\mathbb{R}$ but only over a compact sub-interval of $\mathbb{R}$, one can use the same reasoning as above together with a result attributed to M. Fekete (see \cite{Powers:2000kl}, for example): for any polynomial $f(x)$ with $f(x)\geq 0$ for $x\in[-1,1]$, there are polynomials $f_1(x)$ and $f_2(x)$, non-negative over whole $\mathbb{R}$, such that $f(x) = f_1(x) + (1-x^2)f_2(x)$. Letting again $f(x)=Dp(x)'b$, one can therefore impose our monotonicity constraint by imposing the positive semi-definiteness of the coefficients in the sums-of-squares representation of $f_1(x)$ and $f_2(x)$.\QEDA
\end{remark}

\begin{remark}[Penalization and shape constraints]\label{rem: penalization}
	Recall that the estimators $\hat{g}^u$ and $\hat{g}^c$ require setting the constraint $\|b\|\leq C_b$ in the optimization problems (\ref{eq:ghat}) and (\ref{eq:ghatm}). In practice, this constraint, or similar constraints in terms of Sobolev norms, which also impose bounds on derivatives of $g$, are typically not enforced in the implementation of an NPIV estimator. \cite{horowitz2012} and \cite{Horowitz2012as}, for example, observe that imposing the constraint does not seem to have an effect in their simulations. On the other hand, especially when one includes many series terms in the computation of the estimator, \cite{blundell2007} and \cite{gagliardini2012fd}, for example, argue that penalizing the norm of $g$ and of its derivatives may stabilize the estimator by reducing its variance. In this sense, penalizing the norm of $g$ and of its derivatives may have a similar effect as imposing monotonicity. However, there are at least two important differences between penalization and imposing monotonicity. First, penalization increases bias of the estimators. In fact, especially in severely ill-posed problems, even small amount of penalization may lead to large bias. In contrast, the monotonicity constraint on the estimator does not increase bias much when the function $g$ itself satisfies the monotonicity constraint. Second, penalization requires the choice of a tuning parameter that governs the strength of penalization, which is a difficult statistical problem. In contrast, imposing monotonicity does not require such choices and can often be motivated directly from economic theory.\QEDA
\end{remark}

\section{Identification Bounds under Monotonicity}\label{sec: identification}

In the previous section, we derived non-asymptotic error bounds on the constrained estimator in the NPIV model \eqref{eq:model} assuming that $g$ is point-identified, or equivalently, that the linear operator $T$ is invertible. 
 \cite{Newey:2003p2167} linked point-identification of $g$ to completeness of the conditional distribution of $X$ given $W$, but this completeness condition has been argued to be strong (\cite{Santos2012}) and non-testable (\cite{canay2013re}). In this section, we therefore discard the completeness condition and explore the identification power of our monotonicity conditions, which appear natural in many economic applications. Specifically, we derive informative bounds on the identified set of functions $g$ satisfying \eqref{eq:model}. This means that, under our two monotonicity assumptions, the identified set is a proper subset of all monotone functions $g\in \mathcal{M}$.

%
%

By a slight abuse of notation, we define the sign of the slope of a differentiable, monotone function $f\in\mathcal{M}$ by
$$sign(Df) := \left\{ \begin{array}{cc} 1,& Df(x)\geq 0 \, \forall x\in [0,1]\text{ and } Df(x)>0\text{ for some } x\in[0,1]\\ 0,& Df(x)=0 \, \forall x\in [0,1]\\ -1,&  Df(x)\leq 0 \, \forall x\in [0,1]\text{ and } Df(x)<0\text{ for some } x\in[0,1]\end{array} \right. $$
and the sign of a scalar $b$ by $sign(b) := 1\{b>0 \} - 1\{b<0 \}$.
We first show that if the function $g$ is monotone, the sign of its slope is identified under our monotone IV assumption (and some other technical conditions):

\begin{theorem}[Identification of the sign of the slope]\label{lem: id of sign of Dg}
	Suppose Assumptions~\ref{as: monotone iv} and \ref{as: density} hold and $f_{X,W}(x,w)>0$ for all $(x,w)\in(0,1)^2$. If $g$ is monotone and continuously differentiable, then $sign(Dg)$ is identified.
\end{theorem}

This theorem shows that, under certain regularity conditions, the monotone IV assumption and monotonicity of the regression function $g$ imply identification of the sign of the regression function's slope, even though the regression function itself is, in general, not point-identified. This result is useful because in many empirical applications it is natural to assume a monotone relationship between outcome variable $Y$ and the endogenous regressor $X$, given by the function $g$, but the main question of interest concerns not the exact shape of $g$ itself, but whether the effect of $X$ on $Y$, given by the slope of $g$, is positive, zero, or negative; see, for example, the discussion in \cite{abrevaya2010re}).

\begin{remark}[A test for the sign of the slope of $g$]\label{rem: monotone reduced form}
	In fact, Theorem \ref{lem: id of sign of Dg} yields a surprisingly simple way to test the sign of the slope of the function $g$. Indeed, the proof of Theorem \ref{lem: id of sign of Dg} reveals that $g$ is increasing, constant, or decreasing if the function $w\mapsto \Ep[Y|W=w]$ is increasing, constant, or decreasing, respectively. By Chebyshev's association inequality (Lemma \ref{lem: Chebyshev's association inequality} in the appendix), the latter assertions are equivalent to the coefficient $\beta$ in the linear regression model
	\begin{equation}\label{eq: linear regression model}
	Y=\alpha+\beta W+U,\,\,\, \Ep[U W]=0
	\end{equation}
	being positive, zero, or negative since $sign(\beta)=sign(cov(W,Y))$ and
	\begin{align*}
		cov(W,Y)&=\Ep[W Y]-\Ep[W]\Ep[Y]\\
		&=\Ep[W\Ep[Y|W]]-\Ep[W] \Ep[\Ep[Y|W]]=cov(W,\Ep[Y|W])
	\end{align*}
	by the law of iterated expectations. Therefore, under our conditions, hypotheses about the sign of the slope of the function $g$ can be tested by testing the corresponding hypotheses about the sign of the slope coefficient $\beta$ in the linear regression model \eqref{eq: linear regression model}. In particular, under our two monotonicity assumptions, one can test the hypothesis of ``no effect'' of $X$ on $Y$, i.e. that $g$ is a constant, by testing whether $\beta=0$ or not using the usual t-statistic. The asymptotic theory for this statistic is exactly the same as in the standard regression case with exogenous regressors, yielding the standard normal limiting distribution and, therefore, completely avoiding the ill-posed inverse problem of recovering $g$. \QEDA
\end{remark}

It turns out that our two monotonicity assumptions possess identifying power even beyond the slope of the regression function.

\begin{definition}[Identified set]
We say that two functions $g',g''\in L^2[0,1]$ are observationally equivalent if $\Ep[g'(X)-g''(X) | W]=0$. The identified set $\Theta$ is defined as the set of all functions $g'\in\mathcal{M}$ that are observationally equivalent to the true function $g$ satisfying \eqref{eq:model}.
\end{definition}

The following theorem provides necessary conditions for observational equivalence.
\begin{theorem}[Identification bounds]\label{lem: id set bounds}
	Let Assumptions~\ref{as: monotone iv} and \ref{as: density} be satisfied, and let $g',g''\in L^2[0,1]$. Further, let $\bar{C} := C_1/c_p$ where $C_1:=(\tilde{x}_2-{\tilde{x}_1})^{1/2}\; / \min\{\tilde{x}_1-x_1, x_2-\tilde{x}_2\}$ and $c_p:= \min\{1-w_2,w_1\} \min\{ C_F-1,2 \}c_wc_f/4$. If there exists a function $h\in L^2[0,1]$ such that $g'-g''+h\in\mathcal{M}$ and $\|h\|_{2,t}+\bar{C}\|T\|_2\|h\|_2 < \|g'-g''\|_{2,t}$, then $g'$ and $g''$ are not observationally equivalent.
\end{theorem}

Under Assumption \ref{as: monotone regression} that $g$ is increasing, Theorem~\ref{lem: id set bounds} suggests the construction of a set $\Theta'$ that includes the identified set $\Theta$ by
$ \Theta' := \mathcal{M}_+\backslash \Delta,$
where $\mathcal{M}_+ := \mathcal{H}(0)$ denotes all increasing functions in $\mathcal{M}$ and 
\begin{multline}
	\Delta := \Big\{ g'\in \mathcal{M}_+:\, \text{there exists } h\in L^2[0,1] \text{ such that}\\
 g'-g+h\in\mathcal{M}\; \text{ and }\; \|h\|_{2,t}+\bar{C}\|T\|_2\|h\|_2 < \|g'-g\|_{2,t} \Big\}.\label{eq: condition id set}
\end{multline}
We emphasize that $\Delta$ is not empty, which means that our Assumptions~\ref{as: monotone iv}--\ref{as: monotone regression} possess identifying power leading to nontrivial bounds on $g$. Notice that the constant $\bar{C}$ depends only on the observable quantities $c_w$, $c_f$, and $C_F$ from Assumptions~\ref{as: monotone iv}--\ref{as: density}, and on the known constants $\tilde{x}_1$, $\tilde{x}_2$, $x_1$, $x_2$, $w_1$, and $w_2$. Therefore, the set $\Theta'$ could, in principle, be estimated, but we leave estimation and inference on this set to future research.
\begin{remark}[Further insight on identification bounds]
It is possible to provide more insight into which functions are in $\Delta$ and thus not in $\Theta'$. First, under the additional minor condition that $f_{X,W}(x,w)>0$ for all $(x,w)\in(0,1)^2$,  all functions in $\Theta'$ have to intersect $g$; otherwise they are not observationally equivalent to $g$. Second, for a given $g'\in\mathcal{M}_+$ and $h\in L^2[0,1]$ such that $g'-g+h$ is monotone, the inequality in condition \eqref{eq: condition id set} is satisfied if $\|h\|_2$ is not too large relative to $\|g'-g\|_{2,t}$. In the extreme case, setting $h=0$ shows that $\Theta'$ does not contain elements $g'$ that disagree with $g$ on $[\widetilde{x}_1,\widetilde{x}_2]$ and such that $g'-g$ is monotone. More generally, $\Theta'$ does not contain elements $g'$ whose difference with $g$ is too close to a monotone function. Therefore, for example, functions $g'$ that are much steeper than $g$ are excluded from $\Theta'$. \QEDA
\end{remark}

\section{Testing the Monotonicity Assumptions}\label{sec: testing}


In this section, we propose tests of our two monotonicity assumptions based on an i.i.d. sample $(X_i,W_i)$, $i=1,\dots,n$, from the distribution of $(X,W)$. First, we discuss an adaptive procedure for testing the stochastic dominance condition \eqref{eq: fosd} in our monotone IV Assumption~\ref{as: monotone iv}. The null and alternative hypotheses are
\begin{align*}
&H_0:\, F_{X|W}(x|w')\geq F_{X|W}(x|w'')\text{ for all }x,w',w''\in(0,1)\text{ with }w'\leq w''\\
&H_a:\, F_{X|W}(x|w')< F_{X|W}(x|w'')\text{ for some }x,w',w''\in(0,1)\text{ with }w'\leq w'',
\end{align*}
respectively. The null hypothesis, $H_0$, is equivalent to stochastic monotonicity of the conditional distribution function $F_{X|W}(x|w)$. Although there exist several good tests of $H_0$ in the literature (see \cite{lee2009fd}, \cite{Delgado:2012fk} and \cite{Lee:2014zl}, for example), to the best of our knowledge there does not exist any procedure that adapts to the unknown smoothness level of $F_{X|W}(x|w)$. We provide a test that is adaptive in this sense, a feature that is not only theoretically attractive, but also important in practice: it delivers a data-driven choice of the smoothing parameter $h_n$ (bandwidth value) of the test whereas nonadaptive tests are usually based on the assumption that $h_n\to 0$ with some rate in a range of prespecified rates, leaving the problem of the selection of an appropriate value of $h_n$ in a given data set to the researcher (see, for example, \cite{lee2009fd} and \cite{Lee:2014zl}). We develop the critical value for the test that takes into account the data dependence induced by the data-driven choice of the smoothing parameter. Our construction leads to a test that controls size, and is asymptotically non-conservative.

Our test is based on the ideas in \cite{Chetverikov:2012fk} who in turn builds on the methods for adaptive specification testing in \cite{horowitz2001ww} and on the theoretical results on high dimensional distributional approximations in \cite{CCK1} (CCK). Note that $F_{X|W}(x|w)=\Ep[1\{X\leq x\}|W=w]$, so that for a fixed $x\in(0,1)$, the hypothesis that $F_{X|W}(x|w')\geq F_{X|W}(x,w'')$ for all $0\leq w'\leq w''\leq 1$ is equivalent to the hypothesis that the regression function $w\mapsto \Ep[1\{X\leq x\}|W=w]$ is decreasing. An adaptive test of this hypothesis was developed in \cite{Chetverikov:2012fk}. In our case, $H_0$ requires the regression function $w\mapsto \Ep[1\{X\leq x\}|W=w]$ to be decreasing not only for a particular value $x\in(0,1)$ but for all $x\in(0,1)$, and so we need to extend the results obtained in \cite{Chetverikov:2012fk}.

Let $K:\mathbb{R}\to \mathbb{R}$ be a kernel function satisfying the following conditions:
\begin{assumption}[Kernel]\label{as: kernel}
	The kernel function $K:\mathbb{R}\to \mathbb{R}$ is such that (i) $K(w)> 0$ for all $w\in(-1,1)$, (ii) $K(w)=0$ for all $w\notin(-1,1)$, (iii) $K$ is continuous, and (iv) $\int_{-\infty}^\infty K(w)dw=1$.
\end{assumption}
We assume that the kernel function $K(w)$ has bounded support, is continuous, and is strictly positive on the support. The last condition excludes higher-order kernels. For a bandwidth value $h>0$, define 
$$
K_h(w):=h^{-1}K(w/h),\quad w\in\mathbb{R}.
$$
Suppose $H_0$ is satisfied. Then, by the law of iterated expectations,
\begin{equation}\label{eq: basis for stat}
\Ep\left[(1\{X_i\leq x\}-1\{X_j\leq x\})\text{sign}(W_i-W_j)K_h(W_i-w)K_h(W_j-w)\right]\leq 0
\end{equation}
for all $x,w\in(0,1)$ and $i,j=1,\dots,n$. Denoting 
$$
K_{i j,h}(w):=\text{sign}(W_i-W_j)K_h(W_i-w)K_h(W_j-w),
$$
taking the sum of the left-hand side in (\ref{eq: basis for stat}) over $i,j=1,\dots,n$, and rearranging give
$$
\Ep\left[\sum_{i=1}^n 1\{X_i\leq x\}\sum_{j=1}^n(K_{i j,h}(w)-K_{j i,h}(w))\right]\leq 0,
$$
or, equivalently,
\begin{equation}\label{eq: basis for stat2}
\Ep\left[\sum_{i=1}^n k_{i,h}(w)1\{X_i\leq x\}\right]\leq 0,
\end{equation}
where 
$$
k_{i,h}(w):=\sum_{j=1}^n(K_{i j,h}(w)-K_{j i,h}(w)).
$$

To define the test statistic $T$, let $\mathcal{B}_n$ be a collection of bandwidth values satisfying the following conditions:
\begin{assumption}[Bandwidth values]\label{as: bandwidth values}
	The collection of bandwidth values is $\mathcal{B}_n:=\{h\in\mathbb{R}:\, h=u^l/2,l=0,1,2,\dots,h\geq h_{\min}\}$ for some $u\in(0,1)$ where $h_{\min}:=h_{\min,n}$ is such that $1/(n h_{\min})\leq C_hn^{-c_h}$ for some constants $c_h,C_h>0$.
\end{assumption}
The collection of bandwidth values $\mathcal{B}_n$ is a geometric progression with the coefficient $u\in(0,1)$, the largest value $1/2$, and the smallest value converging to zero not too fast. As the sample size $n$ increases, the collection of bandwidth values $\mathcal{B}_n$ expands.

Let $\mathcal{W}_n:=\{W_1,\dots,W_n\}$, and $\mathcal{X}_n:=\{\epsilon+l(1-2\epsilon)/n:l=0,1,\dots,n\}$ for some small $\epsilon>0$.
We define our test statistic by
\begin{equation}\label{eq: stat}
T:=\max_{(x,w,h)\in \mathcal{X}_n\times \mathcal{W}_n\times \mathcal{B}_n}\frac{\sum_{i=1}^nk_{i,h}(w)1\{X_i\leq x\}}{\left(\sum_{i=1}^nk_{i,h}(w)^2\right)^{1/2}}.
\end{equation}
The statistic $T$ is most closely related to that in \cite{lee2009fd}. The main difference is that we take the maximum with respect to the set of bandwidth values $h\in\mathcal{B}_n$ to achieve adaptiveness of the test. 

We now discuss the construction of a critical value for the test. Suppose that we would like to have a test of level (approximately) $\alpha$. As succinctly demonstrated by \cite{lee2009fd}, the derivation of the asymptotic distribution of $T$ is complicated even when $\mathcal{B}_n$ is a singleton. Moreover, when $\mathcal{B}_n$ is not a singleton, it is generally unknown whether $T$ converges to some nondegenerate asymptotic distribution after an appropriate normalization. We avoid these complications by employing the non-asymptotic approach developed in CCK and using a multiplier bootstrap critical value for the test. Let $e_1,\dots,e_n$ be an i.i.d. sequence of $N(0,1)$ random variables that are independent of the data. Also, let $\hat{F}_{X|W}(x|w)$ be an estimator of $F_{X|W}(x|w)$ satisfying the following conditions:
\begin{assumption}[Estimator of $F_{X|W}(x|w)$]\label{as: cond cdf estimator}
The estimator $\hat{F}_{X|W}(x|w)$ of $F_{X|W}(x|w)$ is such that (i)
$$
\Pr\left(\Pr\left(\max_{(x,w)\in\mathcal{X}_n\times\mathcal{W}_n}|\hat{F}_{X|W}(x|w)-F_{X|W}(x|w)|>C_Fn^{-c_F}|\{\mathcal{W}_n\}\right)>C_Fn^{-c_F}\right)\leq C_Fn^{-c_F}
$$
for some constants $c_F,C_F>0$, and (ii) $|\hat{F}_{X|W}(x|w)|\leq C_F$ for all $(x,w)\in\mathcal{X}_n\times\mathcal{W}_n$.
\end{assumption}
This is a mild assumption implying uniform consistency of an estimator $\hat{F}_{X|W}(x|w)$ of $F_{X|W}(x|w)$ over $(x,w)\in\mathcal{X}_n\times\mathcal{W}_n$. 
Define a bootstrap test statistic by
$$
T^b:=\max_{(x,w,h)\in\mathcal{X}_n\times\mathcal{W}_n\times \mathcal{B}_n}\frac{\sum_{i=1}^ne_i\left(k_{i,h}(w)(1\{X_i\leq x\}-\hat{F}_{X|W}(x|W_i))\right)}{\left(\sum_{i=1}^nk_{i,h}(w)^2\right)^{1/2}}.
$$
Then we define the critical value\footnote{In the terminology of the moment inequalities literature, $c(\alpha)$ can be considered a ``one-step'' or ``plug-in'' critical value. Following \cite{Chetverikov:2012fk}, we could also consider two-step or even multi-step (stepdown) critical values. For brevity of the paper, however, we do not consider these options here.} $c(\alpha)$ for the test as
$$
c(\alpha):=(1-\alpha)\text{ conditional quantile of }T^b\text{ given the data}.
$$

We reject $H_0$ if and only if $T>c(\alpha)$. To prove validity of this test, we assume that the conditional distribution function $F_{X|W}(x|w)$ satisfies the following condition:
\begin{assumption}[Conditional Distribution Function $F_{X|W}(x|w)$]\label{as: conditional density}
	The conditional distribution function $F_{X|W}(x|w)$ is such that $c_\epsilon\leq F_{X|W}(\epsilon|w)\leq F_{X|W}(1-\epsilon|w)\leq C_\epsilon$ for all $w\in(0,1)$ and some constants $0<c_\epsilon<C_\epsilon<1$.
\end{assumption}

The first theorem in this section shows that our test controls size asymptotically and is not conservative:
\begin{theorem}[Polynomial Size Control]\label{thm: size control}
	Let Assumptions \ref{as: density}, \ref{as: kernel}, \ref{as: bandwidth values}, \ref{as: cond cdf estimator}, and \ref{as: conditional density} be satisfied. If $H_0$ holds, then
	\begin{equation}\label{eq: size control stoch dom}
	\Pr\left(T>c(\alpha)\right)\leq\alpha+Cn^{-c}.
	\end{equation}
	If the functions $w\mapsto F_{X|W}(x|w)$ are constant for all $x\in(0,1)$, then
	\begin{equation}\label{eq: nonconservative stoch dom}
	\left|\Pr\left(T>c(\alpha)\right)-\alpha\right|\leq Cn^{-c}.
	\end{equation}
	In both (\ref{eq: size control stoch dom}) and (\ref{eq: nonconservative stoch dom}), the constants $c$ and $C$ depend only on $c_W,C_W,c_h,C_h,c_F,C_F,c_\epsilon,C_\epsilon$, and the kernel $K$.
\end{theorem}

\begin{remark}[Weak Condition on the Bandwidth Values]
	Our theorem requires
	\begin{equation}\label{eq: weak bandwidth condition}
	\frac{1}{nh}\leq C_hn^{-c_h}
	\end{equation}
	for all $h\in\mathcal{B}_n$, which is considerably weaker than the analogous condition in \cite{lee2009fd} who require $1/(nh^3)\to 0$, up-to logs. This is achieved by using a conditional test and by applying the results of CCK. As follows from the proof of the theorem, the multiplier bootstrap distribution approximates the conditional distribution of the test statistic given $\mathcal{W}_n=\{W_1,\dots,W_n\}$. Conditional on $\mathcal{W}_n$, the denominator in the definition of $T$ is fixed, and does not require any approximation. Instead, we could try to approximate the denominator of $T$ by its probability limit. 
	This is done in \cite{GSV2000} using the theory of Hoeffding projections but they require the condition $1/nh^2\to 0$. Our weak condition (\ref{eq: weak bandwidth condition}) also crucially relies on the fact that we use the results of CCK. Indeed, it has already been demonstrated (see \cite{CCK2,Chernozhukov:2013kx}, and \cite{Belloni:2014hl}) that, in typical nonparametric problems, the techniques of CCK often lead to weak conditions on the bandwidth value or the number of series terms. Our theorem is another instance of this fact.\QEDA
\end{remark}
\begin{remark}[Polynomial Size Control]
	Note that, by (\ref{eq: size control stoch dom}) and (\ref{eq: nonconservative stoch dom}), the probability of rejecting $H_0$ when $H_0$ is satisfied can exceed the nominal level $\alpha$ only by a term that is polynomially small in $n$. We refer to this phenomenon as a \textit{polynomial size control}. As explained in \cite{lee2009fd}, when $\mathcal{B}_n$ is a singleton, convergence of $T$ to the limit distribution is logarithmically slow. Therefore, \cite{lee2009fd} used higher-order corrections derived in \cite{Piterbarg1996} to obtain polynomial size control. Here we show that the multiplier bootstrap also gives higher-order corrections and leads to polynomial size control. This feature of our theorem is also inherited from the results of CCK.\QEDA
\end{remark}
\begin{remark}[Uniformity]
	The constants $c$ and $C$ in (\ref{eq: size control stoch dom}) and (\ref{eq: nonconservative stoch dom}) depend on the data generating process only via constants (and the kernel) appearing in Assumptions \ref{as: density}, \ref{as: kernel}, \ref{as: bandwidth values}, \ref{as: cond cdf estimator}, and \ref{as: conditional density}. Therefore, inequalities (\ref{eq: size control stoch dom}) and (\ref{eq: nonconservative stoch dom}) hold uniformly over all data generating processes satisfying these assumptions with the same constants. We obtain uniformity directly from employing the distributional approximation theorems of CCK because they are non-asymptotic and do not rely on convergence arguments.\QEDA
\end{remark}

Our second result in this section concerns the ability of our test to detect models in the alternative $H_a$. Let $\epsilon>0$ be the constant appearing in the definition of $T$ via the set $\mathcal{X}_n$.
\begin{theorem}[Consistency]\label{thm: consistency}
	Let Assumptions \ref{as: density}, \ref{as: kernel}, \ref{as: bandwidth values}, \ref{as: cond cdf estimator}, and \ref{as: conditional density} be satisfied and assume that $F_{X|W}(x|w)$ is continuously differentiable. If $H_a$ holds with $D_wF_{X|W}(x|w)>0$ for some $x\in(\epsilon,1-\epsilon)$ and $w\in(0,1)$, then
	\begin{equation}\label{eq: consistency}
	\Pr\left(T>c(\alpha)\right)\to 1\text{ as }n\to\infty.
	\end{equation}
\end{theorem}
This theorem shows that our test is consistent against any model in $H_a$ (with smooth $F_{X|W}(x|w)$) whose deviation from $H_0$ is not on the boundary, so that the deviation $D_wF_{X|W}(x|w)>0$ occurs for $x\in(\epsilon,1-\epsilon)$. It is also possible to extend our results to show that Theorems~\ref{thm: size control} and \ref{thm: consistency} hold with $\epsilon=0$ at the expense of additional technicalities. Further, using the same arguments as those in \cite{Chetverikov:2012fk}, it is possible to show that the test suggested here has minimax optimal rate of consistency against the alternatives belonging to certain H\"{o}lder classes for a reasonably large range of smoothness levels. We do not derive these results here for the sake of brevity of presentation.  

We conclude this section by proposing a simple test of our second monotonicity assumption, that is, monotonicity of the regression function $g$. The null and alternative hypotheses are
\begin{align*}
&H_0:\, g(x')\leq g(x'') \text{ for all } x',x''\in(0,1)\text{ with } x'\leq x''\\
&H_a:\, g(x')> g(x'') \text{ for some } x',x''\in(0,1)\text{ with } x'\leq x'',
\end{align*}
respectively. The discussion in Remark~\ref{rem: monotone reduced form} reveals that, under Assumptions~\ref{as: monotone iv} and \ref{as: density}, monotonicity of $g(x)$ implies monotonicity of $w\mapsto \Ep[Y|W=w]$. Therefore, under Assumptions~\ref{as: monotone iv} and \ref{as: density}, we can test $H_0$ by testing monotonicity of the conditional expectation $w\mapsto \Ep[Y|W=w]$ using existing tests such as \cite{Chetverikov:2012fk} and \cite{Lee:2014zl}, among others. This procedure tests an implication of $H_0$ instead of $H_0$ itself and therefore may have low power against some alternatives. On the other hand, it does not require solving the model for $g(x)$ and therefore avoids the ill-posedness of the problem.

\section{Simulations}\label{sec: simulations}

In this section, we study the finite-sample behavior of our constrained estimator that imposes monotonicity and compare its performance to that of the unconstrained estimator. We consider the NPIV model $Y=g(X) + \varepsilon$, $\Ep[\varepsilon | W]=0$, for two different regression functions, one that is strictly increasing and a weakly increasing one that is constant over part of its domain:
\begin{description}
	\item[Model 1:] $g(x) = \kappa\sin(\pi x -\pi/2)$
	\item[Model 2:] $g(x) = 10\kappa \left[ -(x-0.25)^2 \mathbf{1}\{x\in[0,0.25]\} + (x-0.75)^2 \mathbf{1}\{x\in[0.75,1] \} \right]$
\end{description}
where $\varepsilon=\kappa\sigma_{\varepsilon}\bar{\varepsilon}$ and $\bar{\varepsilon} = \eta \epsilon + \sqrt{1-\eta^2}\nu$. The regressor and instrument are generated by $X=\Phi(\xi)$ and $W=\Phi(\zeta)$, respectively, where $\Phi$ is the standard normal cdf and $\xi = \rho \zeta + \sqrt{1-\rho^2}\epsilon$. The errors are generated by $(\nu,\zeta,\epsilon)\sim N(0,I)$.

We vary the parameter $\kappa$ in $\{1,0.5,0.1\}$ to study how the constrained and unconstrained estimators' performance compares depending on the maximum slope of the regression function. $\eta$ governs the dependence of $X$ on the regression error $\varepsilon$ and $\rho$ the strength of the first stage. All results are based on $1,000$ MC samples and the normalized B-spline basis for $p(x)$ and $q(w)$ of degree $3$ and $4$, respectively.

Tables~\ref{tab:model1n500}--\ref{tab:re_model2} report the Monte Carlo approximations to the squared bias, variance, and mean squared error (``MSE'') of the two estimators, each averaged over a grid on the interval $[0,1]$. We also show the ratio of the constrained estimator's MSE divided by the unconstrained estimator's MSE. $k_X$ and $k_W$ denote, respectively, the number of knots used for the basis $p(x)$ and $q(w)$. The first two tables vary the number of knots, and the latter two the dependence parameters $\rho$ and $\eta$. Different sample sizes and different values for $\rho$, $\eta$, and $\sigma_{\varepsilon}$ yield qualitatively similar results. Figures~\ref{fig:estim_model1} and \ref{fig:estim_model2} show the two estimators for a particular combination of the simulation parameters. The dashed lines represent confidence bands, computed as two times the (pointwise) empirical standard deviation of the estimators across simulation samples. Both, the constrained and the unconstrained, estimators are computed by ignoring the bound $\|b\|\leq C_b$ in their respective definitions. \cite{Horowitz2012as} and \cite{horowitz2012} also ignore the constraint $\|b\|\leq C_b$ and state that it does not affect the qualitative results of their simulation experiment.

The MSE of the constrained estimator (and, interestingly, also of the unconstrained estimator) decreases as the regression function becomes flatter. This observation is consistent with the error bound in Theorem~\ref{thm: risk bounds} depending positively on the maximum slope of $g$.

Because of the joint normality of $(X,W)$, the simulation design is severely ill-posed and we expect high variability of both estimators. In all simulation scenarios, we do in fact observe a very large variance relative to bias. However, the magnitude of the variance differs significantly across the two estimators: in all scenarios, even in the design with a strictly increasing regression function, imposing the monotonicity constraint significantly reduces the variance of the NPIV estimator. The MSE of the constrained estimator is therefore much smaller than that of the unconstrained estimator, from about a factor of two smaller when $g$ is strictly increasing and the noise level is low ($\sigma_{\varepsilon}=0.1$), to around 20 times smaller when $g$ contains a flat part and the noise level is high ($\sigma_{\varepsilon}=0.7$). Generally, the gains in MSE from imposing monotonicity are larger the higher the noise level $\sigma_{\varepsilon}$ in the regression equation and the higher the first-stage correlation $\rho$.\footnote{Since Tables~\ref{tab:model1n500} and \ref{tab:model2n500} report results for the lower level of $\rho$, and Tables~\ref{tab:re_model1} and \ref{tab:re_model2} results for the lower noise level $\sigma_{\varepsilon}$, we consider the selection of results as, if at all, favoring the unconstrained estimator.}


\section{Gasoline Demand in the United States}
\label{sec: gasoline demand}

In this section, we revisit the problem of estimating demand functions for gasoline in the United States. Because of the dramatic changes in the oil price over the last few decades, understanding the elasticity of gasoline demand is fundamental to evaluating tax policies. Consider the following partially linear specification of the demand function:
$$Y = g(X,Z_1) + \gamma'Z_2 + \varepsilon,\qquad \Ep[\varepsilon | W,Z_1,Z_2]=0,$$
where $Y$ denotes annual log-gasoline consumption of a household, $X$ log-price of gasoline (average local price), $Z_1$ log-household income, $Z_2$ are control variables (such as population density, urbanization, and demographics), and $W$ distance to major oil platform. We allow for price $X$ to be endogenous, but assume that $(Z_1,Z_2)$ is exogenous. $W$ serves as an instrument for price by capturing transport cost and, therefore, shifting the cost of gasoline production. We use the same sample of size $4,812$ from the 2001 National Household Travel Survey and the same control variables $Z_2$ as \cite{blundell2012fg}. More details can be found in their paper.

Moving away from constant price and income elasticities is likely very important as individuals' responses to price changes vary greatly with price and income level. Since economic theory does not provide guidance on the functional form of $g$, finding an appropriate parametrization is difficult. \cite{hausman1995} and \cite{blundell2012fg}, for example, demonstrate the importance of employing flexible estimators of $g$ that do not suffer from misspecification bias due to arbitrary restrictions in the model. \cite{Blundell:2013fk} argue that prices at the local market level vary for several reasons and that they may reflect preferences of the consumers in the local market. Therefore, one would expect prices $X$ to depend on unobserved factors in $\varepsilon$ that determine consumption, rendering price an endogenous variable. Furthermore, the theory of the consumer requires downward-sloping compensated demand curves. Assuming a positive income derivative\footnote{\cite{blundell2012fg} estimate this income derivative and do, in fact, find it to be positive over the price range of interest.} $\partial g/\partial z_1$, the Slutsky condition implies that the uncompensated (Marshallian) demand curves are also downward-sloping, i.e. $g(\cdot,z_1)$ should be monotone for any $z_1$, as long as income effects do not completely offset price effects. Finally, we expect the cost shifter $W$ to monotonically increase cost of producing gasoline and thus satisfy our monotone IV condition. In conclusion, our constrained NPIV estimator appears to be an attractive estimator of demand functions in this setting.

We consider three benchmark estimators. First, we compute the unconstrained nonparametric (``uncon. NP'') series estimator of the regression of $Y$ on $X$ and $Z_1$, treating price as exogenous. As in \cite{blundell2012fg}, we accommodate the high-dimensional vector of additional, exogenous covariates $Z_2$ by (i) estimating $\gamma$ by \cite{198807}'s procedure, (ii) then removing these covariates from the outcome, and (iii) estimating $g$ by regressing the adjusted outcomes on $X$ and $Z_1$. The second benchmark estimator (``con. NP'') repeats the same steps (i)--(iii) except that it imposes monotonicity (in price) of $g$ in steps (i) and (iii). The third benchmark estimator is the unconstrained NPIV estimator (``uncon. NPIV'') that accounts for the covariates $Z_2$ in similar fashion as the first, unconstrained nonparametric estimator, except that (i) and (iii) employ NPIV estimators that impose additive separability and linearity in $Z_2$.

The fourth estimator we consider is the constrained NPIV estimator (``con. NPIV'') that we compare to the three benchmark estimators. We allow for the presence of the covariates $Z_2$ in the same fashion as the unconstrained NPIV estimator except that, in steps (i) and (iii), we impose monotonicity in price.

We report results for the following choice of bases. All estimators employ a quadratic B-spline basis with $3$ knots for price $X$ and a cubic B-spline with $10$ knots for the instrument $W$. Denote these two bases by $\mathbf{P}$ and $\mathbf{Q}$, using the same notation as in Section~\ref{sec: estimation}. In step (i), the NPIV estimators include the additional exogenous covariates $(Z_1,Z_2)$ in the respective bases for $X$ and $W$, so they use the estimator defined in Section~\ref{sec: estimation} except that the bases $\mathbf{P}$ and $\mathbf{Q}$ are replaced by $\tilde{\mathbf{P}} := [\mathbf{P}, \mathbf{P}\times \mathbf{Z}_1, \mathbf{Z}_2]$ and $\tilde{\mathbf{Q}} := [\mathbf{Q}, \mathbf{Q}\times (\mathbf{Z}_1, \mathbf{Z}_2)]$, respectively, where $\mathbf{Z}_k := (Z_{k,1},\ldots,Z_{k,n})'$, $k=1,2$, stacks the observations $i=1,\ldots,n$ and $\mathbf{P}\times \mathbf{Z}_1$ denotes the tensor product of the columns of the two matrices. Since, in the basis $\tilde{\mathbf{P}}$, we include interactions of $\mathbf{P}$ with $\mathbf{Z}_1$, but not with $\mathbf{Z}_2$, the resulting estimator allows for a nonlinear, nonseparable dependence of $Y$ on $X$ and $Z_1$, but imposes additive separability in $Z_2$. The conditional expectation of $Y$ given $W$, $Z_1$, and $Z_2$ does not have to be additively separable in $Z_2$, so that, in the basis $\tilde{\mathbf{Q}}$, we include interactions of $\mathbf{Q}$ with both $\mathbf{Z}_1$ and $\mathbf{Z}_2$.\footnote{Notice that $\mathbf{P}$ and $\mathbf{Q}$ include constant terms so it is not necessary to separately include $\mathbf{Z}_k$ in addition to its interactions with $\mathbf{P}$ and $\mathbf{Q}$, respectively.}

We estimated the demand functions for many different combinations of the order of B-spline for $W$, the number of knots in both bases, and even with various penalization terms (as discussed in Remark~\ref{rem: penalization}). While the shape of the unconstrained NPIV estimate varied slightly across these different choices of tuning parameters (mostly near the boundary of the support of $X$), the constrained NPIV estimator did not exhibit any visible changes at all.

Figure~\ref{fig:cdfXWhat} shows a nonparametric kernel estimate of the conditional distribution of the price $X$ given the instrument $W$. Overall the graph indicates an increasing relationship between the two variables as required by our stochastic dominance condition \eqref{eq: fosd}. We formally test this monotone IV assumption by applying our new test proposed in Section~\ref{sec: testing}. We find a test statistic value of $0.139$ and $95\%$-critical value of $1.720$.\footnote{The critical value is computed from $1,000$ bootstrap samples, using the bandwidth set $\mathcal{B}_n = \{2,1,0.5,0.25,0.125,0.0625\}$, and a kernel estimator for $\hat{F}_{X|W}$ with bandwidth $0.3$ which produces the estimate in Figure~\ref{fig:cdfXWhat}.} Therefore, we fail to reject the monotone IV assumption.

Figure~\ref{fig:emp with cov} shows the estimates of the demand function at three income levels, at the lower quartile ($\$ 42,500$), the median ($\$ 57,500$), and the upper quartile ($\$ 72,500$). The area shaded in grey represents the 90\% uniform confidence bands around the unconstrained NPIV estimator as proposed in \cite{Horowitz2012as}.\footnote{Critical values are computed from $1,000$ bootstrap samples and the bands are computed on a grid of $100$ equally-spaced points in the support of the data for $X$.} The black lines correspond to the estimators assuming exogeneity of price and the red lines to the NPIV estimators that allow for endogeneity of price. The dashed black line shows the kernel estimate of \cite{blundell2012fg} and the solid black line the corresponding series estimator that imposes monotonicity. The dashed and solid red lines similarly depict the unconstrained and constrained NPIV estimators, respectively.

All estimates show an overall decreasing pattern of the demand curves, but the two unconstrained estimators are both increasing over some parts of the price domain. We view these implausible increasing parts as finite-sample phenomena that arise because the unconstrained nonparametric estimators are too imprecise. The wide confidence bands of the unconstrained NPIV estimator are consistent with this view. \cite{hausman1995} and \cite{Horowitz2012as} find similar anomalies in their nonparametric estimates, assuming exogenous prices. Unlike the unconstrained estimates, our constrained NPIV estimates are downward-sloping everywhere and smoother. They lie within the $90\%$ uniform confidence bands of the unconstrained estimator so that the monotonicity constraint appears compatible with the data.

The two constrained estimates are very similar, indicating that endogeneity of prices may not be important in this problem, but they are both significantly flatter than the unconstrained estimates across all three income groups, which implies that households appear to be less sensitive to price changes than the unconstrained estimates suggest. The small maximum slope of the constrained NPIV estimator also suggests that the error bound in Theorem~\ref{thm: risk bounds} may be small and therefore we expect the constrained NPIV estimate to be precise for this data set.



\appendix

\newpage
\section{Proofs for Section  \ref{sec: MIVE}}

For any $h\in L^1[0,1]$, let $\|h\|_{1} := \int_{0}^{1} | h(x)|d x$, $\|h\|_{1,t} := \int_{x_1}^{x_2} | h(x)|d x$ and define the operator norm by $\|T\|_2:=\sup_{h\in L^2[0,1]:\,\|h\|_2>0}\|Th\|_2 / \|h\|_2$. Note that $\|T\|_2\leq \int_0^1\int_0^1f_{X,W}^2(x,w)d x d w$, and so under Assumption \ref{as: density}, $\|T\|_2\leq C_T$.

\begin{proof}[Proof of Theorem~\ref{thm: main result 1}]

	We first show that for any $h\in \mathcal{M}$,
	\begin{equation}\label{eq:bound2t_by_1}
		\|h\|_{2,t} \leq C_1 \|h\|_{1,t}
	\end{equation}
	for $C_1 := (\tilde{x}_2-{\tilde{x}_1})^{1/2}\; / \min\{\tilde{x}_1-x_1, x_2-\tilde{x}_2\}$. Indeed, by monotonicity of $h$,
	\begin{align*}
		\|h\|_{2,t} &= \left( \int_{\tilde{x}_1}^{\tilde{x}_2} h(x)^2dx \right)^{1/2} \leq \sqrt{\tilde{x}_2-{\tilde{x}_1}} \max\left\{ |h({\tilde{x}_1})|, |h(\tilde{x}_2)| \right\}\\
			&\leq \sqrt{\tilde{x}_2-{\tilde{x}_1}}\; \frac{\int_{x_1}^{x_2} |h(x)|dx}{ \min\left\{\tilde{x}_1-x_1, x_2-\tilde{x}_2\right\}}
	\end{align*}
	so that \eqref{eq:bound2t_by_1} follows. Therefore, for any increasing continuously differentiable $h\in \mathcal{M}$,
	$$ \|h\|_{2,t} \leq C_1 \|h\|_{1,t} \leq C_1 C_2 \|Th\|_1 \leq C_1 C_2 \|Th\|_2, $$
	where the first inequality follows from \eqref{eq:bound2t_by_1}, the second from Lemma~\ref{lem: crucial lemma} below (which is the main step in the proof of the theorem), and the third by Jensen's inequality. Hence, conclusion (\ref{eq: bounded inverse}) of Theorem \ref{thm: main result 1} holds for increasing continuously differentiable $h\in\mathcal{M}$ with $\bar{C}:=C_1C_2$ and $C_2$ as defined in Lemma~\ref{lem: crucial lemma}. 

Next, for any increasing function $h\in\mathcal{M}$, it follows from Lemma \ref{lem: monotone approximation} that one can find a sequence of increasing continuously differentiable functions $h_k\in\mathcal{M}$, $k\geq 1$, such that $\|h_k-h\|_2\to 0$ as $k\to \infty$. Therefore, by the triangle inequality,
\begin{align*}
\|h\|_{2,t}&\leq \|h_k\|_{2,t}+\|h_k-h\|_{2,t}\leq \bar{C}\|Th_k\|_2+\|h_k-h\|_{2,t}\\
&\leq \bar{C}\|Th\|_2+\bar{C}\|T(h_k-h)\|_2+\|h_k-h\|_{2,t}\\
&\leq \bar{C}\|Th\|_2+\bar{C}\|T\|_2\|(h_k-h)\|_2+\|h_k-h\|_{2,t}\\
&\leq \bar{C}\|Th\|_2+(\bar{C}\|T\|_2+1)\|(h_k-h)\|_2\\
&\leq \bar{C}\|Th\|_2+(\bar{C}C_T+1)\|h_k-h\|_2
\end{align*}
where the third line follows from the Cauchy-Schwarz inequality, the fourth from $\|h_k-h\|_{2,t}\leq \|h_k-h\|_2$, and the fifth from Assumption \ref{as: density}(i). Taking the limit as $k\to\infty$ of both the left-hand and the right-hand sides of this chain of inequalities yields conclusion (\ref{eq: bounded inverse}) of Theorem \ref{thm: main result 1} for all increasing $h\in\mathcal{M}$.

Finally, since for any decreasing $h\in\mathcal{M}$, we have that $-h\in\mathcal{M}$ is increasing, $\|-h\|_{2,t}=\|h\|_{2,t}$ and $\|Th\|_2=\|T(-h)\|_2$, conclusion (\ref{eq: bounded inverse}) of Theorem \ref{thm: main result 1} also holds for all decreasing $h\in\mathcal{M}$, and thus for all $h\in\mathcal{M}$. This completes the proof of the theorem.
\end{proof}

\begin{lemma}\label{lem: crucial lemma}
	Let Assumptions~\ref{as: monotone iv} and \ref{as: density} hold. Then for any increasing continuously differentiable $h\in L^1[0,1]$,
	$$ \|h\|_{1,t} \leq C_2 \|Th\|_1 $$
	where $C_2 := 1/c_p$ and $c_p:= c_wc_f/2 \min\{1-w_2,w_1\} \min\{ (C_F-1)/2,1 \}$.
\end{lemma}

\begin{proof} 
	Take any increasing continuously differentiable function $h\in L^1[0,1]$ such that $\|h\|_{1,t}=1$. Define $M(w):=\Ep[h(X)|W=w]$ for all $w\in[0,1]$ and note that
	$$
	\|Th\|_1 = \int_0^1 |M(w)f_W(w)|d w \geq c_W \int_{0}^{1} |M(w)|d w
	$$
	where the inequality follows from Assumption \ref{as: density}(iii). Therefore, the asserted claim follows if we can show that $\int_{0}^{1}|M(w)|d w$ is bounded away from zero by a constant that depends only on $\zeta$.

	First, note that $M(w)$ is increasing. This is because, by integration by parts,
	$$ M(w) = \int_0^1 h(x)f_{X|W}(x|w)dx = h(1) - \int_0^1 Dh(x) F_{X|W}(x|w)dx, $$
	so that condition \eqref{eq: fosd} of Assumption \ref{as: monotone iv} and $Dh(x)\geq 0$ for all $x$ imply that the function $M(w)$ is increasing. 

	Consider the case in which $h(x)\geq 0$ for all $x\in [0,1]$. Then $M(w)\geq 0$ for all $w\in[0,1]$. Therefore,
	\begin{align*}
		\int_{0}^{1} | M(w)|d w &\geq \int_{w_2}^{1} | M(w)|d w \geq (1-w_2) M(w_2)= (1-w_2) \int_0^1 h(x) f_{X|W}(x|w_2)d x\\
			&\geq (1-w_2) \int_{x_1}^{x_2} h(x) f_{X|W}(x|w_2)d x\geq (1-w_2) c_f \int_{x_1}^{x_2} h(x)d x \\
			&=(1-w_2)c_f\|h\|_{1,t}= (1-w_2)c_f > 0
	\end{align*}
by Assumption \ref{as: density}(ii).
	Similarly,
	$$
	\int_{0}^{1} | M(w)|d w\geq w_1c_f>0
	$$
	when $h(x)\leq 0$ for all $x\in [0,1]$. Therefore, it remains to consider the case in which there exists $x^*\in(0,1)$ such that $h(x)\leq 0$ for $x\leq x^*$ and $h(x)\geq 0$ for $x> x^*$. Since $h(x)$ is continuous, $h(x^*)=0$, and so integration by parts yields
	\begin{align}
		M(w) &= \int_0^{x^*} h(x) f_{X|W}(x|w)d x + \int_{x^*}^1 h(x) f_{X|W}(x|w)d x\nonumber\\
			&= -\int_0^{x^*} D h(x) F_{X|W}(x|w)d x + \int_{x^*}^1 D h(x) (1-F_{X|W}(x|w))d x. \label{eq: M decomposition}
	\end{align}
	For $k=1,2$, let $A_k := \int_{x^*}^{1} D h(x) (1-F_{X|W}(x|w_k))$ and $B_k:= \int_{0}^{x^*} D h(x) F_{X|W}(x|w_k)d x$, so that $M(w_k) = A_k - B_k$.

	Consider the following three cases separately, depending on where $x^*$ lies relative to $x_1$ and $x_2$.

	\paragraph{Case I ($x_1< x^* < x_2$):} First, we have
		\begin{align}
			A_1 + B_2 &= \int_{x^*}^{1} D h(x) (1-F_{X|W}(x|w_1))d x + \int_{0}^{x^*} D h(x) F_{X|W}(x|w_2)d x\nonumber\\
				&= \int_{x^*}^{1} h(x) f_{X|W}(x|w_1)d x - \int_{0}^{x^*} h(x) f_{X|W}(x|w_2)d x\nonumber\\
				&\geq \int_{x^*}^{x_2} h(x) f_{X|W}(x|w_1)d x - \int_{x_1}^{x^*} h(x) f_{X|W}(x|w_2)d x\nonumber\\
				&\geq c_1\int_{x^*}^{x_2} h(x) d x + c_f\int_{x_1}^{x^*} |h(x)|d x=c_f\int_{x_1}^{x_2}|h(x)|d x\nonumber\\
				&= c_f \|h\|_{1,t} = c_f > 0 \label{eq: A plus B bound}
		\end{align}
		where the fourth line follows from Assumption~\ref{as: density}(ii). Second, by \eqref{eq: fosd} and \eqref{eq: lower cdf bound} of Assumption~\ref{as: monotone iv},
		\begin{align*}
			M(w_1) &= \int_{x^*}^{1} D h(x) (1-F_{X|W}(x|w_1))d x -\int_{0}^{x^*} D h(x) F_{X|W}(x|w_1)d x\\
				&\leq \int_{x^*}^{1} D h(x) (1-F_{X|W}(x|w_2))d x - C_F \int_{0}^{x^*} D h(x) F_{X|W}(x|w_2)d x\\
				&= A_2 - C_F B_2
		\end{align*}
		so that, together with $M(w_2) = A_2-B_2$, we obtain
		\begin{equation}\label{eq: M difference in B}
			M(w_2)-M(w_1) \geq (C_F-1) B_2.
		\end{equation}
		Similarly, by \eqref{eq: fosd} and \eqref{eq: upper cdf bound} of Assumption \ref{as: monotone iv},
		\begin{align*}
			M(w_2) &= \int_{x^*}^{1} D h(x) (1-F_{X|W}(x|w_2))d x -\int_{0}^{x^*} D h(x) F_{X|W}(x|w_2)d x\\
				&\geq C_F \int_{x^*}^{1} D h(x) (1-F_{X|W}(x|w_1))d x - \int_{0}^{x^*} D h(x) F_{X|W}(x|w_1)d x\\
				&= C_FA_1 - B_1
		\end{align*}
		so that, together with $M(w_1) = A_1-B_1$, we obtain
		\begin{equation}\label{eq: M difference in A}
			M(w_2)-M(w_1) \geq (C_F-1) A_1.
		\end{equation}
		In conclusion, equations \eqref{eq: A plus B bound}, \eqref{eq: M difference in B}, and \eqref{eq: M difference in A} yield
		\begin{equation}
			M(w_2) - M(w_1) \geq (C_F-1) (A_1 + B_2)/2 \geq (C_F-1)c_f/2>0.
		\end{equation}
		Consider the case $M(w_1)\geq 0$ and $M(w_2)\geq 0$. Then $M(w_2) \geq M(w_2)-M(w_1)$ and thus
		\begin{align}
			\int_{0}^{1} |M(w)|d w &\geq \int_{w_2}^{1} |M(w)|d w \geq (1-w_2) M(w_2)\geq (1-w_2)(C_F-1)c_f/2>0.\label{eq: M norm bound}
		\end{align}
		Similarly,
		\begin{equation}\label{eq: negative case bound}
		\int_{0}^{1} |M(w)|d w\geq w_1(C_F-1)c_f/2>0
		\end{equation}
		when $M(w_1)\leq 0$ and $M(w_2)\leq 0$.
				
		Finally, consider the case $M(w_1)\leq 0$ and $M(w_2)\geq 0$. If $M(w_2)\geq |M(w_1)|$, then $M(w_2)\geq (M(w_2)-M(w_1))/2$ and the same argument as in \eqref{eq: M norm bound} shows that 
		$$
		\int_{0}^{1} |M(w)|d w \geq (1-w_2)(C_F-1)c_f/4.
		$$
		If $|M(w_1)|\geq M(w_2)$, then $|M(w_1)|\geq (M(w_2)-M(w_1))/2$ and we obtain
				$$
					\int_{0}^{1} |M(w)|d w \geq \int_{0}^{w_1} |M(w)|d w \geq w_1(C_F-1)c_f/4>0.\label{eq: M norm bound 2}
				$$
		This completes the proof of Case I.

	\paragraph{Case II ($x_2 \leq x^* $):} Suppose $M(w_1)\geq -c_f/2$. As in Case I, we have $M(w_2) \geq C_FA_1 -B_1$. Together with $M(w_1) = A_1 - B_1$, this inequality yields
		\begin{align*}
			M(w_2)-M(w_1) &= M(w_2) - C_F M(w_1) + C_F M(w_1) - M(w_1)\\
				&\geq (C_F-1) B_1 + (C_F-1) M(w_1)\\
				&= (C_F-1) \left( \int_0^{x^*} D h(x) F_{X|W}(x|w_1)d x + M(w_1)\right)\\
				&= (C_F-1) \left(  \int_0^{x^*} |h(x)| f_{X|W}(x|w_1)d x + M(w_1)\right)\\
				&\geq (C_F-1) \left( \int_{x_1}^{x_2} |h(x)| f_{X|W}(x|w_1)d x-\frac{c_f}{2}\right)\\
				&\geq (C_F-1)\left(c_f \int_{x_1}^{x_2} |h(x)|d x -\frac{c_f}{2}\right) = \frac{(C_F-1)c_f}{2} > 0
		\end{align*}
		With this inequality we proceed as in Case I to show that $\int_{0}^{1} |M(w)|d w$ is bounded from below by a positive constant that depends only on $\zeta$. On the other hand, when $M(w_1)\leq -c_f/2$ we bound $\int_{0}^{1} |M(w)|d w$ as in \eqref{eq: negative case bound}, and the proof of Case II is complete.

	\paragraph{Case III ($x^* \leq x_1$):} Similarly as in Case II, suppose first that $M(w_2)\leq c_f/2$. As in Case I we have $M(w_1) \leq A_2 - C_F B_2$ so that together with $M(w_2) = A_2 - B_2$,
		\begin{align*}
			M(w_2)-M(w_1) &= M(w_2)-C_FM(w_2) + C_FM(w_2)-M(w_1)\\
				&\geq (1-C_F) M(w_2) + (C_F-1) A_2 \\
				&= (C_F-1) \left(\int_{x^*}^1 D h(x) (1-F_{X|W}(x|w_2))d x - M(w_2)\right)\\
				&= (C_F-1) \left(\int_{x^*}^1 h(x) f_{X|W}(x|w_2)d x - M(w_2)\right)\\
				&\geq (C_F-1) \left(\int_{x_1}^{x_2} h(x) f_{X|W}(x|w_2)d x - M(w_2)\right)\\
				&\geq (C_F-1) \left(c_f \int_{x_1}^{x_2} h(x) d x - \frac{c_f}{2}\right) = \frac{(C_F-1) c_f}{2}>0
		\end{align*}
		and we proceed as in Case I to bound $\int_{0}^{1} |M(w)|d w$ from below by a positive constant that depends only on $\zeta$. On the other hand, when $M(w_2)> c_f/2$, we bound $\int_{0}^{1} |M(w)|d w$ as in \eqref{eq: M norm bound}, and the proof of Case III is complete. The lemma is proven.
\end{proof}

\begin{proof}[Proof of Corollary~\ref{co:bound_tau}]
	Note that since $\tau(a')\leq \tau(a'')$ whenever $a'\leq a''$, the claim for $a\leq 0$, follows from $\tau(a)\leq \tau(0)\leq \bar{C}$, where the second inequality holds by Theorem \ref{thm: main result 1}. Therefore, assume that $a>0$. Fix any $\alpha\in(0,1)$. Take any function $h\in \mathcal{H}(a)$ such that $\|h\|_{2,t}=1$. Set $h'(x)=ax$ for all $x\in[0,1]$. Note that the function $x\mapsto h(x) + ax$ is increasing and so belongs to the class $\mathcal{M}$. Also, $\|h'\|_{2,t}\leq \|h'\|_2\leq a/\sqrt{3}$. Thus, the bound \eqref{eq: main corollary general} in Lemma~\ref{lem: general bound on tau} below applies whenever $(1+\bar{C}\|T\|_2)a/\sqrt{3}\leq \alpha$. Therefore, for all $a$ satisfying the inequality
	$$
	a\leq \frac{\sqrt{3}\alpha}{1+\bar{C}\|T\|_2},
	$$
	we have $\tau(a)\leq \bar{C}/(1-\alpha)$. This completes the proof of the corollary.
\end{proof}

\begin{lemma}\label{lem: general bound on tau}
	Let Assumptions~\ref{as: monotone iv} and \ref{as: density} be satisfied. Consider any function $h\in L^2[0,1]$. If there exist $h'\in L^2[0,1]$ and $\alpha\in(0,1)$ such that $h+h'\in\mathcal{M}$ and $\|h'\|_{2,t}+\bar{C}\|T\|_2\|h'\|_2\leq \alpha\|h\|_{2,t}$, then
			\begin{equation}\label{eq: main corollary general}
			\|h\|_{2,t}\leq \frac{\bar{C}}{1-\alpha}\|Th\|_2
			\end{equation}
			for the constant $\bar{C}$ defined in Theorem \ref{thm: main result 1}.
\end{lemma}

\begin{proof}
	Define
	$$
	\tilde{h}(x):=\frac{h(x)+h'(x)}{\|h\|_{2,t}-\|h'\|_{2,t}},\quad x\in[0,1].
	$$
	By assumption, $\|h'\|_{2,t}<\|h\|_{2,t}$, and so the triangle inequality yields
	$$
	\|\tilde{h}\|_{2,t}\geq \frac{\|h\|_{2,t}-\|h'\|_{2,t}}{\|h\|_{2,t}-\|h'\|_{2,t}}=1.
	$$
	Therefore, since $\tilde{h}\in\mathcal{M}$, Theorem \ref{thm: main result 1} gives 
	$$
	\|T\tilde{h}\|_2\geq \|\tilde{h}\|_{2,t}/\bar{C}\geq 1/\bar{C}.
	$$
	Hence, applying the triangle inequality once again yields
	\begin{align*}
		\|Th\|_2&\geq (\|h\|_{2,t}-\|h'\|_{2,t})\|T\tilde{h}\|_2-\|Th'\|_2\geq (\|h\|_{2,t}-\|h'\|_{2,t})\|T\tilde{h}\|_2-\|T\|_2\|h'\|_2\\
			&\geq \frac{\|h\|_{2,t}-\|h'\|_{2,t}}{\bar{C}}-\|T\|_2\|h'\|_2
			=\frac{\|h\|_{2,t}}{\bar{C}}\left(1-\frac{\|h'\|_{2,t}+\bar{C}\|T\|_2\|h'\|_2}{\|h\|_{2,t}}\right)
	\end{align*}
	Since the expression in the last parentheses is bounded from below by $1-\alpha$ by assumption, we obtain the inequality
	$$
	\|Th\|_2\geq \frac{1-\alpha}{\bar{C}}\|h\|_{2,t},
	$$
	which is equivalent to \eqref{eq: main corollary general}.
\end{proof}

\section{Proofs for Section \ref{sec: estimation}}

In this section, we use $C$ to denote a strictly positive constant, which value may change from place to place.
Also, we use $\Ep_n[\cdot]$ to denote the average over index $i=1,\dots,n$; for example, $\Ep_n[X_i]=n^{-1}\sum_{i=1}^n X_i$.

\label{sec: risk bound proofs}

\begin{proof}[Proof of Lemma \ref{lem: asymptotic problem}]
Observe that if $D\hat{g}^u(x)\geq 0$ for all $x\in[0,1]$, then $\hat{g}^c$ coincides with $\hat{g}^u$, so that to prove (\ref{eq: asymptotic coincidence}), it suffices to show that
\begin{equation}\label{eq: lemma 3 reduction 1}
\Pr\Big(D\hat{g}^u(x)\geq 0\text{ for all }x\in[0,1]\Big)\to 1\text{ as }n\to\infty.
\end{equation}
In turn, (\ref{eq: lemma 3 reduction 1}) follows if
\begin{equation}\label{eq: lemma 3 reduction 2}
\sup_{x\in[0,1]}|D\hat{g}^u(x) - D g(x)|=o_p(1)
\end{equation}
since $D g(x)\geq c_g$ for all $x\in[0,1]$ and some $c_g>0$.

To prove (\ref{eq: lemma 3 reduction 2}), define a function $\hat{m}\in L^2[0,1]$ by 
\begin{equation}
\hat{m}(w)=q(w)'\Ep_n[q(W_i)Y_i],\qquad w\in[0,1],\label{eq: m hat}
\end{equation}
and an operator $\hat{T}:L^2[0,1]\to L^2[0,1]$ by
$$
(\hat{T}h)(w)=q(w)'\Ep_n[q(W_i)p(X_i)']\Ep[p(U)h(U)],\quad w\in[0,1],\ h\in L^2[0,1].
$$
Throughout the proof, we assume that the events
\begin{align}
&\|\Ep_n[q(W_i)p(X_i)']-\Ep[q(W)p(X)']\|\leq C(\xi_n^2\log n/n)^{1/2},\label{eq: matrix bound 1}\\
&\|\Ep_n[q(W_i)q(W_i)']-\Ep[q(W)q(W)']\|\leq C(\xi_n^2\log n/n)^{1/2},\label{eq: matrix bound 2}\\
&\|\Ep_n[q(W_i) g_n (X_i)] - \Ep[q(W)g_n(X)]\|\leq C(J/(\alpha n))^{1/2},\label{eq: matrix bound 3}\\
&\|\hat{m}-m\|_2\leq C((J/(\alpha n))^{1/2}+\tau_n^{-1}J^{-s})\label{eq: m bound prob}
\end{align}
hold for some sufficiently large constant $0<C<\infty$. It follows from Markov's inequality and Lemmas \ref{lem: m estimation} and \ref{lem: matrix LLN} that all four events hold jointly with probability at least $1-\alpha-n^{-1}$ since the constant $C$ is large enough.

Next, we derive a bound on $\|\hat{g}^u-g_n\|_2$. By the definition of $\tau_n$,
\begin{align*}
\|\hat{g}^u - g_n\|_2 
&\leq \tau_n \|T(\hat{g}^u - g_n)\|_2\\
&\leq \tau_n \|T(\hat{g}^u - g)\|_2 + \tau_n\|T(g - g_n)\|_2\leq \tau_n \|T(\hat{g}^u - g)\|_2 + C_g K^{-s}
\end{align*}
where the second inequality follows from the triangle inequality, and the third inequality from Assumption \ref{as: g approximation}(iii). Next, since $m=T g$,
$$
\|T(\hat{g}^u-g)\|_2\leq \|(T-T_n)\hat{g}^u\|_2+\|(T_n-\hat{T})\hat{g}^u\|_2+\|\hat{T}\hat{g}^u-\hat{m}\|_2+\|\hat{m}-m\|_2
$$
by the triangle inequality. The bound on $\|\hat{m}-m\|_2$ is given in \eqref{eq: m bound prob}.
Also, since $\|\hat{g}^u\|_2\leq C_b$ by construction, 
$$
\|(T-T_n)\hat{g}^u\|_2\leq C_b C_a\tau_n^{-1}K^{-s}
$$
by Assumption \ref{as: T approximation}(ii). In addition, by the triangle inequality,
\begin{align*}
\|(T_n - \hat{T})\hat{g}^u\|_2
&\leq \|(T_n - \hat{T})(\hat{g}^u - g_n)\|_2 + \|(T_n - \hat{T})g_n\|_2\\
&\leq \|T_n - \hat{T}\|_2 \|\hat{g}^u - g_n\|_2 + \|(T_n - \hat{T})g_n\|_2.
\end{align*}
Moreover,
$$
\|T_n - \hat{T}\|_2 = \|\Ep_n[q(W_i)p(X_i)']-\Ep[q(W)p(X)']\|\leq C(\xi_n^2\log n/n)^{1/2}
$$
by (\ref{eq: matrix bound 1}), and
$$
 \|(T_n - \hat{T})g_n\|_2 = \|\Ep_n[q(W_i) g_n (X_i)] - \Ep[q(W)g_n(X)]\|\leq C(J/(\alpha n))^{1/2}
$$
by (\ref{eq: matrix bound 3}).

Further, by Assumption \ref{as: density}(iii), all eigenvalues of $\Ep[q(W)q(W)']$ are bounded from below by $c_w$ and from above by $C_w$, and so it follows from (\ref{eq: matrix bound 2}) that for large $n$, all eigenvalues of $Q_n:=\Ep_n[q(W_i)q(W_i)']$ are bounded below from zero and from above. Therefore,
\begin{align*}
\|\hat{T}\hat{g}^u-\hat{m}\|_2&=\|\Ep_n[q(W_i)(p(X_i)'\hat{\beta}^u-Y_i)]\|\\
&\leq C\|\Ep_n[(Y_i-p(X_i)'\hat{\beta}^u)q(W_i)']Q_n^{-1}\Ep_n[q(W_i)(Y_i-p(X_i)'\hat{\beta}^u)]\|^{1/2}\\
&\leq C\|\Ep_n[(Y_i-p(X_i)'\beta_n)q(W_i)']Q_n^{-1}\Ep_n[q(W_i)(Y_i-p(X_i)'\beta_n)]\|^{1/2}\\
&\leq C\|\Ep_n[q(W_i)(p(X_i)'\beta_n-Y_i)]\|
\end{align*}
by optimality of $\hat{\beta}^u$.
Moreover,
\begin{align*}
\|\Ep_n[q(W_i)(p(X_i)'\beta_n-Y_i)]\|
&\leq \|(\hat{T}-T_n)g_n\|_2+\|(T_n - T) g_n\|_2\\
&\quad+\|T(g_n-g)\|_2+\|m-\hat{m}\|_2
\end{align*}
by the triangle inequality. The terms $\|(\hat{T}-T_n)g_n\|_2$ and $\|m-\hat{m}\|_2$ have been bounded above. Also, by Assumptions \ref{as: T approximation}(ii) and \ref{as: g approximation}(iii),
$$
\|(T_n - T)g_n\|_2\leq C\tau_n^{-1}K^{-s},\qquad \|T(g-g_n)\|_2\leq C_g \tau_n^{-1}K^{-s}.
$$
Combining the inequalities above shows that the inequality
\begin{equation}\label{eq: main bound cond exp}
\|\hat{g}^u-g_n\|_2\leq C\Big(\tau_n(J/(\alpha n))^{1/2}+K^{-s} + \tau_n (\xi_n^2\log n/n)^{1/2}\|\hat{g} - g_n\|_2\Big)
\end{equation}
holds with probability at least $1-\alpha-n^{-c}$. Since $\tau_n^2\xi_n^2\log n/n\to 0$, it follows that with the same probability,
$$
\|\hat{\beta}^u - \beta_n\| = \|\hat{g}^u - g_n\|_2 \leq C\Big(\tau_n(J/(\alpha n))^{1/2}+K^{-s}\Big),
$$
and so by the triangle inequality,
\begin{align*}
|D\hat{g}^u(x) - D g(x)|
&\leq |D\hat{g}^u(x) - D g_n(x)| + |D g_n(x) - D g(x)|\\
&\leq C\sup_{x\in[0,1]}\|D p(x)\|(\tau_n(K/(\alpha n))^{1/2}+K^{-s}) + o(1)
\end{align*}
uniformly over $x\in[0,1]$ since $J\leq C_J K$ by Assumption \ref{as: J-K}.  Since by the conditions of the lemma, $\sup_{x\in[0,1]}\|D p(x)\|(\tau_n(K/n)^{1/2}+K^{-s})\to 0$, (\ref{eq: lemma 3 reduction 2}) follows by taking $\alpha=\alpha_n\to 0$ slowly enough. This completes the proof of the lemma.
\end{proof}



 		%


\begin{proof}[Proof of Theorem \ref{thm: risk bounds}]
Consider the event that inequalities (\ref{eq: matrix bound 1})-(\ref{eq: m bound prob}) hold for some sufficiently large constant $C$. As in the proof of Lemma \ref{lem: asymptotic problem}, this events occurs with probability at least $1-\alpha-n^{-1}$. Also, applying the same arguments as those in the proof of Lemma \ref{lem: asymptotic problem} with $\hat{g}^c$ replacing $\hat{g}^u$ and using the bound
$$
\|(T_n - \hat{T})\hat{g}^c\|_2\leq \|T_n - \hat{T}\|_2\|\hat{g}^c\|_2 \leq C_b \|T_n - \hat{T}\|_2
$$
instead of the bound for $\|(T_n - \hat{T})\hat{g}^u\|_2$ used in the proof of Lemma \ref{lem: asymptotic problem}, it follows that on this event,
\begin{equation}\label{eq: imp ineq thm 2}
\|T(\hat{g}^c - g_n)\|_2\leq C\Big((K/(\alpha n))^{1/2} + (\xi_n^2\log n/n)^{1/2} + \tau_n^{-1} K^{-s}\Big).
\end{equation}
Further,
$$
\|\hat{g}^c-g_n\|_{2,t}\leq \delta + \tau_{n,t}\Big(\frac{\|D g_n\|_{\infty}}{\delta}\Big)\|T(\hat{g}^c-g_n)\|_2
$$
since $\hat{g}^c$ is increasing (indeed, if $\|\hat{g}^c-g\|_{2,t}\leq \delta$, the bound is trivial; otherwise, apply the definition of $\tau_{n,t}$ to the function $(\hat{g}^c-g_n)/\|\hat{g}^c-g_n\|_{2,t}$ and use the inequality
$\tau_{n,t}(\|D g_n\|_{\infty}/\|\hat{g}^c-g_n\|_{2,t})\leq \tau_{n,t}(\|D g_n\|_\infty/\delta)$). 
Finally, by the triangle inequality,
$$
\|\hat{g}^c - g\|_{2,t} \leq \|\hat{g}^c - g_n\|_{2,t} + \|g_n - g\|_{2,t} \leq \|\hat{g}^c - g_n\|_{2,t} + C_g K^{-s}.
$$
Combining these inequalities gives the asserted claim (\ref{eq: main rate 2}). 

To prove (\ref{eq: main rate 3}), observe that combining (\ref{eq: imp ineq thm 2}) and Assumption \ref{as: g approximation}(iii) and applying the triangle inequality shows that with probability at least $1-\alpha - n^{-1}$,
$$
\|T(\hat{g}^c - g)\|_2\leq C\Big((K/(\alpha n))^{1/2} + (\xi_n^2\log n/n)^{1/2} + \tau_n^{-1} K^{-s}\Big),
$$
which, by the same argument as that used to prove (\ref{eq: main rate 2}), gives
\begin{equation}\label{eq: imp ineq thm 2 - 2}
\|\hat{g}^c-g\|_{2,t}\leq C \Big\{\delta + \tau\Big(\frac{\|D g\|_{\infty}}{\delta}\Big)\Big(\frac{K}{\alpha n}+\frac{\xi_n^2\log n}{n}\Big)^{1/2} + K^{-s}\Big\}.
\end{equation}
The asserted claim (\ref{eq: main rate 3}) now follows by applying (\ref{eq: main rate 2}) with $\delta=0$ and (\ref{eq: imp ineq thm 2 - 2}) with $\delta=\|D g\|_\infty/c_\tau$ and using Corollary \ref{co:bound_tau} to bound $\tau(c_\tau)$. This completes the proof of the theorem.
\end{proof}

\begin{lemma}\label{lem: m estimation}
Under conditions of Theorem \ref{thm: risk bounds}, $\|\hat{m}-m\|_2\leq C((J/(\alpha n))^{1/2}+\tau_n^{-1}J^{-s})$ with probability at least $1-\alpha$ where $\hat{m}$ is defined in \eqref{eq: m hat}.
\end{lemma}
\begin{proof}
Using the triangle inequality and an elementary inequality $(a+b)^2\leq 2a^2+2b^2$ for all $a,b\geq 0$,
$$
\|\Ep_n[q(W_i)Y_i]-E[q(W)g(X)]\|^2\leq 2\|\Ep_n[q(W_i)\varepsilon_i]\|^2+2\|\Ep_n[q(W_i)g(X_i)]-\Ep[q(W)g(X)]\|^2.
$$
To bound the first term on the right-hand side of this inequality, we have
\begin{align*}
E\left[\|\Ep_n[q(W_i)\varepsilon_i]\|^2\right]
=n^{-1}\Ep[\|q(W)\varepsilon\|^2]\leq (C_B/n)\Ep[\|q(W)\|^2]\leq C J/n
\end{align*}
where the first and the second inequalities follow from Assumptions \ref{as: moments} and \ref{as: density}, respectively. Similarly,
\begin{align*}
\Ep\left[\|\Ep_n[q(W_i)g(X_i)]-\Ep[q(W)g(X)]\|^2\right]
&\leq n^{-1}\Ep[\|q(W)g(X)\|^2]\\
&\leq (C_B/n)\Ep[\|q(W)\|^2]\leq C J/n
\end{align*}
by Assumption \ref{as: moments}. Therefore, denoting $\bar{m}_n(w):=q(w)'\Ep[q(W)g(X)]$ for all $w\in[0,1]$, we obtain
$$
\Ep[\|\hat{m}-\bar{m}_n\|_2^2]\leq C J/n,
$$
and so by Markov's inequality, $\|\hat{m}-\bar{m}_n\|_2\leq C(J/(\alpha n))^{1/2}$ with probability at least $1-\alpha$.
Further, using $\gamma_n\in\mathbb{R}^J$ from Assumption \ref{as: m approximation}, so that $m_n(w)=q(w)'\gamma_n$ for all $w\in[0,1]$, and denoting $r_n(w):=m(w)-m_n(w)$ for all $w\in[0,1]$, we obtain
\begin{align*}
\bar{m}_n(w)
&=q(w)'\int_0^1\int_0^1q(t)g(x)f_{X,W}(x,t)d x d t\\
&=q(w)'\int_0^1q(t)m(t)d t=q(w)'\int_0^1q(t)(q(t)'\gamma_n+r_n(t))d t\\
&=q(w)'\gamma_n+q(w)'\int_0^1q(t)r_n(t)d t=m(w)-r_n(w)+q(w)'\int_0^1q(t)r_n(t)d t.
\end{align*}
Hence, by the triangle inequality,
$$
\|\bar{m}_n-m\|_2\leq \|r_n\|_2+\left\Vert\int_0^1q(t)r_n(t)d t\right\Vert\leq 2\|r_n\|_2\leq 2C_m\tau_n^{-1}J^{-s}
$$
by Bessel's inequality and Assumption \ref{as: m approximation}. Applying the triangle inequality one more time, we obtain
$$
\|\hat{m}-m\|_2\leq \|\hat{m}-\bar{m}_n\|+\|\bar{m}_n-m\|_2\leq C((J/(\alpha n))^{1/2}+\tau_n^{-1}J^{-s})
$$
with probability at least $1-\alpha$.
This completes the proof of the lemma.
\end{proof}

\begin{proof}[Proof of Corollary \ref{cor: fast rate under local-to-constant}]
The corollary follows immediately from Theorem \ref{thm: risk bounds}.
\end{proof}

	%
	%
	%
	%

\section{Proofs for Section \ref{sec: identification}}

Let $\mathcal{M}_\uparrow$ be the set of all functions in $\mathcal{M}$ that are increasing but not constant. Similarly, let $\mathcal{M}_\downarrow$ be the set of all functions in $\mathcal{M}$ that are decreasing but not constant, and let $\mathcal{M}_\rightarrow$ be the set of all constant functions in $\mathcal{M}$.

\begin{proof}[Proof of Theorem~\ref{lem: id of sign of Dg}]
Assume that $g$ is increasing but not constant, that is, $g\in\mathcal{M}_\uparrow$. Define $M(w) := \Ep[Y|W=w]$, $w\in[0,1]$. Below we show that $M\in\mathcal{M}_\uparrow$. To prove it, observe that, as in the proof of Lemma \ref{lem: crucial lemma}, integration by parts gives
$$
M(w)=g(1) - \int_0^1 D g(x)F_{X|W}(x|w)d x,
$$
and so Assumption \ref{as: monotone iv} implies that $M$ is increasing. Let us show that $M$ is not constant. To this end, note that
$$
M(w_2) - M(w_1) = \int_0^1 D g(x)(F_{X|W}(x|w_1) - F_{X|W}(x|w_2))d x.
$$
Since $g$ is not constant and is continuously differentiable, there exists $\bar{x}\in(0,1)$ such that $D g(\bar{x})>0$. Also, since $0\leq x_1<x_2\leq 1$ (the constants $x_1$ and $x_2$ appear in Assumption \ref{as: monotone iv}), we have $\bar{x}\in(0,x_2)$ or $\bar{x}\in (x_1,1)$. In the first case,
$$
M(w_2) - M(w_1) \geq  \int_0^{x_2} (C_F - 1)D g(x) F_{X|W}(x|w_2)d x > 0.
$$
In the second case,
$$
M(w_2) - M(w_1) \geq  \int_{x_1}^{1} (C_F - 1)D g(x) (1 - F_{X|W}(x|w_1))d x > 0.
$$
Thus, $M$ is not constant, and so $M\in\mathcal{M}_\uparrow$. Similarly, one can show that if $g\in\mathcal{M}_\downarrow$, then $M\in\mathcal{M}_\downarrow$, and if $g\in\mathcal{M}_\rightarrow$, then $M\in\mathcal{M}_\rightarrow$. However, the distribution of the triple $(Y,X,W)$ uniquely determines whether $M\in\mathcal{M}_\uparrow$, $\mathcal{M}_\downarrow$, or $\mathcal{M}_\rightarrow$, and so it also uniquely determines whether $g\in\mathcal{M}_\uparrow$, $\mathcal{M}_\downarrow$, or $\mathcal{M}_\rightarrow$  This completes the proof.
%
\end{proof}

\begin{proof}[Proof of Theorem~\ref{lem: id set bounds}]
Suppose $g'$ and $g''$ are observationally equivalent. Then $\|T(g'-g'')\|_2=0$. On the other hand, since $0\leq \|h\|_{2,t}+\bar{C}\|T\|_2\|h\|_2 < \|g'-g''\|_{2,t}$, there exists $\alpha\in(0,1)$ such that $\|h\|_{2,t}+\bar{C}\|T\|_2\|h\|_2 \leq \alpha\|g'-g''\|_{2,t}$. Therefore, by Lemma~\ref{lem: general bound on tau}, $\|T(g'-g'')\|_2 \geq \|g'-g''\|_{2,t}(1-\alpha)/\bar{C}>0$, which is a contradiction. This completes the proof of the theorem.
\end{proof}

\section{Proofs for Section \ref{sec: testing}}
\begin{proof}[Proof of Theorem \ref{thm: size control}]
In this proof, $c$ and $C$ are understood as sufficiently small and large constants, respectively, whose values may change at each appearance but can be chosen to depend only on $c_W,C_W,c_h,C_H,c_F,C_F,c_\epsilon,C_\epsilon$, and the kernel $K$.

To prove the asserted claims, we apply Corollary 3.1, Case (E.3), from CCK conditional on $\mathcal{W}_n=\{W_1,\dots,W_n\}$. Under $H_0$,
\begin{equation}\label{eq: trivial bound null}
T\leq \max_{(x,w,h)\in\mathcal{X}_n\times\mathcal{W}_n\times\mathcal{B}_n}\frac{\sum_{i=1}^nk_{i,h}(w)(1\{X_i\leq x\}-F_{X|W}(x|W_i))}{\left(\sum_{i=1}^nk_{i,h}(w)^2\right)^{1/2}}=:T_0
\end{equation}
with equality if the functions $w\mapsto F_{X|W}(x|w)$ are constant for all $x\in(0,1)$. Using the notation of CCK,
$$
T_0=\max_{1\leq j\leq p}\frac{1}{\sqrt{n}}\sum_{i=1}^n x_{i j}
$$
where $p=|\mathcal{X}_n\times\mathcal{W}_n\times\mathcal{B}_n|$, the number of elements in the set $\mathcal{X}_n\times\mathcal{W}_n\times\mathcal{B}_n$, $x_{i j}=z_{i j}\varepsilon_{i j}$ with $z_{i j}$ having the form $\sqrt{n}k_{i,h}(w)/(\sum_{i=1}^nk_{i,h}(w)^2)^{1/2}$, and $\varepsilon_{i j}$ having the form $1\{X_i\leq x\}-F_{X|W}(x|W_i)$ for some $(x,w,h)\in\mathcal{X}_n\times\mathcal{W}_n\times\mathcal{B}_n$. The dimension $p$ satisfies $\log p\leq C\log n$. Also, $n^{-1}\sum_{i=1}^n z_{i j}^2=1$. Further, since $0\leq 1\{X_i\leq x\}\leq 1$, we have $|\varepsilon_{i j}|\leq 1$, and so $\Ep[\exp(|\varepsilon_{i j}|/2)|\mathcal{W}_n]\leq 2$. In addition, $\Ep[\varepsilon_{i j}^2|\mathcal{W}_n]\geq c_\epsilon(1-C_\epsilon)>0$ by Assumption \ref{as: conditional density}. Thus, $T_0$ satisfies the conditions of Case (E.3) in CCK with a sequence of constants $B_n$ as long as $|z_{i j}|\leq B_n$ for all $j=1,\dots,p$. In turn, Proposition B.2 in \cite{Chetverikov:2012fk} shows that under Assumptions \ref{as: density}, \ref{as: kernel}, and \ref{as: bandwidth values}, with probability at least $1-C n^{-c}$, $z_{i j}\leq C/\sqrt{h_{\min}}=:B_n$ uniformly over all $j=1,\dots,p$ (Proposition B.2 in \cite{Chetverikov:2012fk} is stated with ``w.p.a.1'' replacing ``$1-Cn^{-c}$''; however, inspecting the proof of Proposition B.2 (and supporting Lemma H.1) shows that the result applies with ``$1-Cn^{-c}$'' instead of ``w.p.a.1''). Let $\mathcal{B}_{1,n}$ denote the event that $|z_{i j}|\leq C/\sqrt{h_{\min}}=B_n$ for all $j=1,\dots,p$. As we just established, $\Pr(\mathcal{B}_{1,n})\geq 1-Cn^{-c}$. Since $(\log n)^7/(nh_{\min})\leq C_hn^{-c_h}$ by Assumption \ref{as: bandwidth values}, we have that $B_n^2(\log n)^7/n\leq Cn^{-c}$, and so condition (i) of Corollary 3.1 in CCK is satisfied on the event $\mathcal{B}_{1,n}$.

Let $\mathcal{B}_{2,n}$ denote the event that
$$
\Pr\left(\max_{(x,w)\in\mathcal{X}_n\times\mathcal{W}_n}|\hat{F}_{X|W}(x|w)-F_{X|W}(x|w)|>C_Fn^{-c_F}|\{\mathcal{W}_n\}\right)\leq C_Fn^{-c_F}.
$$
By Assumption \ref{as: cond cdf estimator}, $\Pr(\mathcal{B}_{2,n})\geq 1-C_Fn^{-c_F}$. We apply Corollary 3.1 from CCK conditional on $\mathcal{W}_n$ on the event $\mathcal{B}_{1,n}\cap\mathcal{B}_{2,n}$.
For this, we need to show that on the event $\mathcal{B}_{2,n}$, $\zeta_{1,n}\sqrt{\log n}+\zeta_{2,n}\leq Cn^{-c}$ where $\zeta_{1,n}$ and $\zeta_{2,n}$ are positive sequences such that
\begin{equation}\label{eq: CCK verification main}
\Pr\left(\Pr_e(|T^b-T^b_0|>\zeta_{1,n})>\zeta_{2,n}|\mathcal{W}_n\right)<\zeta_{2,n}
\end{equation}
where
$$
T^b_0:=\max_{(x,w,h)\in\mathcal{X}_n\times\mathcal{W}_n\times\mathcal{B}_n}\frac{\sum_{i=1}^ne_i\left(k_{i,h}(w)(1\{X_i\leq x\}-F_{X|W}(x|W_i))\right)}{\left(\sum_{i=1}^nk_{i,h}(w)^2\right)^{1/2}}
$$
and where $\Pr_e(\cdot)$ denotes the probability distribution with respect to the distribution of $e_1,\dots,e_n$ and keeping everything else fixed. To find such sequences $\zeta_{1,n}$ and $\zeta_{2,n}$, note that $\zeta_{1,n}\sqrt{\log n}+\zeta_{2,n}\leq Cn^{-c}$ follows from $\zeta_{1,n}+\zeta_{2,n}\leq Cn^{-c}$ (with different constants $c,C>0$), so that it suffices to verify the latter condition. Also, 
$$
|T^b-T^b_0|\leq \max_{(x,w,h)\in\mathcal{X}_n\times\mathcal{W}_n\times\mathcal{B}_n}\left|\frac{\sum_{i=1}^n e_i k_{i,h}(w)(\hat{F}_{X|W}(x|W_i)-F_{X|W}(x|W_i))}{\left(\sum_{i=1}^nk_{i,h}(w)^2\right)^{1/2}}\right|.
$$
For fixed $W_1,\dots,W_n$ and $X_1,\dots,X_n$, the random variables under the modulus on the right-hand side of this inequality are normal with zero mean and variance bounded from above by $\max_{(x,w)\in\mathcal{X}_n\times\mathcal{W}_n}|\hat{F}_{X|W}(x|w)-F_{X|W}(x|w)|^2$. Therefore,
$$
\Pr_e\left(|T^b-T^b_0|>C\sqrt{\log n}\max_{(x,w)\in\mathcal{X}_n\times\mathcal{W}_n}\left|\hat{F}_{X|W}(x|w)-F_{X|W}(x|w)\right|\right)\leq Cn^{-c}.
$$
Hence, on the event that 
$$
\max_{(x,w)\in\mathcal{X}_n\times\mathcal{W}_n}\left|\hat{F}_{X|W}(x|w)-F_{X|W}(x|w)\right|\leq C_Fn^{-c_F},
$$
whose conditional probability given $\mathcal{W}_n$ on $\mathcal{B}_{2,n}$ is at least $1-C_Fn^{-c_F}$ by the definition of $\mathcal{B}_{2,n}$,
$$
\Pr_e\left(|T^b-T^b_0|>Cn^{-c}\right)\leq Cn^{-c}
$$
implying that (\ref{eq: CCK verification main}) holds for some $\zeta_{1,n}$ and $\zeta_{2,n}$ satisfying $\zeta_{1,n}+\zeta_{2,n}\leq Cn^{-c}$.

Thus, applying Corollary 3.1, Case (E.3), from CCK conditional on $\{W_1,\dots,W_n\}$ on the event $\mathcal{B}_{1,n}\cap\mathcal{B}_{2,n}$ gives
$$
\alpha-Cn^{-c}\leq \Pr(T_0>c(\alpha)|\mathcal{W}_n)\leq \alpha+Cn^{-c}.
$$
Since $\Pr(\mathcal{B}_{1,n}\cap\mathcal{B}_{2,n})\geq 1-Cn^{-c}$, integrating this inequality over the distribution of $\mathcal{W}_n=\{W_1,\dots,W_n\}$ gives (\ref{eq: nonconservative stoch dom}). Combining this inequality with (\ref{eq: trivial bound null}) gives (\ref{eq: size control stoch dom}). This completes the proof of the theorem.
\end{proof}

\begin{proof}[Proof of Theorem \ref{thm: consistency}]
Conditional on the data, the random variables
$$
T^b(x,w,h):=\frac{\sum_{i=1}^ne_i\left(k_{i,h}(w)(1\{X_i\leq x\}-\hat{F}_{X|W}(x|W_i))\right)}{\left(\sum_{i=1}^nk_{i,h}(w)^2\right)^{1/2}}
$$
for $(x,w,h)\in\mathcal{X}_n\times\mathcal{W}_n\times\mathcal{B}_n$ are normal with zero mean and variances bounded from above by
\begin{align*}
&\frac{\sum_{i=1}^n\left(k_{i,h}(w)(1\{X_i\leq x\}-\hat{F}_{X|W}(x|W_i))\right)^2}{\sum_{i=1}^nk_{i,h}(w)^2}\\
&\qquad \leq \max_{(x,w,h)\in \mathcal{X}_n\times\mathcal{W}_n\times\mathcal{B}_n}\max_{1\leq i\leq n}\left(1\{X_i\leq x\}-\hat{F}_{X|W}(x|W_i)\right)^2\leq (1+C_h)^2
\end{align*}
by Assumption \ref{as: cond cdf estimator}. Therefore, $c(\alpha)\leq C(\log n)^{1/2}$ for some constant $C>0$ since $c(\alpha)$ is the $(1-\alpha)$ conditional quantile of $T^b$ given the data, $T^b=\max_{(x,w,h)\in\mathcal{X}_n\times\mathcal{W}_n\times\mathcal{B}_n}T^b(x,w,h)$, and $p:=|\mathcal{X}_n\times\mathcal{W}_n\times\mathcal{B}_n|$, the number of elements of the set $\mathcal{X}_n\times\mathcal{W}_n\times\mathcal{B}_n$, satisfies $\log p\leq C\log n$ (with a possibly different constant $C>0$). Thus, the growth rate of the critical value $c(\alpha)$ satisfies the same upper bound $(\log n)^{1/2}$ as if we were testing monotonicity of one particular regression function $w\mapsto \Ep[1\{X\leq x_0\}|W=w]$ with $\mathcal{X}_n$ replaced by $x_0$ for some $x_0\in(0,1)$ in the definition of $T$ and $T^b$. Hence, the asserted claim follows from the same arguments as those given in the proof of Theorem 4.2 in \cite{Chetverikov:2012fk}. This completes the proof of the theorem.
\end{proof}

\section{Technical tools}
In this section, we provide a set of technical results that are used to prove the statements from the main text.

\begin{lemma}\label{lem: Chebyshev's association inequality}
Let $W$ be a random variable with the density function bounded below from zero on its support $[0,1]$, and let $M:[0,1]\to \mathbb{R}$ be a monotone function. If $M$ is constant, then $cov(W,M(W))=0$. If $M$ is increasing in the sense that there exist $0<w_1<w_2<1$ such that $M(w_1)<M(w_2)$, then $cov(W,M(W))>0$.
\end{lemma}
\begin{proof}
The first claim is trivial. The second claim follows by introducing an independent copy $W'$ of the random variable $W$, and rearranging the inequality
$$
\Ep[(M(W)-M(W'))(W-W')]>0,
$$
which holds for increasing $M$ since $(M(W)-M(W'))(W-W')\geq 0$ almost surely and $(M(W)-M(W'))(W-W')> 0$ with strictly positive probability. This completes the proof of the lemma.
\end{proof}

\begin{lemma}\label{lem: degenerate basis}
For any orthonormal basis $\{h_j,j\geq 1\}$ in $L^2[0,1]$, any $0\leq x_1<x_2\leq 1$, and any $\alpha>0$,
$$
\|h_j\|_{2,t}=\left(\int_{x_1}^{x_2}h_j^2(x)d x\right)^{1/2}>j^{-1/2-\alpha}
$$
for infinitely many $j$.
\end{lemma}
\begin{proof}
Fix $M\in\mathbb{N}$ and consider any partition $x_1=t_0<t_1<\dots<t_M=x_2$. Further, fix $m=1,\dots,M$ and consider the function
$$
h(x)=\begin{cases}
\frac{1}{\sqrt{t_{m}-t_{m-1}}} & x\in(t_{m-1},t_{m}],\\
0, & x\notin(t_{m-1},t_{m}].
\end{cases}
$$
Note that $\|h\|_2=1$, so that
$$
h=\sum_{j=1}^\infty \beta_j h_j\text{ in }L^2[0,1],\,\,\,\beta_j:=\frac{\int_{t_{m-1}}^{t_m}h_j(x)d x}{(t_m-t_{m-1})^{1/2}},\,\,\,\text{ and }\sum_{j=1}^\infty \beta_j^2=1.
$$
Therefore, by the Cauchy-Schwarz inequality,
\begin{align*}
1&=\sum_{j=1}^\infty\beta_j^2=\frac{1}{t_m-t_{m-1}}\sum_{j=1}^\infty\left(\int_{t_{m-1}}^{t_m}h_j(x)d x\right)^2\leq \sum_{j=1}^\infty \int_{t_{m-1}}^{t_m}(h_j(x))^2d x.
\end{align*}
Hence, $\sum_{j=1}^\infty \|h_j\|^2_{2,t}\geq M$. Since $M$ is arbitrary, we obtain $\sum_{j=1}^\infty\|h_j\|^2_{2,t}=\infty$, and so for any $J$, there exists $j>J$ such that $\|h_j\|_{2,t}>j^{-1/2-\alpha}$. Otherwise, we would have $\sum_{j=1}^\infty\|h_j\|^2_{2,t}<\infty$. This completes the proof of the lemma.
\end{proof}

\begin{lemma}\label{lem: normal density}
Let $(X,W)$ be a pair of random variables defined as in Example \ref{ex: normal density}. Then Assumptions \ref{as: monotone iv} and \ref{as: density} of Section \ref{sec: MIVE} are satisfied if $0<x_1<x_2<1$ and $0<w_1<w_2<1$.
\end{lemma}
\begin{proof}
As noted in Example \ref{ex: normal density}, we have
$$
X=\Phi(\rho\Phi^{-1}(W)+(1-\rho^2)^{1/2}U)
$$
where $\Phi(x)$ is the distribution function of a $N(0,1)$ random variable and $U$ is a $N(0,1)$ random variable that is independent of $W$. Therefore, the conditional distribution function of $X$ given $W$ is
$$
F_{X|W}(x|w):=\Phi\left(\frac{\Phi^{-1}(x)-\rho\Phi^{-1}(w)}{\sqrt{1-\rho^2}}\right).
$$
Since the function $w\mapsto F_{X|W}(x|w)$ is decreasing for all $x\in(0,1)$, condition (\ref{eq: fosd}) of Assumption \ref{as: monotone iv} follows. Further, to prove condition \eqref{eq: lower cdf bound} of Assumption \ref{as: monotone iv}, it suffices to show that
\begin{equation}\label{eq: need to show log derivative}
	\frac{\partial \log F_{X|W}(x|w)}{\partial w}\leq c_F
\end{equation}
for some constant $c_F<0$, all $x\in(0,x_2)$, and all $w\in(w_1,w_2)$ because, for every $x\in(0,x_2)$ and $w\in(w_1,w_2)$, there exists $\bar{w}\in(w_1,w_2)$ such that
$$ \log\left( \frac{F_{X|W}(x|w_1)}{ F_{X|W}(x|w_2)} \right) = \log F_{X|W}(x|w_1) - \log F_{X|W}(x|w_2) = -(w_2 - w_1) \frac{\partial \log F_{X|W}(x|\bar{w})}{\partial w}.$$
Therefore, $\partial \log F_{X|W}(x|w) / \partial w \leq c_F <0 $ for all $x\in(0,x_2)$ and $w\in(w_1,w_2)$ implies
$$\frac{F_{X|W}(x|w_1)}{ F_{X|W}(x|w_2)} \geq e^{-c_F (w_2 - w_1)} > 1 $$
for all $x\in(0,x_2)$. To show \eqref{eq: need to show log derivative}, observe that
\begin{equation}\label{eq: intermediate bound to log derivative}
	\frac{\partial \log F_{X|W}(x|w)}{\partial w}=-\frac{\rho}{\sqrt{1-\rho^2}}\frac{\phi(y)}{\Phi(y)}\frac{1}{\phi(\Phi^{-1}(w))}\leq -\frac{\sqrt{2\pi}\rho}{\sqrt{1-\rho^2}}\frac{\phi(y)}{\Phi(y)}
\end{equation}
where $y:=(\Phi^{-1}(x)-\rho\Phi^{-1}(w))/(1-\rho^2)^{1/2}$. Thus, (\ref{eq: need to show log derivative}) holds for some $c_F<0$ and all $x\in(0,x_2)$ and $w\in (w_1,w_2)$ such that $\Phi^{-1}(x)\geq \rho\Phi^{-1}(w)$ since $x_2<1$ and $0<w_1<w_2<1$. On the other hand, when $\Phi^{-1}(x)<\rho \Phi^{-1}(w)$, so that $y<0$, it follows from Proposition 2.5 in \cite{dudley2014} that $\phi(y)/\Phi(y)\geq (2/\pi)^{1/2}$, and so (\ref{eq: intermediate bound to log derivative}) implies that
$$
\frac{\partial \log F_{X|W}(x|w)}{\partial w}\leq -\frac{2\rho}{\sqrt{1-\rho^2}}
$$
in this case. Hence, condition (\ref{eq: lower cdf bound}) of Assumption \ref{as: monotone iv} is satisfied. Similar argument also shows that condition (\ref{eq: upper cdf bound}) of Assumption \ref{as: monotone iv} is satisfied as well.

We next consider Assumption \ref{as: density}. Since $W$ is distributed uniformly on $[0,1]$ (remember that $\tilde{W}\sim N(0,1)$ and $W=\Phi(\tilde{W})$), condition (iii) of Assumption \ref{as: density} is satisfied. Further, differentiating $x\mapsto F_{X|W}(x|w)$ gives
\begin{equation}\label{eq: normal conditional density}
f_{X|W}(x|w):=\frac{1}{\sqrt{1-\rho^2}}\phi\left(\frac{\Phi^{-1}(x)-\rho\Phi^{-1}(w)}{\sqrt{1-\rho^2}}\right)\frac{1}{\phi(\Phi^{-1}(x))}.
\end{equation}
Since $0<x_1<x_2<1$ and $0<w_1<w_2<1$, condition (ii) of Assumption \ref{as: density} is satisfied as well. Finally, to prove condition (i) of Assumption \ref{as: density}, note that since $f_{W}(w)=1$ for all $w\in[0,1]$, (\ref{eq: normal conditional density}) combined with the change of variables formula with $x=\Phi(\tilde{x})$ and $w=\Phi(\tilde{w})$ give
\begin{align*}
&(1-\rho^2)\int_0^1\int_0^1f_{X,W}^2(x,w)dxdw=(1-\rho^2)\int_0^1\int_0^1f_{X|W}^2(x|w)dxdw\\
&\qquad=\int_{-\infty}^{+\infty}\int_{-\infty}^{+\infty}\phi^2\left(\frac{\tilde{x}-\rho\tilde{w}}{\sqrt{1-\rho^2}}\right)\frac{\phi(\tilde{w})}{\phi(\tilde{x})}d\tilde{x}d\tilde{w}\\
&\qquad=\frac{1}{2\pi}\int_{-\infty}^{+\infty}\int_{-\infty}^{+\infty}\exp\left[\left(\frac{1}{2}-\frac{1}{1-\rho^2}\right)\tilde{x}^2+\frac{2\rho}{1-\rho^2}\tilde{x}\tilde{w}-\left(\frac{\rho^2}{1-\rho^2}+\frac{1}{2}\right)\tilde{w}^2\right]d\tilde{x}d\tilde{w}\\
&\qquad=\frac{1}{2\pi}\int_{-\infty}^{+\infty}\int_{-\infty}^{+\infty}\exp\left[-\frac{1+\rho^2}{2(1-\rho^2)}\left(\tilde{x}^2-\frac{4\rho}{1+\rho^2}\tilde{x}\tilde{w}+\tilde{w}^2\right)\right]d\tilde{x}d\tilde{w}.
\end{align*}
Since $4\rho/(1+\rho^2)<2$, the integral in the last line is finite implying that condition (i) of Assumption \ref{as: density} is satisfied. This completes the proof of the lemma.
\end{proof}

\begin{lemma}\label{lem: uniforms}
	Let $X=U_1 + U_2 W$ where $U_1, U_2, W$ are mutually independent, $U_1,U_2\sim U[0,1/2]$ and $W\sim U[0,1]$. Then Assumptions~\ref{as: monotone iv} and \ref{as: density} of Section \ref{sec: MIVE} are satisfied if $0<w_1< w_2<1$, $0<x_1 < x_2<1$, and $w_1 > w_2 - \sqrt{w_2/2}$.
\end{lemma}

\begin{proof}
	Since $X|W=w$ is a convolution of the random variables $U_1$ and $U_2w$,
	\begin{align*}
		f_{X|W}(x|w) &= \int_0^{1/2} f_{U_1}(x-u_2w) f_{U_2}(u_2)du_2\\
			&= 4 \int_0^{1/2} 1\left\{ 0 \leq x-u_2 w \leq \frac{1}{2} \right\}du_2\\
			&= 4 \int_0^{1/2} 1\left\{ \frac{x}{w}-\frac{1}{2w} \leq u_2 \leq \frac{x}{w} \right\}du_2\\
			&= \left\{ 	\begin{array}{cl} 	\frac{4x}{w}, 					& 	0\leq x< \frac{w}{2}\\
										 	2,								&	\frac{w}{2} \leq x < \frac{1}{2}\\
										 	\frac{2(1+w)}{w}-\frac{4x}{w},	&	\frac{1}{2} \leq x < \frac{1+w}{2}\\
										 	0, 								&	\frac{1+w}{2} \leq x \leq 1
						\end{array} \right.
	\end{align*}
	and, thus,
	\begin{align*}
		F_{X|W}(x|w) &= \left\{\begin{array}{cl} 	\frac{2x^2}{w}, 									& 	0\leq x< \frac{w}{2}\\
										 			2x - \frac{w}{2},								&	\frac{w}{2} \leq x < \frac{1}{2}\\
										 			1-\frac{2}{w}\left(x-\frac{1+w}{2}\right)^2,	&	\frac{1}{2} \leq x < \frac{1+w}{2}\\
										 			1, 												&	\frac{1+w}{2} \leq x \leq 1
						\end{array} \right..
	\end{align*}
	It is easy to check that $\partial F_{X|W}(x|w) /\partial w\leq 0$ for all $x,w\in[0,1]$ so that condition \eqref{eq: fosd} of Assumption~\ref{as: monotone iv} is satisfied. To check conditions \eqref{eq: lower cdf bound} and \eqref{eq: upper cdf bound}, we proceed as in Lemma~\ref{lem: normal density} and show $\partial \log F_{X|W}(x|w) / \partial w < 0$ uniformly for all $x\in [\underline{x}_2,\overline{x}_1]$ and $w\in(\tilde{w}_1,\tilde{w}_2)$. First, notice that, as required by Assumption~\ref{as: density}(iv), $[\underline{x}_k,\overline{x}_k] = [0,(1+\tilde{w}_k)/2]$, $k=1,2$. For $0\leq x< w/2$ and $w\in(\tilde{w}_1,\tilde{w}_2)$, 
	$$\frac{\partial F_{X|W}(x|w)}{\partial w} = \frac{-2x^2/w^2}{2x^2/w} = -\frac{1}{w} < -\frac{1}{\tilde{w}_1}<0, $$
	and, for $w/2 \leq x < 1/2$ and $w\in(\tilde{w}_1,\tilde{w}_2)$,
	$$\frac{\partial F_{X|W}(x|w)}{\partial w} = \frac{-1/2}{2x-w/2} < \frac{-1/2}{w-w/2} < -\frac{1}{\tilde{w}_1}<0.$$
	Therefore, \eqref{eq: lower cdf bound} holds uniformly over $x\in (\underline{x}_2,1/2)$ and \eqref{eq: upper cdf bound} uniformly over $x\in (x_1,1/2)$. Now, consider $1/2 \leq x < (1+\tilde{w}_1)/2$ and $w\in(\tilde{w}_1,\tilde{w}_2)$. Notice that, on this interval, $\partial (F_{X|W}(x|\tilde{w}_1)/F_{X|W}(x|\tilde{w}_2)) /\partial x \leq 0$ so that 
	$$\frac{F_{X|W}(x|\tilde{w}_1)}{F_{X|W}(x|\tilde{w}_2)} = \frac{1-\frac{2}{\tilde{w}_1}\left(x-\frac{1+\tilde{w}_1}{2}\right)^2}{1-\frac{2}{\tilde{w}_2}\left(x-\frac{1+\tilde{w}_2}{2}\right)^2} \geq \frac{1}{1-\frac{2}{\tilde{w}_2}\left(\frac{1+\tilde{w}_1}{2}-\frac{1+\tilde{w}_2}{2}\right)^2} = \frac{\tilde{w}_2}{\tilde{w}_2 - 2(\tilde{w}_1-\tilde{w}_2)^2} > 1,$$
	where the last inequality uses $\tilde{w}_1 > \tilde{w}_2 - \sqrt{\tilde{w}_2/2}$, and thus \eqref{eq: lower cdf bound} holds also uniformly over $1/2 \leq x < x_2$. Similarly,
	$$\frac{1-F_{X|W}(x|\tilde{w}_2)}{1-F_{X|W}(x|\tilde{w}_1)} = \frac{\frac{2}{\tilde{w}_2}\left(x-\frac{1+\tilde{w}_2}{2}\right)^2}{\frac{2}{\tilde{w}_1}\left(x-\frac{1+\tilde{w}_1}{2}\right)^2} \geq \frac{\frac{2}{\tilde{w}_2}\left(\frac{\tilde{w}_2}{2}\right)^2}{\frac{2}{\tilde{w}_1}\left(\frac{\tilde{w}_1}{2}\right)^2} = \frac{\tilde{w}_2}{\tilde{w}_1}>1$$
	so that \eqref{eq: upper cdf bound} also holds uniformly over $1/2 \leq x < \overline{x}_1$. Assumption~\ref{as: density}(i) trivially holds. Parts (ii) and (iii) of Assumption~\ref{as: density} hold for any $0<\tilde{x}_1<\tilde{x}_2\leq \overline{x}_1\leq 1$ and $0\leq w_1 < \tilde{w}_1<\tilde{w}_2 < w_2\leq 1$ with $[\underline{x}_k,\overline{x}_k] = [0,(1+\tilde{w}_k)/2]$, $k=1,2$.
\end{proof}

\begin{lemma}\label{lem: monotone approximation}
For any increasing function $h\in L^2[0,1]$, one can find a sequence of increasing continuously differentiable functions $h_k\in L^2[0,1]$, $k\geq1$, such that $\|h_k-h\|_2\to 0$ as $k\to \infty$.
\end{lemma}
\begin{proof}
Fix some increasing $h\in L^2[0,1]$. For $a>0$, consider the truncated function:
$$
\tilde{h}_a(x):=h(x)1\{|h(x)|\leq a\}+a 1\{h(x)>a\} - a 1\{h(x)<-a\}
$$
for all $x\in[0,1]$. Then $\|\tilde{h}_a-h\|_2\to 0$ as $a\to \infty$ by Lebesgue's dominated convergence theorem. Hence, by scaling and shifting $h$ if necessary, we can assume without loss of generality that $h(0)=0$ and $h(1)=1$.

To approximate $h$, set $h(x)=0$ for all $x\in\mathbb{R}\backslash [0,1]$ and for $\sigma>0$, consider the function
$$
h_\sigma(x):=\frac{1}{\sigma}\int_0^1 h(y)\phi\left(\frac{y-x}{\sigma}\right)d y = \frac{1}{\sigma}\int_{-\infty}^{\infty} h(y)\phi\left(\frac{y-x}{\sigma}\right)d y
$$
for $y\in\mathbb{R}$ where $\phi$ is the density function of a $N(0,1)$ random variable. Theorem 6.3.14 in \cite{S99} shows that
\begin{align*}
\|h_\sigma-h\|_2
&=\left(\int_0^1 (h_\sigma(x)-h(x))^2 d x\right)^{1/2}\\
&\leq \left(\int_{-\infty}^{\infty} (h_\sigma(x)-h(x))^2 d x\right)^{1/2}\to 0
\end{align*}
as $\sigma\to 0$. The function $h_\sigma$ is continuously differentiable but it is not necessarily increasing, and so we need to further approximate it by an increasing continuously differentiable function. However, integration by parts yields for all $x\in[0,1]$,
\begin{align*}
D h_\sigma(x)
&= - \frac{1}{\sigma^2}\int_0^1h(y)D\phi\left(\frac{y-x}{\sigma}\right)d y\\
&= - \frac{1}{\sigma}\left(h(1)\phi\left(\frac{1-x}{\sigma}\right)-h(0)\phi\left(\frac{-x}{\sigma}\right)-\int_0^1 \phi\left(\frac{y-x}{\sigma}\right)d h(y)\right)\\
&\geq -\frac{1}{\sigma}\phi\left(\frac{1-x}{\sigma}\right)
\end{align*}
since $h(0)=0$, $h(1)=1$, and $\int_0^1\phi((y-x)\sigma)d h(y)\geq 0$ by $h$ being increasing. 
Therefore, the function
$$
h_{\sigma,\bar{x}}(x)=\begin{cases}
h_{\sigma}(x)+(x/\sigma)\phi((1-\bar{x})/\sigma), & \text{for }x\in[0,\overline{x}]\\
h_\sigma(\bar{x})+(\bar{x}/\sigma)\phi((1-\bar{x})/\sigma), & \text{for }x\in(\overline{x},1]
\end{cases}
$$
defined for all $x\in[0,1]$ and some $\bar{x}\in(0,1)$ is increasing and continuously differentiable for all $x\in(0,1)\backslash\bar{x}$, where it has a kink. Also, setting $\bar{x}=\bar{x}_\sigma=1-\sqrt{\sigma}$ and observing that $0\leq h_\sigma(x)\leq 1$ for all $x\in[0,1]$, we obtain
\begin{align*}
\|h_{\sigma,\overline{x}_\sigma}-h_\sigma\|_2\leq &
\frac{1}{\sigma}\phi\left(\frac{1}{\sqrt{\sigma}}\right)\left(\int_0^{1-\sqrt{\sigma}}d x\right)^{1/2}+\left(1+\frac{1}{\sigma}\phi\left(\frac{1}{\sqrt{\sigma}}\right)\right)\left(\int_{1-\sqrt{\sigma}}^1 d x\right)^{1/2}\to 0
\end{align*}
as $\sigma\to 0$ because $\sigma^{-1}\phi(\sigma^{-1/2})\to 0$. Smoothing the kink of $h_{\sigma,\bar{x}_\sigma}$ and using the triangle inequality, we obtain the asserted claim. This completes the proof of the lemma.
\end{proof}

\begin{lemma}\label{lem: matrix LLN}
Let $(p_1',q_1')',\dots,(p_n',q_n')'$ be a sequence of i.i.d. random vectors where $p_i$'s are vectors in $\mathbb{R}^K$ and $q_i$'s are vectors in $\mathbb{R}^J$. Assume that $\|p_1\|\leq \xi_{n}$, $\|q_1\|\leq \xi_{n}$, $\|\Ep[p_1 p_1']\|\leq C_p$, and $\|\Ep[q_1 q_1']\|\leq C_q$ where $\xi_n\geq 1$. Then for all $t\geq 0$,
$$
\Pr\left(\|\Ep_n[p_i q_i']-\Ep[p_1 q_1']\|\geq t\right)\leq \exp\left(\log(K+J)-\frac{A n t^2}{\xi_n^2(1+t)}\right)
$$
where $A>0$ is a constant depending only on $C_p$ and $C_q$.
\end{lemma}
\begin{remark}
Closely related results have been used previously by \cite{Belloni:2014hl} and \cite{Chen:2013fk}.
\end{remark}
\begin{proof}
The proof follows from Corollary 6.2.1 in \cite{T12}. Below we perform some auxiliary calculations. For any $a\in\mathbb{R}^K$ and $b\in\mathbb{R}^J$,
\begin{align*}
a' \Ep[p_1q_1']b
&=\Ep[(a' p_1)(b' q_1)]\\
&\leq \left(\Ep[(a' p_1)^2]\Ep[(b' q_1)^2]\right)^{1/2}\leq \|a\|\|b\|(C_p C_q)^{1/2}
\end{align*} 
by H\"{o}lder's inequality. Therefore, $\|\Ep[p_1 q_1']\|\leq (C_p C_q)^{1/2}$. Further, denote $S_i:=p_i q_i'-\Ep[p_i q_i']$ for $i=1,\dots, n$. By the triangle inequality and calculations above,
\begin{align*}
\|S_1\|
&\leq \|p_1 q_1'\|+\|\Ep[p_1 q_1']\|\\
&\leq \xi_n^2+(C_p C_q)^{1/2}\leq \xi_n^2(1+(C_p C_q)^{1/2})=:R.
\end{align*}
Now, denote $Z_n:=\sum_{i=1}^n S_i$. Then
\begin{align*}
\|\Ep[Z_n Z_n']\|
&\leq n\|\Ep[S_1 S_1']\|\\
&\leq n\|\Ep[p_1q_1'q_1p_1']\|+n\|\Ep[p_1q_1']\Ep[q_1p_1']\|
\leq n\|\Ep[p_1q_1'q_1p_1']\|+n C_p C_q.
\end{align*}
For any $a\in\mathbb{R}^K$,
$$
a' \Ep[p_1q_1'q_1p_1']a\leq \xi_n^2\Ep[(a' p_1)^2]\leq \xi_n^2\|a\|^2C_p.
$$
Therefore, $\|\Ep[p_1q_1'q_1p_1']\|\leq \xi_n^2 C_p$, and so
$$
\|\Ep[Z_n Z_n']\|\leq n C_p(\xi_n^2+C_q)\leq n\xi_n^2(1+C_p)(1+C_q).
$$
Similarly, $\|\Ep[Z_n' Z_n]\|\leq n\xi_n^2(1+C_p)(1+C_q)$, and so
$$
\sigma^2:=\max(\|\Ep[Z_n Z_n']\|,\|\Ep[Z_n' Z_n]\|)\leq n\xi_n^2(1+C_p)(1+C_q).
$$
Hence, by Corollary 6.2.1 in \cite{T12},
\begin{align*}
\Pr\left(\|n^{-1}Z_n\|\geq t\right)
&\leq (K+J)\exp\left(-\frac{n^2t^2/2}{\sigma^2+2nRt/3}\right)\\
&\leq \exp\left(\log(K+J)-\frac{Ant^2}{\xi_n^2(1+t)}\right).
\end{align*}
This completes the proof of the lemma.
\end{proof}

\bibliographystyle{economet}
\bibliography{ref}

\clearpage
\begin{figure}[!tp]
	\centering
	\includegraphics[width=\textwidth]{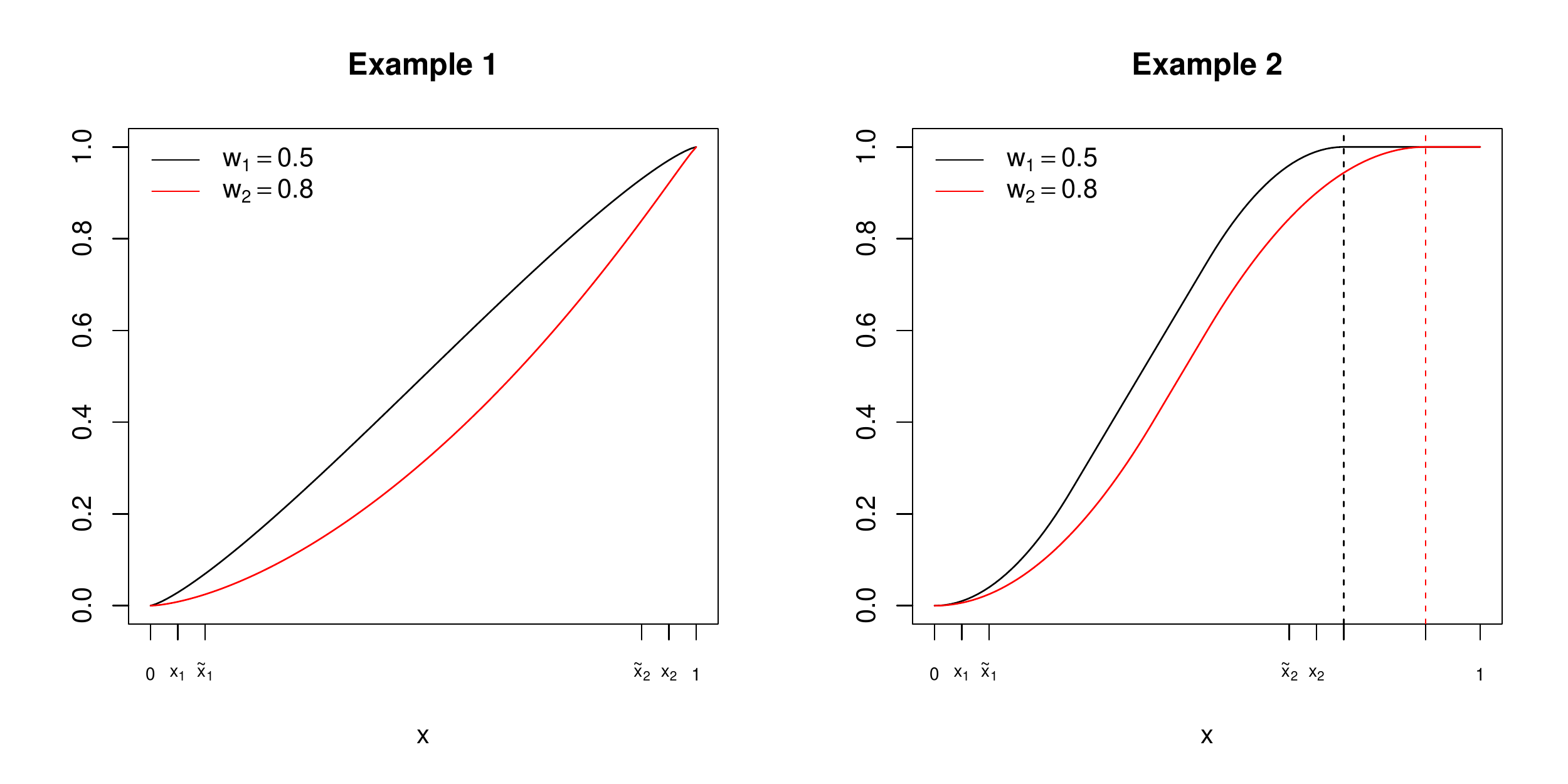}
	\caption{Plots of $F_{X|W}(x|w_1)$ and $F_{X|W}(x|w_2)$ in Examples~\ref{ex: normal density} and \ref{ex: uniforms}, respectively.}
	\label{fig:cdfXW}
\end{figure}

\begin{table}[htp]\scriptsize
\begin{center}\begin{tabular}{rrrlrrrrrrrrrr}
\hline\hline\\
& & & && \multicolumn{8}{c}{Model 1}\\\cline{6-13}
& & & && \multicolumn{2}{c}{$\kappa=1$} & & \multicolumn{2}{c}{$\kappa=0.5$} && \multicolumn{2}{c}{$\kappa=0.1$} \\\cline{6-7} \cline{9-10} \cline{12-13}
$\sigma_{\varepsilon}$ & $k_X$ & $k_W$ & && uncon. & con. && uncon. & con. && uncon. & con.  \\\hline

$0.1$ &  2  & 3 & bias sq.  & & 0.000                     & 0.001 &                           & 0.000 & 0.000                     &  & 0.000 & 0.000 \\
      &     &   & var       & & 0.021                     & 0.007 &                           & 0.005 & 0.002                     &  & 0.000 & 0.000 \\
      &     &   & MSE       & & 0.021                     & 0.009 &                           & 0.005 & 0.002                     &  & 0.000 & 0.000 \\
      &     &   & MSE ratio & & \multicolumn{2}{r}{0.406} &       & \multicolumn{2}{r}{0.409} &       & \multicolumn{2}{r}{0.347} \\

      \\ 		
      &  2  & 5 & bias sq.  & & 0.000                     & 0.001 &                           & 0.000 & 0.000                     &  & 0.000 & 0.000 \\
      &     &   & var       & & 0.009                     & 0.004 &                           & 0.002 & 0.001                     &  & 0.000 & 0.000 \\
      &     &   & MSE       & & 0.009                     & 0.005 &                           & 0.002 & 0.001                     &  & 0.000 & 0.000 \\
      &     &   & MSE ratio & & \multicolumn{2}{r}{0.529} &       & \multicolumn{2}{r}{0.510} &       & \multicolumn{2}{r}{0.542} \\

      \\ 		
      &  3  & 4 & bias sq.  & & 0.000                     & 0.001 &                           & 0.000 & 0.000                     &  & 0.000 & 0.000 \\
      &     &   & var       & & 0.026                     & 0.009 &                           & 0.005 & 0.002                     &  & 0.000 & 0.000 \\
      &     &   & MSE       & & 0.026                     & 0.009 &                           & 0.005 & 0.002                     &  & 0.000 & 0.000 \\
      &     &   & MSE ratio & & \multicolumn{2}{r}{0.355} &       & \multicolumn{2}{r}{0.412} &       & \multicolumn{2}{r}{0.372} \\
		  
      \\ 		
      &  3  & 7 & bias sq.  & & 0.000                     & 0.000 &                           & 0.000 & 0.000                     &  & 0.000 & 0.000 \\
      &     &   & var       & & 0.013                     & 0.005 &                           & 0.003 & 0.001                     &  & 0.000 & 0.000 \\
      &     &   & MSE       & & 0.013                     & 0.005 &                           & 0.003 & 0.001                     &  & 0.000 & 0.000 \\
      &     &   & MSE ratio & & \multicolumn{2}{r}{0.405} &       & \multicolumn{2}{r}{0.486} &       & \multicolumn{2}{r}{0.605} \\
		  
      \\ 		
      &  5  & 8 & bias sq.  & & 0.000                     & 0.000 &                           & 0.000 & 0.000                     &  & 0.000 & 0.000 \\
      &     &   & var       & & 0.027                     & 0.007 &                           & 0.005 & 0.002                     &  & 0.000 & 0.000 \\
      &     &   & MSE       & & 0.027                     & 0.007 &                           & 0.005 & 0.002                     &  & 0.000 & 0.000 \\
      &     &   & MSE ratio & & \multicolumn{2}{r}{0.266} &       & \multicolumn{2}{r}{0.339} &       & \multicolumn{2}{r}{0.411} \\

      \\
$0.7$ &  2  & 3 & bias sq.  & & 0.001                     & 0.020 &                           & 0.000 & 0.005                     &  & 0.000 & 0.000 \\
      &     &   & var       & & 0.857                     & 0.097 &                           & 0.263 & 0.024                     &  & 0.012 & 0.001 \\
      &     &   & MSE       & & 0.857                     & 0.118 &                           & 0.263 & 0.029                     &  & 0.012 & 0.001 \\
      &     &   & MSE ratio & & \multicolumn{2}{r}{0.137} &       & \multicolumn{2}{r}{0.110} &       & \multicolumn{2}{r}{0.101} \\

      \\ 		
      &  2  & 5 & bias sq.  & & 0.001                     & 0.015 &                           & 0.000 & 0.004                     &  & 0.000 & 0.000 \\
      &     &   & var       & & 0.419                     & 0.080 &                           & 0.102 & 0.020                     &  & 0.004 & 0.001 \\
      &     &   & MSE       & & 0.420                     & 0.095 &                           & 0.102 & 0.024                     &  & 0.004 & 0.001 \\
      &     &   & MSE ratio & & \multicolumn{2}{r}{0.227} &       & \multicolumn{2}{r}{0.235} &       & \multicolumn{2}{r}{0.221} \\

      \\ 		
      &  3  & 4 & bias sq.  & & 0.001                     & 0.016 &                           & 0.000 & 0.004                     &  & 0.000 & 0.000 \\
      &     &   & var       & & 0.763                     & 0.104 &                           & 0.223 & 0.026                     &  & 0.010 & 0.001 \\
      &     &   & MSE       & & 0.763                     & 0.121 &                           & 0.223 & 0.030                     &  & 0.010 & 0.001 \\
      &     &   & MSE ratio & & \multicolumn{2}{r}{0.158} &       & \multicolumn{2}{r}{0.133} &       & \multicolumn{2}{r}{0.119} \\
		  
      \\ 		
      &  3  & 7 & bias sq.  & & 0.001                     & 0.011 &                           & 0.000 & 0.003                     &  & 0.000 & 0.000 \\
      &     &   & var       & & 0.350                     & 0.083 &                           & 0.104 & 0.020                     &  & 0.004 & 0.001 \\
      &     &   & MSE       & & 0.351                     & 0.094 &                           & 0.104 & 0.023                     &  & 0.004 & 0.001 \\
      &     &   & MSE ratio & & \multicolumn{2}{r}{0.267} &       & \multicolumn{2}{r}{0.218} &       & \multicolumn{2}{r}{0.229} \\
		  
      \\ 		
      &  5  & 8 & bias sq.  & & 0.001                     & 0.011 &                           & 0.000 & 0.003                     &  & 0.000 & 0.000 \\
      &     &   & var       & & 0.433                     & 0.094 &                           & 0.131 & 0.023                     &  & 0.006 & 0.001 \\
      &     &   & MSE       & & 0.434                     & 0.105 &                           & 0.131 & 0.025                     &  & 0.006 & 0.001 \\
      &     &   & MSE ratio & & \multicolumn{2}{r}{0.243} &       & \multicolumn{2}{r}{0.193} &       & \multicolumn{2}{r}{0.170} \\

\hline\hline
\end{tabular} \caption{Model 1: Performance of the unconstrained and constrained estimators for $N=500$, $\rho=0.3$, $\eta=0.3$.}
\label{tab:model1n500}
\end{center}
\end{table}

\begin{table}[htp]\scriptsize
\begin{center}\begin{tabular}{rrrlrrrrrrrrrr}
\hline\hline\\
& & & && \multicolumn{8}{c}{Model 2}\\\cline{6-13}
& & & && \multicolumn{2}{c}{$\kappa=1$} & & \multicolumn{2}{c}{$\kappa=0.5$} && \multicolumn{2}{c}{$\kappa=0.1$} \\\cline{6-7} \cline{9-10} \cline{12-13}
$\sigma_{\varepsilon}$ & $k_X$ & $k_W$ & && uncon. & con. && uncon. & con. && uncon. & con.  \\\hline

$0.1$ &  2  & 3 & bias sq.  & & 0.001                     & 0.002 &                           & 0.000 & 0.001                     &  & 0.000 & 0.000 \\
      &     &   & var       & & 0.024                     & 0.003 &                           & 0.007 & 0.001                     &  & 0.000 & 0.000 \\
      &     &   & MSE       & & 0.024                     & 0.006 &                           & 0.007 & 0.001                     &  & 0.000 & 0.000 \\
      &     &   & MSE ratio & & \multicolumn{2}{r}{0.229} &       & \multicolumn{2}{r}{0.201} &       & \multicolumn{2}{r}{0.222} \\

      \\ 		
      &  2  & 5 & bias sq.  & & 0.001                     & 0.002 &                           & 0.000 & 0.000                     &  & 0.000 & 0.000 \\
      &     &   & var       & & 0.010                     & 0.002 &                           & 0.002 & 0.001                     &  & 0.000 & 0.000 \\
      &     &   & MSE       & & 0.011                     & 0.004 &                           & 0.002 & 0.001                     &  & 0.000 & 0.000 \\
      &     &   & MSE ratio & & \multicolumn{2}{r}{0.405} &       & \multicolumn{2}{r}{0.475} &       & \multicolumn{2}{r}{0.446} \\

      \\ 		
      &  3  & 4 & bias sq.  & & 0.000                     & 0.001 &                           & 0.000 & 0.000                     &  & 0.000 & 0.000 \\
      &     &   & var       & & 0.022                     & 0.003 &                           & 0.006 & 0.001                     &  & 0.000 & 0.000 \\
      &     &   & MSE       & & 0.022                     & 0.004 &                           & 0.006 & 0.001                     &  & 0.000 & 0.000 \\
      &     &   & MSE ratio & & \multicolumn{2}{r}{0.192} &       & \multicolumn{2}{r}{0.176} &       & \multicolumn{2}{r}{0.157} \\
		  
      \\ 		
      &  3  & 7 & bias sq.  & & 0.000                     & 0.001 &                           & 0.000 & 0.000                     &  & 0.000 & 0.000 \\
      &     &   & var       & & 0.009                     & 0.002 &                           & 0.002 & 0.001                     &  & 0.000 & 0.000 \\
      &     &   & MSE       & & 0.009                     & 0.003 &                           & 0.002 & 0.001                     &  & 0.000 & 0.000 \\
      &     &   & MSE ratio & & \multicolumn{2}{r}{0.325} &       & \multicolumn{2}{r}{0.292} &       & \multicolumn{2}{r}{0.323} \\
		  
      \\ 		
      &  5  & 8 & bias sq.  & & 0.000                     & 0.001 &                           & 0.000 & 0.000                     &  & 0.000 & 0.000 \\
      &     &   & var       & & 0.014                     & 0.003 &                           & 0.003 & 0.001                     &  & 0.000 & 0.000 \\
      &     &   & MSE       & & 0.014                     & 0.004 &                           & 0.003 & 0.001                     &  & 0.000 & 0.000 \\
      &     &   & MSE ratio & & \multicolumn{2}{r}{0.269} &       & \multicolumn{2}{r}{0.268} &       & \multicolumn{2}{r}{0.217} \\

\\
$0.7$ &  2  & 3 & bias sq.  & & 0.002                     & 0.005 &                           & 0.001 & 0.001                     &  & 0.000 & 0.000 \\
      &     &   & var       & & 1.102                     & 0.032 &                           & 0.321 & 0.008                     &  & 0.012 & 0.000 \\
      &     &   & MSE       & & 1.104                     & 0.038 &                           & 0.321 & 0.009                     &  & 0.012 & 0.000 \\
      &     &   & MSE ratio & & \multicolumn{2}{r}{0.034} &       & \multicolumn{2}{r}{0.029} &       & \multicolumn{2}{r}{0.032} \\

      \\ 		
      &  2  & 5 & bias sq.  & & 0.001                     & 0.006 &                           & 0.000 & 0.002                     &  & 0.000 & 0.000 \\
      &     &   & var       & & 0.462                     & 0.031 &                           & 0.103 & 0.008                     &  & 0.004 & 0.000 \\
      &     &   & MSE       & & 0.463                     & 0.037 &                           & 0.104 & 0.009                     &  & 0.004 & 0.000 \\
      &     &   & MSE ratio & & \multicolumn{2}{r}{0.080} &       & \multicolumn{2}{r}{0.088} &       & \multicolumn{2}{r}{0.088} \\

      \\ 		
      &  3  & 4 & bias sq.  & & 0.001                     & 0.004 &                           & 0.000 & 0.001                     &  & 0.000 & 0.000 \\
      &     &   & var       & & 0.936                     & 0.036 &                           & 0.255 & 0.009                     &  & 0.012 & 0.000 \\
      &     &   & MSE       & & 0.936                     & 0.040 &                           & 0.255 & 0.010                     &  & 0.012 & 0.000 \\
      &     &   & MSE ratio & & \multicolumn{2}{r}{0.043} &       & \multicolumn{2}{r}{0.039} &       & \multicolumn{2}{r}{0.034} \\
		  
      \\ 		
      &  3  & 7 & bias sq.  & & 0.001                     & 0.005 &                           & 0.000 & 0.001                     &  & 0.000 & 0.000 \\
      &     &   & var       & & 0.387                     & 0.035 &                           & 0.110 & 0.009                     &  & 0.004 & 0.000 \\
      &     &   & MSE       & & 0.388                     & 0.040 &                           & 0.110 & 0.010                     &  & 0.004 & 0.000 \\
      &     &   & MSE ratio & & \multicolumn{2}{r}{0.103} &       & \multicolumn{2}{r}{0.089} &       & \multicolumn{2}{r}{0.092} \\
		  
      \\ 		
      &  5  & 8 & bias sq.  & & 0.002                     & 0.005 &                           & 0.000 & 0.001                     &  & 0.000 & 0.000 \\
      &     &   & var       & & 0.508                     & 0.041 &                           & 0.144 & 0.010                     &  & 0.007 & 0.000 \\
      &     &   & MSE       & & 0.510                     & 0.046 &                           & 0.144 & 0.011                     &  & 0.007 & 0.000 \\
      &     &   & MSE ratio & & \multicolumn{2}{r}{0.090} &       & \multicolumn{2}{r}{0.078} &       & \multicolumn{2}{r}{0.065} \\

\hline\hline
\end{tabular} \caption{Model 2: Performance of the unconstrained and constrained estimators for $N=500$, $\rho=0.3$, $\eta=0.3$.}
\label{tab:model2n500}
\end{center}
\end{table}

\begin{table}[htp]
\begin{center}\begin{tabular}{rrlrrrrrrrrrr}
\hline\hline\\
& & && \multicolumn{8}{c}{Model 1}\\\cline{5-12}
& & && \multicolumn{2}{c}{$\kappa=1$} & & \multicolumn{2}{c}{$\kappa=0.5$} && \multicolumn{2}{c}{$\kappa=0.1$} \\\cline{5-6} \cline{8-9} \cline{11-12}
$\rho$ & $\eta$ & && uncon. & con. && uncon. & con. && uncon. & con.  \\\hline

0.3 &  0.3 & bias sq.  & & 0.000                     & 0.001 &                           & 0.000 & 0.000                     &  & 0.000 & 0.000 \\
    &      & var       & & 0.026                     & 0.009 &                           & 0.005 & 0.002                     &  & 0.000 & 0.000 \\
    &      & MSE       & & 0.026                     & 0.009 &                           & 0.005 & 0.002                     &  & 0.000 & 0.000 \\
    &      & MSE ratio & & \multicolumn{2}{r}{0.355} &       & \multicolumn{2}{r}{0.412} &       & \multicolumn{2}{r}{0.372} \\

    \\ 		
0.3 &  0.7 & bias sq.  & & 0.000                     & 0.001 &                           & 0.000 & 0.000                     &  & 0.000 & 0.000 \\
    &      & var       & & 0.026                     & 0.008 &                           & 0.005 & 0.002                     &  & 0.000 & 0.000 \\
    &      & MSE       & & 0.026                     & 0.009 &                           & 0.005 & 0.002                     &  & 0.000 & 0.000 \\
    &      & MSE ratio & & \multicolumn{2}{r}{0.342} &       & \multicolumn{2}{r}{0.395} &       & \multicolumn{2}{r}{0.449} \\

    \\ 		
0.7 &  0.3 & bias sq.  & & 0.000                     & 0.001 &                           & 0.000 & 0.000                     &  & 0.000 & 0.000 \\
    &      & var       & & 0.025                     & 0.002 &                           & 0.003 & 0.001                     &  & 0.000 & 0.000 \\
    &      & MSE       & & 0.025                     & 0.003 &                           & 0.003 & 0.001                     &  & 0.000 & 0.000 \\
    &      & MSE ratio & & \multicolumn{2}{r}{0.125} &       & \multicolumn{2}{r}{0.248} &       & \multicolumn{2}{r}{0.266} \\
		  
    \\ 		
0.7 &  0.7 & bias sq.  & & 0.000                     & 0.001 &                           & 0.000 & 0.000                     &  & 0.000 & 0.000 \\
    &      & var       & & 0.023                     & 0.002 &                           & 0.004 & 0.001                     &  & 0.000 & 0.000 \\
    &      & MSE       & & 0.023                     & 0.003 &                           & 0.004 & 0.001                     &  & 0.000 & 0.000 \\
    &      & MSE ratio & & \multicolumn{2}{r}{0.136} &       & \multicolumn{2}{r}{0.212} &       & \multicolumn{2}{r}{0.259} \\
\hline\hline
\end{tabular} \caption{Model 1: Performance of the unconstrained and constrained estimators for $\sigma_{\varepsilon}=0.1$, $k_X=3$, $k_W=4$, $N=500$.}
\label{tab:re_model1}
\end{center}
\end{table}

\begin{table}[htp]
\begin{center}\begin{tabular}{rrlrrrrrrrrrr}
\hline\hline\\
& & && \multicolumn{8}{c}{Model 2}\\\cline{5-12}
& & && \multicolumn{2}{c}{$\kappa=1$} & & \multicolumn{2}{c}{$\kappa=0.5$} && \multicolumn{2}{c}{$\kappa=0.1$} \\\cline{5-6} \cline{8-9} \cline{11-12}
$\rho$ & $\eta$ & && uncon. & con. && uncon. & con. && uncon. & con.  \\\hline

0.3 & 0.3 & bias sq.  & & 0.000                     & 0.001 &                           & 0.000 & 0.000                          & & 0.000 & 0.000 \\
    &     & var       & & 0.022                     & 0.003 &                           & 0.006 & 0.001                          & & 0.000 & 0.000 \\
    &     & MSE       & & 0.022                     & 0.004 &                           & 0.006 & 0.001                          & & 0.000 & 0.000 \\
    &     & MSE ratio & & \multicolumn{2}{r}{0.192} &       & \multicolumn{2}{r}{0.176} &       & \multicolumn{2}{r}{0.157} \\

\\		   
0.3 & 0.7 & bias sq.  & & 0.000                     & 0.001 &                           & 0.000 & 0.000                          & & 0.000 & 0.000 \\
    &     & var       & & 0.020                     & 0.003 &                           & 0.006 & 0.001                          & & 0.000 & 0.000 \\
    &     & MSE       & & 0.020                     & 0.004 &                           & 0.006 & 0.001                          & & 0.000 & 0.000 \\
    &     & MSE ratio & & \multicolumn{2}{r}{0.209} &       & \multicolumn{2}{r}{0.163} &       & \multicolumn{2}{r}{0.160} \\

\\		   
0.7 & 0.3 & bias sq.  & & 0.000                     & 0.000 &                           & 0.000 & 0.000                          & & 0.000 & 0.000 \\
    &     & var       & & 0.013                     & 0.000 &                           & 0.002 & 0.000                          & & 0.000 & 0.000 \\
    &     & MSE       & & 0.013                     & 0.001 &                           & 0.002 & 0.000                          & & 0.000 & 0.000 \\
    &     & MSE ratio & & \multicolumn{2}{r}{0.040} &       & \multicolumn{2}{r}{0.063} &       & \multicolumn{2}{r}{0.047} \\
		  
\\		  
0.7 & 0.7 & bias sq.  & & 0.000                     & 0.000 &                           & 0.000 & 0.000                          & & 0.000 & 0.000   \\
    &     & var       & & 0.010                     & 0.000 &                           & 0.002 & 0.000                          & & 0.000 & 0.000   \\
    &     & MSE       & & 0.011                     & 0.001 &                           & 0.002 & 0.000                          & & 0.000 & 0.000   \\
    &     & MSE ratio & & \multicolumn{2}{r}{0.051} &       & \multicolumn{2}{r}{0.060} &       & \multicolumn{2}{r}{0.050}   \\
\hline\hline
\end{tabular} \caption{Model 2: Performance of the unconstrained and constrained estimators for $\sigma_{\varepsilon}=0.1$, $k_X=3$, $k_W=4$, $N=500$.}
\label{tab:re_model2}
\end{center}
\end{table}

\clearpage
\begin{figure}[!tp]
	\centering
	\includegraphics[width=\textwidth]{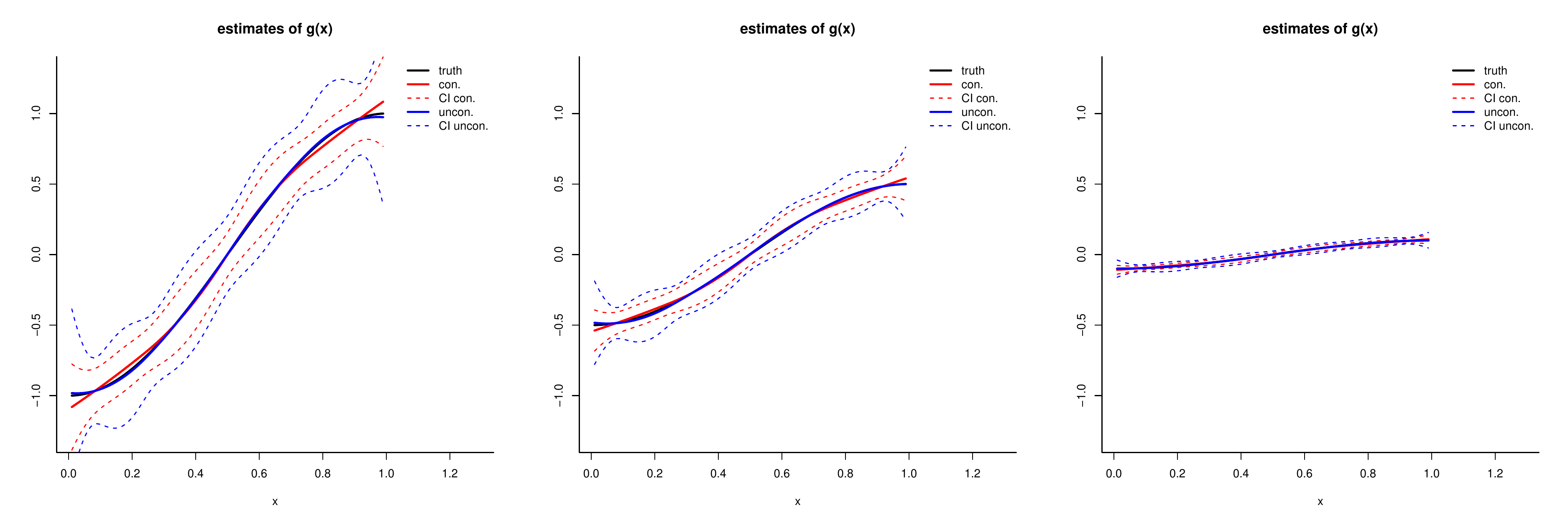}
	\caption{Model 1: unconstrained and constrained estimates of $g(x)$ for $N=500$, $\rho=0.3$, $\eta=0.3$, $\sigma_{\varepsilon}=0.1$, $k_X=3$, $k_W=4$.}
	\label{fig:estim_model1}
\end{figure}

\begin{figure}[!tp]
	\centering
	\includegraphics[width=\textwidth]{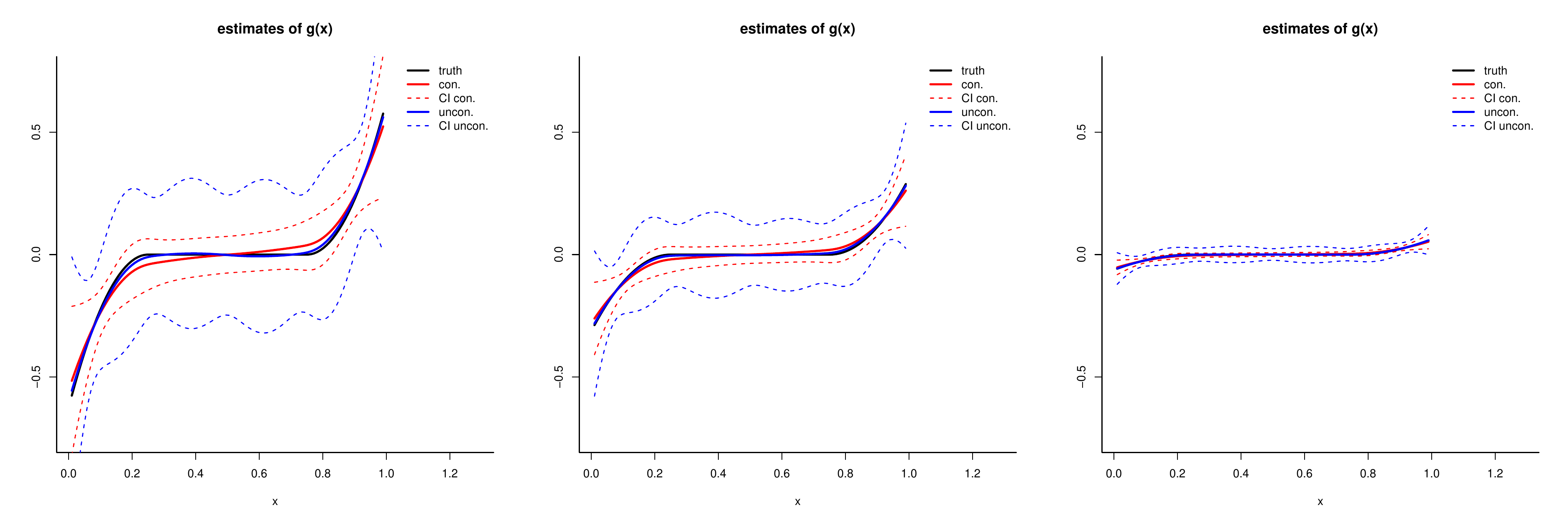}
	\caption{Model 2: unconstrained and constrained estimates of $g(x)$ for $N=500$, $\rho=0.3$, $\eta=0.3$, $\sigma_{\varepsilon}=0.1$, $k_X=3$, $k_W=4$.}
	\label{fig:estim_model2}
\end{figure}

\clearpage

\begin{figure}[!tp]
	\centering
	\includegraphics[width=0.8\textwidth]{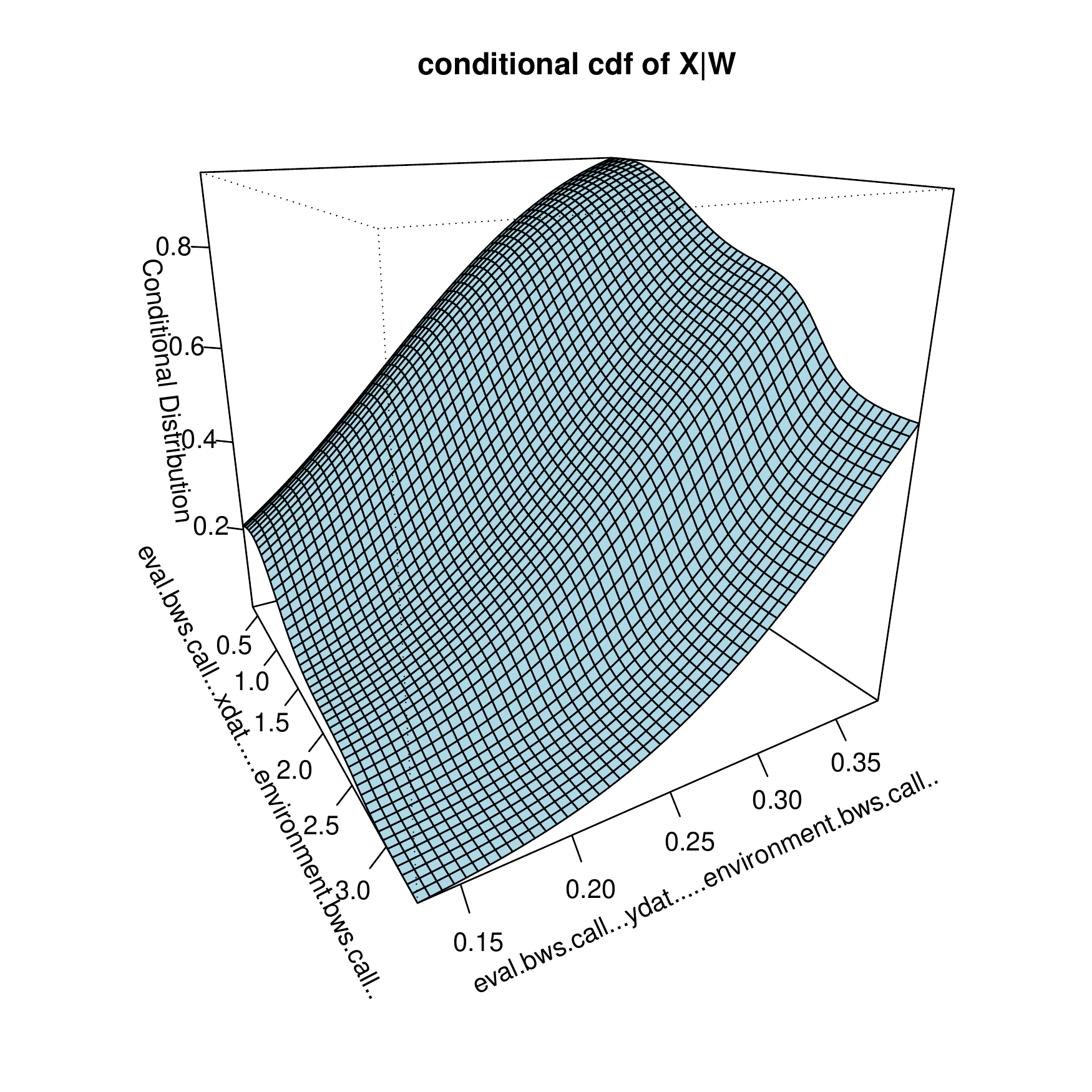}
	\caption{Nonparametric kernel estimate of the conditional cdf $F_{X|W}(x|w)$.}
	\label{fig:cdfXWhat}
\end{figure}


\begin{figure}[!tp]
	\centering
	\includegraphics[width=\textwidth]{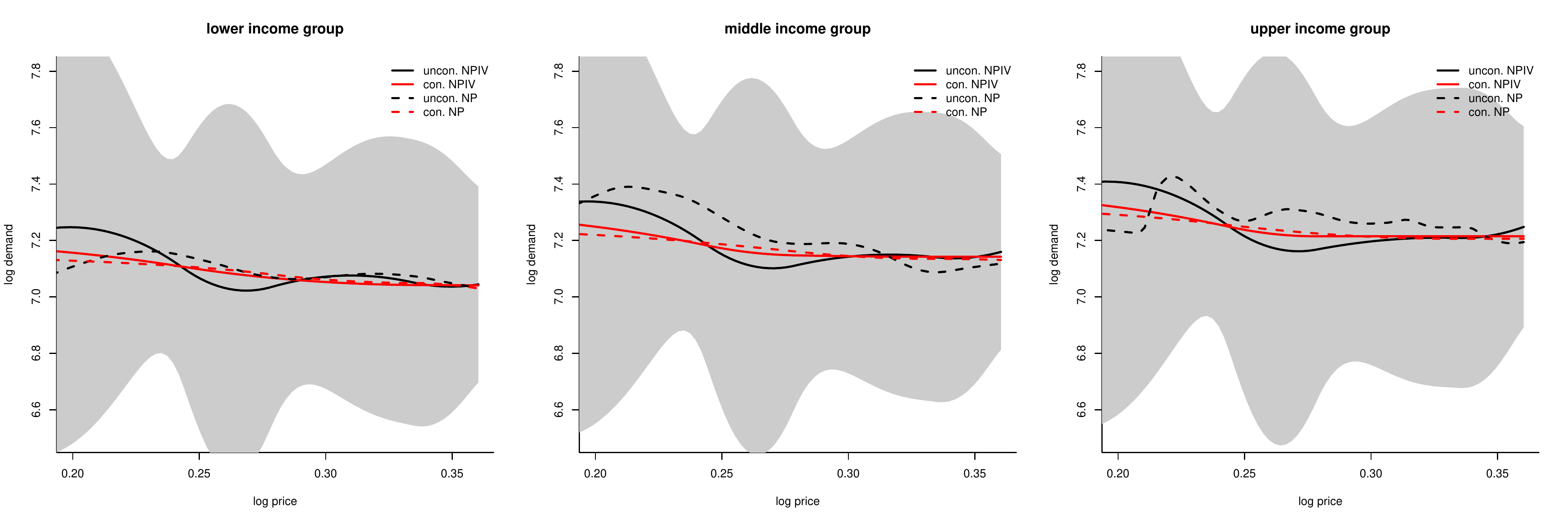}
	\caption{Estimates of $g(x,z_1)$ plotted as a function of price $x$ for $z_1$ fixed at three income levels.}
	\label{fig:emp with cov}
\end{figure}

\end{document}